\documentclass[draft]{IEEEtran}
\onecolumn
\usepackage[english]{babel}
\usepackage{amsmath,amsthm}
\usepackage{amssymb,mathrsfs}
\usepackage{amsfonts}
\usepackage[final,hyperfootnotes=false]{hyperref}
\usepackage[final]{graphicx}
\usepackage{epstopdf}
\usepackage{tikz}
\usepackage{pgfplots}
\usepackage{subfig}
\usepackage{flushend}
\usepackage{cprotect}

\usepackage{tikz}
\usetikzlibrary{arrows,positioning,fit,backgrounds}
\usetikzlibrary{calc}

\usetikzlibrary{arrows,positioning,fit,backgrounds,shapes}
\usetikzlibrary{calc}
\newtheorem{thm}{Theorem}
\newtheorem{cor}[thm]{Corollary}
\newtheorem{lem}[thm]{Lemma}
\newtheorem{prop}[thm]{Proposition}

\theoremstyle{remark}
\newtheorem{rem}{Remark}
\newtheorem{conv}{Convention}
\newtheorem{defn}{Definition}
\newtheorem{notat}{Notation}

\newcommand{\mb}{\mathbf}
\DeclareMathOperator*{\Cov}{Cov}
\DeclareMathOperator*{\Tr}{Tr}

\newcommand{\EE}{\mathbb{E}}
\newcommand{\uh}{\hat{U}}
\newcommand{\bsigma}{\mb{\Sigma}}
\DeclareMathOperator*{\BL}{\mathcal{G}_{BL}}
\newcommand{\1}{^{(1)}}
\newcommand{\2}{^{(2)}}
\newcommand{\n}{^{(n)}}
\newcommand{\half}{{\frac{1}{2}}}

\usepackage{xcolor}
\hypersetup{
    colorlinks,
    linkcolor={blue!80!black},
    citecolor={blue!80!black},
    urlcolor={blue!80!black}
}

\begin{document}
\title{
Information-Theoretic
Perspectives on Brascamp-Lieb Inequality and Its Reverse
}%
\author{Jingbo Liu~~~~~~~~~
Thomas A. Courtade~~~~~~~~~
Paul Cuff~~~~~~~~~
Sergio Verd\'{u}
}%

\maketitle
\newcommand\blfootnote[1]{%
  \begingroup
  \renewcommand\thefootnote{}\footnote{#1}%
  \addtocounter{footnote}{-1}%
  \endgroup
}
\blfootnote{Jingbo Liu, Paul Cuff and Sergio Verd\'{u} are with the Department of Electrical Engineering, Princeton University. Emails: \{{jingbo,cuff,verdu}@princeton.edu\}.
Thomas A.~Courtade is with the Department of Electrical Engineering and Computer Sciences,
University of California, Berkeley. Email: courtade@eecs.berkeley.edu.
This paper was presented in part at ISIT 2016.
}

\begin{abstract}
We introduce an inequality which may be viewed as a generalization of both the Brascamp-Lieb inequality and its reverse (Barthe's inequality), and prove its information-theoretic (i.e.\ entropic) formulation.
This result leads to a unified approach to
functional inequalities such as
the variational formula of R\'enyi entropy,
hypercontractivity and its reverse, strong data processing inequalities,
and transportation-cost inequalities,
whose utility in the proofs of various coding theorems has gained growing popularity recently.
We show that our information-theoretic
setting is convenient for proving properties such as data processing, tensorization, convexity
(Riesz-Thorin interpolation) and Gaussian optimality. In particular, we elaborate on a ``doubling trick''
used by Lieb and Geng-Nair to prove several results on Gaussian optimality.
Several applications are discussed, including
a generalization of the Brascamp-Lieb inequality involving Gaussian random transformations,
the determination of Wyner's  common information
of vector Gaussian sources,
and the achievable rate region of certain key generation problems in the case of vector Gaussian sources.
%
\end{abstract}

\section{Introduction}\label{sec1}
The Brascamp-Lieb inequality and its reverse \cite{brascamp1976best} concern the optimality of Gaussian functions in a certain type of integral inequality.\footnote{Not to be confused with the ``variance Brascamp-Lieb inequality'' (cf.~\cite{brascamp1976extensions}\cite{bobkov2000}\cite{cordero2015transport}),
which is a different type of inequality that generalizes the Poincar\'e inequality.}
They have been generalized in various ways over the nearly 40 years since their discovery.
To be concrete, let us take a look at a modern formulation of the result from Barthe's paper \cite{barthe1998reverse}:\footnote{
\cite[Theorem~1]{barthe1998reverse} actually contains additional assumptions, which make the best constants $D$ and $F$ positive and finite, but not really necessary for the conclusion to hold (\cite[Remark~1]{barthe1998reverse}).}
\begin{thm}[{\cite[Theorem~1]{barthe1998reverse}}]
Let $E$, $E_1$, \dots, $E_m$ be Euclidean spaces, and $\mb{B}_i\colon E\to E_i$ be linear maps. Let $(c_i)_{i=1}^m$ and $D$ be positive real numbers.
Then the \emph{Brascamp-Lieb} inequality
\begin{align}
\int \prod_{i=1}^m f_i(\mb{B}_i\mb{x})\,{\rm d}\mb{x}
\le
D\prod_{i=1}^m \|f_i\|_{\frac{1}{c_i}},
\label{e_bl}
\end{align}
for all $f_i\in L^{\frac{1}{c_i}}(E_i)$, $i=1,\dots,m$, holds if and only if it holds whenever $f_i$, $i=1,\dots,m$ are centered Gaussian functions\footnote{A centered Gaussian function is of the form $\mb{x}\mapsto \exp(-\mb{x}^{\top}\mb{A}\mb{x})$ where $\mb{A}$ is a positive semidefinite matrix.}.
Similarly, for $F$ a positive real number,
the \emph{reverse} Brascamp-Lieb inequality\footnote{$\mb{B}_i^*$ denotes the adjoint of $\mb{B}_i$. In other words, the matrix of $\mb{B}_i^*$ is the transpose of the matrix of $\mb{B}_i$.}
\begin{align}
\int \sup_{(\mb{y}_i)\colon \sum_{i=1}^mc_i\mb{B}_i^*
\mb{y}_i=\mb{x}}
\prod_{i=1}^mf_i(\mb{y}_i)
\,{\rm d}\mb{x}
\ge
F\prod_{i=1}^m \|f_i\|_{\frac{1}{c_i}},
\label{e_barthe}
\end{align}
for all nonnegative $f_i\in L^{\frac{1}{c_i}}(E_i)$, $i=1,\dots,m$, holds if and only if it holds for all centered Gaussian functions.
\end{thm}

Motivated by the problem of determining the exact constant in the sharp Young inequality, the original paper by Brascamp and Lieb \cite{brascamp1976best} considered \eqref{e_bl} with  one-dimensional linear projections $(\mb{B}_j)_{j=1}^m$, and showed that the inequality holds if and only if it holds for all real, centered Gaussian functions $(f_j)_{j=1}^m$.
Their proof is based on rearrangement inequalities and hinges on the fact that the linear projections are one-dimensional.
A reverse form of the sharp Young inequality, which is a special case of \eqref{e_barthe}, is also discussed in \cite{brascamp1976best}.

The Gaussian optimality result of \cite{brascamp1976best}
was later generalized by Lieb \cite{lieb1990gaussian} so that $(\mb{B}_j)_{j=1}^m$ in \eqref{e_bl} can be arbitrary surjective linear maps. Actually, Lieb's result \cite[Theorem~6.2]{lieb1990gaussian} covers complex functions and kernels,
including the important case of Fourier kernels.
We will only be concerned with real functions and kernels in the present paper.
Lieb's proof
used, among other things, a rotational invariance property of Gaussian random variables (also called ``doubling trick'' in \cite{carlen1991superadditivity}) which will also play a role in this paper.
For a result as fundamental as Lieb's theorem with far-reaching consequences,
alternative proof methods have received considerable attention,
including methods based on measure-preserving maps \cite{barthe1998optimal} \cite{barthe1998reverse}\footnote{See also Theorem~7 in the arXiv version of \cite{barthe1998reverse} for a beautiful result that contains \eqref{e_bl} and \eqref{e_barthe} as a limiting case (which, however, appears to be unrelated to our forward-reverse Brascamp-Lieb inequality in Section~\ref{sec_fr}).}, heat flow \cite{carlen2004sharp} \cite{carlen2009subadditivity} and the related semigroup argument \cite{barthe2010correlation}; see \cite[Remark~1.10]{bennett2008brascamp} for a brief account of the history. It is enlightening to summarize the properties of the Gaussian distribution which play a role in those proofs of Gaussian optimality:
\begin{itemize}
  \item \cite{brascamp1976best}  The tensor power of a one-dimensional Gaussian distribution is a multidimensional Gaussian distribution, which is stable under Schwarz symmetrization (i.e.~spherically decreasing rearrangement).
  \item \cite{lieb1990gaussian}  Rotational invariance: if $f$ is a one-dimensional Gaussian function, then
      \begin{align}
      f(x)f(y)
      =f\left(\frac{x-y}{\sqrt{2}}\right)
      f\left(\frac{x+y}{\sqrt{2}}\right).
      \end{align}
  \item \cite[Lemma~2]{barthe1998optimal} The convolution of Gaussian functions is Gaussian.
  \item \cite{carlen2009subadditivity} If a real valued random variable is added to an independent Gaussian noise, then the derivative of the differential entropy of the sum with respect to the variance of the noise is half the Fisher information (de Bruijn's identity), and of course the non-Gaussianness of the sum eventually disappears as the variance goes to infinity.
\end{itemize}

Inequality \eqref{e_bl} can be seen as a generalization of several other inequalities, including H\"{o}lder's inequality, the sharp Young inequality, the Loomis-Whitney inequality, the entropy power inequality (cf.~\cite{bennett2008brascamp} or the survey paper \cite{gardner2002brunn}), hypercontractivity and the logarithmic Sobolev inequality \cite{gross1975logarithmic}.
To see its connection to the sharp Young inequality, for example, consider \eqref{e_bl} with $\mb{B}_1,\mb{B}_2,\mb{B}_3$ being the following linear transforms from $\mathbb{R}^2$ to $\mathbb{R}$:
\begin{align}
\mb{B}_1&\colon (x_1,x_2)\mapsto x_1;
\\
\mb{B}_2&\colon (x_1,x_2)\mapsto x_1-x_2;
\\
\mb{B}_3&\colon (x_1,x_2)\mapsto x_2.
\end{align}
Then \eqref{e_bl} becomes an upper-bound on the inner product
\begin{align}
(f_1*f_2,f_3)&=\int\int f_1(x_1)f_2(x_1-x_2)f_3(x_2){\rm d}x_1{\rm d}x_2
\\
&\le D \|f_1\|_{p_1}\|f_2\|_{p_2}\|f_3\|_{p_3}
\end{align}
for appropriate values of $p_1,p_2,p_3$ and $D$, which is equivalent to the sharp Young inequality in view of the duality of the Banach spaces $L^{p_3}$ and $L^{q_3}$ where $q_3:=\frac{p_3}{p_3-1}$.
As observed by Dembo, Cover and Thomas \cite[Theorem~12]{DCT91},
the sharp Young inequality admits an equivalent formulation in terms of the R\'{e}nyi differential entropy, since the R\'{e}nyi differential entropy is (up to a factor) the logarithm of the norm of the density function of a random variable, and additions of vector-valued random variables translate to convolutions of their density functions.
As the orders of the R\'{e}nyi differential entropies converge to $1$, the well-known entropy power inequality is recovered.
Another information-theoretic implication of Lieb's result \cite[Theorem~6.2]{lieb1990gaussian}
is the Beckner-Hirschman inequality (also known as the entropic uncertainty principle, which strengthens the well-known Weyl-Heisenberg uncertainty principle).
This can be shown by specializing Lieb's theorem to the Fourier kernel to obtain the sharp Hausdorff-Young's inequality and then applying a differentiation argument \cite{beckner1975inequalities}.

A deeper and more general connection between \eqref{e_bl} and information measures was observed by Carlen and Cordero-Erausquin
\cite[Theorem~2.1]{carlen2009subadditivity} (see also \cite{carlen2004sharp} for a preliminary version on the sphere with $2$-norms).
By cleverly using the nonnegativity of relative entropy,
it is revealed that such a ``submultiplicativity'' of norms is equivalent to a superadditivity property of relative entropies with corresponding coefficients.
More precisely, \cite[Theorem~2.1]{carlen2009subadditivity}
states that \eqref{e_bl} is equivalent to
\begin{align}
-h(P_{\bf X})+\log D\ge -\sum_{j=1}^m c_j h(P_{{\bf Y}_j})
\end{align}
for all continuous probability measure $P_{\bf X}$ on $E$, where $h(\cdot)$ denotes the differential entropy, and $P_{{\bf Y}_j}$ is induced by the map $\mb{x}\mapsto\mb{B}_j\mb{x}$ (i.e.\ the push-forward).
This connection is less intuitive than the aforementioned connection between the Brascamp-Lieb inequality and R\'enyi differential entropy inequalities discussed in
\cite[Theorem~12]{DCT91}, in the sense that the functions in the functional inequalities cannot be interpreted as the probability densities in the corresponding information-theoretic inequality.
This connection is also very general, since in order for it to hold,
\begin{itemize}
\item The random variables in \eqref{e_bl} can be arbitrary rather than living on a space with additive structure.
\item The reference measures in \eqref{e_bl} need not be Gaussian or Lebesgue.
\item $(\mb{B}_j)_{j=1}^m$ may be replaced with arbitrary maps.
\end{itemize}
In \cite{ISIT_lccv2016} we referred to such an extension as a \emph{Brascamp-Lieb like inequality}\footnote{In the literature, e.g.~\cite{carlen2009subadditivity}, this has been referred to as ``Brascamp-Lieb type inequality''; we adopt a different name here to avoid possible connotations with the method of types in information theory.},
in order to distinguish it from the conventional notion of the Brascamp-Lieb inequality, which refers to the Gaussian optimality in \eqref{e_bl} in the case of Gaussian or Lebesgue measures and linear maps.
Using such a relation in conjunction with the superadditivity of Fisher's information,
Carlen and Cordero-Erausquin proposed a proof of the Brascamp-Lieb inequality with which the uniqueness of the extremizer in the inequality is simple to establish.
Similar connections between functional inequalities and information measures may be traced further back.
For example, in \cite[Theorem~5]{ahlswede1976spreading}
Ahlswede and G\'{a}cs proved an equivalent formulation of the \emph{strong data processing inequality} \cite[P45]{csiszar2011information} in terms of a functional inequality.
Indeed, we shall see that the results of Ahlswede-G\'{a}cs and Carlen--Cordero-Erausquin can in fact be subsumed in a common framework.
As for the reverse Brascamp-Lieb inequality, Lehec \cite[Theorem~18]{lehec2010representation} essentially proved in a special setting that it is implied by an entropic inequality but did not prove the converse implication (which, as we shall see, is the more nontrivial direction).
Due in part to their utility in establishing impossibility bounds,
these functional inequalities have attracted a lot of attention in information theory
\cite{erkip98}\cite{courtade2013outer}\cite{PW15}\cite{pw_2015}\cite{Liu}\cite{liu2015key}\cite{xu15}\cite{kamath2015non},
theoretical computer science
\cite{kahn1988influence}\cite{ganor2014exponential}\cite{dvir2014sylvester}\cite{braverman2015communication}\cite{garg2016algorithmic},
and statistics
\cite{talagrand1994russo}\cite{friedgut2002boolean}\cite{bourgain2002distribution}\cite{mossel2010noise}\cite{garban2010fourier}\cite{duchi13},
to name only a small subset of the literature.

In this paper, the connections between functional inequalities and information-theoretic (i.e.\ entropic) inequalities are further explored.
We propose a new inequality that generalizes both \eqref{e_bl} and \eqref{e_barthe}, and prove its properties using information-theoretic methods.
The organization is as follows.
In Section~\ref{sec_dual} we prove an extension of the duality of Carlen and Cordero-Erausquin, with a functional inequality that generalizes \eqref{e_bl} by allowing a cost function and non-deterministic transformations.
Both generalizations are essential for certain information-theoretic applications.
In Section~\ref{sec_fr}, a ``forward-reverse Brascamp-Lieb inequality'' is introduced, and we prove its information-theoretic formulation.
Although such an inequality essentially generalizes the forward inequality, the proof of its equivalent formulation is more involved and applies only to certain ``regular'' (though fairly general) spaces\footnote{More precisely, the ``entropic$\Rightarrow$functional''
direction is not more difficult than the case of forward inequality, but the ``functional$\Rightarrow$entropic'' direction requires sophisticated min-max theorems and is only proved in for Polish spaces.
In the finite alphabet case, the latter difficulty can be circumvented by using KKT conditions \cite{ISIT_lccv2016}.}.
Section~\ref{sec_special} discusses how the duality result unifies/generalizes the equivalent formulations of
R\'{e}nyi divergence,
the strong data processing inequality, hypercontractivity and its reverse (with positive or negative parameters), Loomis-Whitney inequality/Shearer's lemma, and transportation-cost inequalities, which have been proved by different methods (see for example \cite{atar2014information}\cite{ahlswede1976spreading}\cite{rad03}\cite{mt10}\cite{nair}\cite{beigi2016}\cite{bobkov1999}). The relationship among these inequalities is illustrated in Figure~\ref{fig_relation}.
In some of these examples (e.g.~strong data processing \cite{ahlswede1976spreading}) the previous approach relies heavily on the finiteness of the alphabet, whereas the present approach (essentially based on the nonnegativity of the relative entropy) is simpler and holds for general alphabets.

Sections~\ref{sec_ele}-\ref{sec_gaussian} illustrate several advantages of the information-theoretic formulation.
Data processing property, tensorization, and convexity are studied in Section~\ref{sec_ele}. Section~\ref{sec_gaussian} proves the Gaussian optimality in some information-theoretic optimization problems related to the dual (i.e.\ entropic) form of the Brascamp-Lieb inequality.
These can be viewed as generalizations of \eqref{e_bl} where the deterministic linear maps are replaced by Gaussian random transformations\footnote{That is, a random transformation $\mb{x}\mapsto \mb{A}\mb{x}+\mb{w}$ where $\mb{A}$ is deterministic and $\mb{w}$ is a Gaussian vector independent of $\bf x$.}.
In most cases, we are able to prove the Gaussian extremality and uniqueness of the minimizer under a certain non-degenerate assumption, while establishing the Gaussian exhaustibility in full generality\footnote{See the beginning of Section~\ref{sec_gaussian} for precise definitions of extremisability and exhaustibility.}.
In Section~\ref{sec_frg} we further establish the Gaussian optimality in the forward-reverse Brascamp-Lieb inequality.

Section~\ref{sec_consequence} discusses several implications of the Gaussian optimality results: some quantities/rate regions arising in information theory can be efficiently computed by solving a finite dimensional optimization problem in the Gaussian cases. Examples include multi-variate hypercontractivity, Wyner's common information for multiple variables, and certain secret key or common randomness generation problems. The relationship between the Gaussian optimality in the forward-reverse Brascamp-Lieb inequality and the transportation-cost inequalities for Gaussian measures is also discussed.

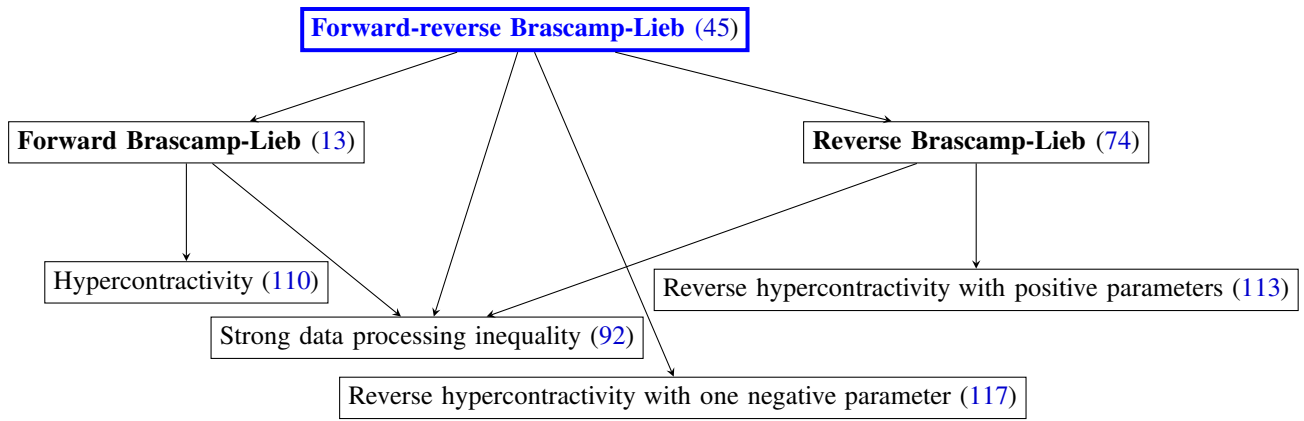
\begin{figure}[h!]
  \centering
\begin{tikzpicture}
[scale=2,
      dot/.style={draw,fill=black,circle,minimum size=0.7mm,inner sep=10pt},arw/.style={->,>=stealth}]
    \node[rectangle,draw,color=blue, line width=1.5pt] (X) {{\bf Forward-reverse Brascamp-Lieb \eqref{e_frbl_func}}};
    \node[rectangle,draw] (BL) [xshift=-45mm, yshift=-15mm]  {\bf Forward Brascamp-Lieb \eqref{e_func}};
    \node[rectangle,draw] (SDPI) [xshift=-13mm, yshift=-41mm]  {Strong data processing inequality \eqref{e_sdpi_func}};
    \node[rectangle,draw] (RHCN) [xshift=21mm,
    yshift=-49mm] {Reverse hypercontractivity with one negative parameter \eqref{e_rhcn}};
    \node[rectangle,draw] (RBL) [xshift=60mm,
    yshift=-15mm] {\bf Reverse Brascamp-Lieb \eqref{e_rbl}};
    \node[rectangle,draw] (HC) [below=13mm of BL]
    {Hypercontractivity \eqref{e_hc}};
    \node[rectangle,draw] (RHC) [below=14mm of RBL]
    {Reverse hypercontractivity with positive parameters \eqref{e_rhc}};
  \draw [arw] (X) to node[midway,above]{} (BL);
  \draw [arw] (X) to node[midway,left]{} (SDPI);
  \draw [arw] (X) to node[midway,above]{} (RBL);
  \draw [arw] (X) to node[midway]{} (RHCN);
  \draw [arw] (BL) to node[] {} (HC);
  \draw [arw] (RBL) to node[] {} (SDPI);
  \draw [arw] (RBL) to node[] {} (RHC);
  \draw [arw] (BL) to node[] {} (SDPI);
\end{tikzpicture}
\caption{The ``partial relations'' between the inequalities discussed in 
this paper with respect to implication.}
\label{fig_relation}
\end{figure}

\section{Dual Formulation of the Forward Brascamp-Lieb
Inequality}\label{sec_dual}
In this section we introduce a generalization of the forward Brascamp-Lieb inequality allowing cost functions and non-deterministic transformations, and prove its equivalent entropic formulation.
This will set the stage for the forward-reverse Brascamp-Lieb inequality to be discussed in Section~\ref{sec_fr}.

Given two nonnegative $\sigma$-finite measures\footnote{We shall use Greek letters to indicate unnormalized non-negative measures, and use capital English letters such as $P$ and $Q$ for probability measures.}
$\theta\ll \mu$ on $\mathcal{X}$, define the \emph{relative information} as the logarithm of the Radon-Nikodym derivative:
\begin{align}
\imath_{\theta\|\mu}(x):=\log\frac{{\rm d}\theta}{{\rm d}\mu}(x)
\end{align}
where $x\in\mathcal{X}$. Note that there is no assumption about $|\mathcal{X}|$.
The relative entropy between a probability measure $P$ and a $\sigma$-finite measure $\mu$ on the same measurable space is defined as
\begin{align}
D(P\|\mu):=\mathbb{E}\left[\imath_{P\|\mu}(X)\right]
\label{e8}
\end{align}
where $X\sim P$, if $P\ll \mu$, and infinity otherwise.

\begin{thm}\label{thm_1}
Fix $Q_X$, integer $m\in\{1,2,\dots\}$, and $Q_{Y_j|X}$, $c_j\in(0,\infty)$ for $j\in\{1,\dots,m\}$. Let $(X,Y_j)\sim Q_XQ_{Y_j|X}$. Assume that $d\colon\mathcal{X}\to(-\infty,\infty]$ is a measurable function satisfying
\begin{align}
0<\mathbb{E}[\exp(-d(X))]< \infty.
\label{e51}
\end{align}
The following statements are equivalent:
\begin{enumerate}
  \item For any non-negative measurable functions $f_j\colon\mathcal{Y}_j\to \mathbb{R}$, $j\in\{1,\dots,m\}$, it holds that
      \begin{align}
      \mathbb{E}\left[\exp\left(\sum_{j=1}^m\mathbb{E}[\log f_j(Y_j)|X]-d(X)\right)\right]
      &\le \prod_{j=1}^m\|f_j\|_{\frac{1}{c_j}}\label{e_func}
      \end{align}
      where the norm $\|f_j\|_{\frac{1}{c_j}}$ is with respect to $Q_{Y_j}$.
  \item For any distribution $P_X\ll Q_X$, it holds that
      \begin{align}
      D(P_X||Q_X)+\mathbb{E}[d(\hat{X})]\ge\sum_{j=1}^m c_jD(P_{Y_j}||Q_{Y_j})
        \label{e_info}
      \end{align}
      where $\hat{X}\sim P_X$, and $P_X\to Q_{Y_j|X}\to P_{Y_j}$ for $j\in\{1,\dots,m\}$.
\end{enumerate}
\end{thm}
\begin{proof}
\begin{itemize}
  \item 1)$\Rightarrow$2) Define
  \begin{align}
  d_0:=\log\mathbb{E}\left[\exp(-d(X))
  \prod_{j=1}^m\exp(c_j\mathbb{E}
  [\imath_{P_{Y_j}||Q_{Y_j}}(Y_j)|X])\right].
  \end{align}
    Invoking statement 1) with
\begin{align}
    f_j\leftarrow \left(\frac{{\rm d}P_{Y_j}}{{\rm d}Q_{Y_j}}\right)^{c_j}
    \label{e_f}
\end{align}
    we obtain
      \begin{align}\label{e_31}
  \exp(d_0)\le \prod_{j=1}^m\left(\mathbb{E}\left[\frac{{\rm d}P_{ Y_j}}{{\rm d}Q_{Y_j}}(Y_j)\right]\right)^{c_j}=1.
  \end{align}
  Now if
  $0<\exp(d_0)\le 1$ then
  \begin{align}
  {\rm d}S_X(x):=\exp(-d(x)-d_0)\prod_{j=1}^m\exp(c_j\mathbb{E}
  [\imath_{P_{Y_j}||Q_{Y_j}}(Y_j)|X=x]){\rm d}Q_X(x)
  \end{align}
  is a probability measure.
  Then \eqref{e_31} combined with the nonnegativity of relative entropy shows that
  \begin{align}
  \sum_{j=1}^mc_jD(P_{Y_j}||Q_{Y_j})
  &\le D(P_X||S_X)-d_0+\sum_{j=1}^mc_jD(P_{Y_j}||Q_{Y_j})
  \\
  &=D(P_X||Q_X)+\mathbb{E}[d(\hat{X})]
  \label{e_32}
  \end{align}
    and statement 2) holds.
  On the other hand, if $\exp(d_0)=0$, then for $Q_X$-almost all $x$,
\begin{align}
  \exp(-d(x))\prod_{j=1}^m\exp(c_j\mathbb{E}
  [\imath_{P_{Y_j}||Q_{Y_j}}(Y_j)|X=x])=0.
\end{align}
    Taking logarithms on both sides and taking the expectation with respect to $P_X$, we have
  \begin{align}
  -\mathbb{E}[d(\hat{X})]+\sum_{j=1}^m c_jD(P_{ Y_j}||Q_{Y_j})=-\infty
  \end{align}
  and statement 2) also follows.

  \item 2)$\Rightarrow$1) It suffices to prove for $f_j$'s such that $0 < a < f_j < b < \infty$ for some $a$ and $b$, since the general case will then follow by taking limits (e.g.~using monotone convergence theorem).
      By this assumption and \eqref{e51}, we can always define $P_X$ through
  \begin{align}
\imath_{P_X\|Q_X}(x)=-d(x)-d_0+\mathbb{E}
\left[\left.\sum_{j=1}^m\log f_j(Y_j)\right|{ X=x}\right]
\label{e12}
\end{align}
and $S_{Y_j}$ through
  \begin{align}
  \imath_{S_{Y_j}\|Q_{Y_j}}({ y_j}):=\frac{1}{c_j}\log f_j({ y_j})-d_j,
  \end{align}
  for each $j\in\{1,\dots,m\}$,
  where $d_j\in\mathbb{R}$, $j\in\{0,\dots,m\}$ are normalization constants, therefore
  \begin{align}
  \exp(d_0)&=\mathbb{E}\left[\exp\left(-d(X)
  +\mathbb{E}\left[\left.\sum_{j=1}^m\log f_j(Y_j)\right|X\right]\right)\right];
  \label{e35}
  \\
  \exp(d_j)&=\mathbb{E}\left[\exp\left(\frac{1}{c_j}\log f_j(Y_j)\right)\right],\quad j\in\{1,\dots,m\}.\label{e36}
  \end{align}
  But direct computation gives
  \begin{align}
  D(P_X||Q_X)&=-\mathbb{E}[d(\hat{ X})]-d_0+\mathbb{E}\left[\sum_{j=1}^m\log f_j(\hat{Y}_j)\right]
  \\
  D(P_{Y_j}||Q_{Y_j})&=D(P_{Y_j}||S_{Y_j})
  +\mathbb{E}\left[\imath_{S_{Y_j}||Q_{Y_j}}(\hat{ Y}_j)\right]
  \\
  &=D(P_{Y_j}||S_{Y_j})-d_j+\mathbb{E}\left[\frac{1}{c_j}\log f_j(\hat{Y}_j)\right]
  \end{align}
  where $\hat{Y}_j\sim P_{Y_j}$. Therefore statement 2) yields
  \begin{align}
  -d_0+\mathbb{E}\left[\sum_{j=1}^m \log f_j(\hat{ Y}_j)\right]
  \ge
 \sum_{j=1}^m c_jD(P_{Y_j}||S_{Y_j})-\sum_{j=1}^m c_jd_j
  +\mathbb{E}\left[\sum_{j=1}^m \log f_j(\hat{Y}_j)\right].
  \end{align}
  Since $f_j$'s are assumed to be bounded,
  $-\infty<\mathbb{E}\left[\sum_{j=1}^m \log f_j(\hat{Y}_j)\right]<\infty$
  so we can cancel it from the two sides of the inequality. It then follows from the non-negativity of relative entropy that
  \begin{align}
  d_0\le \sum_{j=1}^m c_jd_j
  \end{align}
  which is equivalent to statement 1) in view of \eqref{e35} and \eqref{e36}.
\end{itemize}
\end{proof}
\begin{rem}
If \eqref{e_func} holds and the equality is achieved by some $f_1,\dots,f_m$\footnote{Note that it is possible that some $f_1,\dots,f_m$ achieve the equality in \eqref{e_func}, but the inequality \eqref{e_func} does not hold for all functions.}
then the $P_X$ defined through \eqref{e12} achieves the equality in \eqref{e_info} (where $d_0$ is a normalization constant); conversely, if
\eqref{e_info} holds and the equality is achieved by $P_X$ then
\begin{align}
f_j=\left(\frac{{\rm d}P_{Y_j}}{{\rm d}Q_{Y_j}}\right)^{c_j}
\end{align}
for $j\in\{1,\dots,m\}$ achieve the equality in \eqref{e_func}.
These can be immediately verified by inspecting the tightness of each step in the proof of Theorem~\ref{thm_1}.
\end{rem}
\begin{rem}
The special case where $d$ is a constant function and $Q_{Y_j|X}$, $j\in\{1,\dots,m\}$ are deterministic was proved by Carlen and Cordero-Erausquin \cite[Theorem~2.1]{carlen2009subadditivity}, where the proof is based on the Donsker-Varadhan variational formula for the relative entropy.
In contrast, we prove Theorem~\ref{thm_1} by defining certain auxiliary measures and then reducing \eqref{e_func} or \eqref{e_info} to the nonnegativity of relative entropy.
These two methods, however, are closely related, since the variational formula of the relative entropy may be proved using the nonnegativity of the relative entropy.
\end{rem}
\begin{rem}\label{rem6}
As a convention, the left side of \eqref{e_info} is understood as $\mathbb{E}\left[\imath_{P_X\|Q_X}(\hat{X})+d(\hat{X})\right]$ in case that it is otherwise undefined.
Note that the assumption \eqref{e51} ensures that
\begin{align}
{\rm d}T_X(x):=\frac{\exp(-d(x)){\rm d}Q_X(x)}{\mathbb{E}[\exp(-d(X))]}
\end{align}
defines a probability measure $T_X$ so that
\begin{align}
\mathbb{E}\left[\imath_{P_X\|Q_X}(\hat{X})+d(\hat{X})\right]
=D\left(P_X\|T_X\right)
-\log \mathbb{E}[\exp(-d(X))]
\end{align}
is always well-defined (finite or $+\infty$).
\end{rem}
\begin{rem}
By a result of Csisz\'{a}r regarding the $I$-projection onto a set specified by linear constraints \cite[Section~3]{csiszar1975divergence},
\eqref{e_info} for all $P_X\ll Q_X$ is equivalent to
\eqref{e_info} for all $P_X$ of the form \eqref{e12} (for some $(f_j)_{j=1}^m$).
\end{rem}

\begin{rem}
For finite alphabets, it might be possible to first prove the equivalence of \eqref{e_func} and \eqref{e_info} for the source distribution $Q_{XY^m}$ and the random transformations $Q_{X|XY^m},Q_{Y_1|XY^m},\dots,Q_{Y_m|XY^m}$ (that is, when the random transformations are simply projections onto the coordinates), and then obtain the equivalence for the source distribution $Q_X$ and the random transformations $(Q_{Y_j|X})_{j=1}^m$ by taking suitable limits, in a similar manner that hypercontractivity is shown to recover the strong data processing inequality
(cf.~\cite[Theorem 5a]{ahlswede1976spreading}).
However, the argument of taking limits is technically involved even for finite alphabets.
Moreover, that approach appears to be insufficient if we are further interested in the cases of equality in Theorem~\ref{thm_1}.
\end{rem}
\begin{rem}\label{thm_1_ext}
In the statements of Theorem~\ref{thm_1}, $Q_X$ is a probability measure and
$Q_X$ and $Q_{Y_j}$ are connected through the random transformation $Q_{Y_j|X}$. These help to keep our notations simple and suffice for most applications in our paper.
However, from the proof it is clear that these restrictions are not really necessary.
In other words, we have the extension of Theorem~\ref{thm_1} that the following two statements are equivalent:
      \begin{align}
      \int\exp\left(\sum_{j=1}^m\mathbb{E}[\log f_j(Y_j)|X=x]-d(x)\right){\rm d}\nu(x)
      &\le \prod_{j=1}^m\|f_j\|_{\frac{1}{c_j}},
      \quad\forall f_1,\dots,f_m;
      \label{e27_ext}
      \end{align}
      \begin{align}
      D(P_X||\nu)
      +\int d(x){\rm d}\nu(x)
      \ge\sum_{j=1}^m c_jD(P_{Y_j}||\mu_j),
      \quad \forall P_X,
      \label{e_35_ext}
      \end{align}
      where again
      $f_j\colon\mathcal{Y}_j\to \mathbb{R}, j\in\{1,\dots,m\}$
      are nonnegative measurable functions,
      $P_X\ll \nu$, and $P_X\to Q_{Y_j|X}\to P_{Y_j}$ for $j\in\{1,\dots,m\}$.
      Here $\nu$ and $\mu_j$ need not be normalized and need not be connected by $Q_{Y_j|X}$,
      and $\|\cdot\|_{\frac{1}{c_j}}$ is with respect to $\mu_j$.
\end{rem}

\section{Dual Formulation of a Forward-Reverse Brascamp-Lieb Inequality}\label{sec_fr}
In this section we introduce a new type of functional inequality which may be called ``forward-reverse Brascamp-Lieb inequality'', and prove its equivalent entropic formulation.
As alluded in Section~\ref{sec1}, the proof of the ``functional$\Rightarrow$entropic'' direction for the forward-reverse inequality is much more sophisticated than for the forward inequality, and some regularity assumptions on the alphabets and the measures appear to be necessary.

Throughout this paper, when the forward-reverse inequality is considered, we always assume that the alphabets are Polish spaces, and the measures are Borel measures\footnote{A Polish space is a complete separable metric space. It enjoys several nice properties that we use heavily in this section, including Prokhorov theorem and Riesz-Kakutani theorem (the latter is related to the fact that every Borel probability measure on a Polish space is inner regular, hence a Radon measure). Short introductions on the Polish space can be found in e.g.~\cite{villani2003topics}\cite{dembo2009large}.}. Of course, this covers the cases where the alphabet is Euclidean or discrete (endowed with the Hamming metric, which induces the discrete topology, making every function on the discrete set continuous), among others. Readers interested in finite-alphabets only may refer to the (much simpler) argument in \cite{ISIT_lccv2016} based on the KKT condition.

\begin{notat}
Let $\mathcal{X}$ be a topological space.
\begin{itemize}
\item $C_c(\mathcal{X})$ denotes the space of continuous functions on $\mathcal{X}$ with a compact support;
\item $C_0(\mathcal{X})$ denotes the space of all continuous function $f$ on $\mathcal{X}$ that vanishes at infinity (i.e.~for any $\epsilon>0$ there exists a compact set $\mathcal{K}\subseteq \mathcal{X}$ such that $|f(x)|<\epsilon$ for $x\in \mathcal{X}\setminus\mathcal{K}$);
\item $C_b(\mathcal{X})$ denotes the space of bounded continuous functions on $\mathcal{X}$;
\item $\mathcal{M}(\mathcal{X})$ denotes the space of finite signed Borel measures on $\mathcal{X}$;
\item $\mathcal{P}(\mathcal{X})$ denotes the space of probability measures on $\mathcal{X}$.
\end{itemize}
\end{notat}
We consider $C_c$, $C_0$ and $C_b$ as topological vector spaces, with the topology induced from the sup norm. The following theorem, usually attributed to Riesz, Markov and Kakutani, is well-known in functional analysis and can be found in, e.g.~\cite{lax}\cite{tao2009}.
\begin{thm}[Riesz-Markov-Kakutani]\label{thm_rmk}
If $\mathcal{X}$ is a locally compact, $\sigma$-compact Polish space, the dual\footnote{The dual of a topological vector space consists of all continuous linear functionals on that space, which is naturally also topological vector space (with the weak$^*$ topology).} of both $C_c(\mathcal{X})$ and $C_0(\mathcal{X})$ is $\mathcal{M}(\mathcal{X})$.
\end{thm}
\begin{rem}
The dual space of $C_b(\mathcal{X})$ can be strictly larger than $\mathcal{M}(\mathcal{X})$, since it also contains those linear functionals that depend on the ``limit at infinity'' of a function $f\in\mathcal{C}_b(\mathcal{X})$ (originally defined for those $f$ that do have a limit at the infinity, and then extended to the whole $C_b(\mathcal{X})$ by Hahn-Banach theorem; see e.g.~\cite{lax}).
\end{rem}

Of course, any $\mu\in\mathcal{M}(\mathcal{X})$ is a continuous linear functional on $C_0(\mathcal{X})$ or $C_c(\mathcal{X})$, given by
\begin{align}
f\mapsto \int f{\rm d}\mu
\end{align}
where $f$ is a function in $C_0(\mathcal{X})$ or $C_c(\mathcal{X})$. Remarkably, Theorem~\ref{thm_rmk} states that the converse is also true under mild regularity assumptions on the space. Thus, we can view measures as continuous linear functionals on a certain function space;\footnote{In fact, some authors prefer to construct the measure theory by \emph{defining} a measure as a linear functional on a suitable measure space; see Lax \cite{lax} or Bourbaki \cite{bourbaki}.}
this justifies the shorthand notation
\begin{align}
\mu(f):=\int f{\rm d}\mu
\label{e_shorthand}
\end{align}
which we employ in the rest of the paper. This viewpoint is the most natural for our setting since in the proof of the equivalent formulation of the forward-reverse Brascamp-Lieb inequality we shall use the Hahn-Banach theorem to show the existence of certain linear functionals.

\begin{defn}\label{defn_LFT}
Let $\Lambda\colon C_b(\mathcal{X})\to (-\infty,+\infty]$ be a lower semicontinuous, proper convex function.
Its \emph{Legendre-Fenchel transform} $\Lambda^*\colon\mathcal{C}_b(\mathcal{X})^*\to (-\infty,+\infty]$ is given by
\begin{align}
\Lambda^*(\ell):=\sup_{u\in C_b(\mathcal{X})}[\ell(u)-\Lambda(u)].
\end{align}
\end{defn}
Let $\nu$ be a nonnegative finite Borel measure on a Polish space $\mathcal{X}$, and define a convex functional on $C_b(\mathcal{X})$:
\begin{align}
\Lambda(f)&:= \log\nu(\exp(f))
\\
&=\log\int \exp(f){\rm d}\nu.
\label{e38}
\end{align}
Then note that the relative entropy has the following alternative definition: for any $\mu\in\mathcal{M}(\mathcal{\mathcal{X}})$,
\begin{align}
D(\mu\|\nu):=\sup_{f\in C_b(\mathcal{X})}[\mu(f)-\Lambda(f)]
\label{e_dv}
\end{align}
which agrees with the definition \eqref{e8} when $\nu$ is a probability measure, by the Donsker-Varadhan formula (c.f.~\cite[Lemma~6.2.13]{dembo2009large}).
If $\mu$ is not a probability measure, then $D(\mu\|\nu)$ as defined in \eqref{e_dv} is $+\infty$.

Given a bounded linear operator $T\colon C_b(\mathcal{Y})\to C_b(\mathcal{X})$, the dual operator $T^*\colon C_b(\mathcal{X})^*\to C_b(\mathcal{Y})^*$ is defined in terms of
\begin{align}
T^*\mu_X\colon f\in C_b(\mathcal{Y})\mapsto \mu_X(Tf),
\end{align}
for any $\mu_X\in C_b(\mathcal{X})^*$. Since $\mathcal{P}(\mathcal{X})\subseteq \mathcal{M}(\mathcal{X})\subseteq C_b(\mathcal{X})^*$,
we can define a \emph{conditional expectation operator} as any $T$ such that $T^*P\in \mathcal{P}(\mathcal{Y})$ for any $P\in\mathcal{P}(\mathcal{X})$.
A \emph{random transformation} $T^*$ is defined as the dual of some conditional expectation operator.

\begin{rem}
From the viewpoint of category theory (see for example \cite{category}\cite{hatcher606algebraic}),
$C_b$ is a functor from the category of topological spaces to the category of topological vector spaces, which is contra-variant because for any continuous,
$\phi\colon \mathcal{X}\to\mathcal{Y}$ (morphism between topological spaces),
we have $C_b(\phi)\colon C_b(\mathcal{Y})\to C_b(\mathcal{X})$, $u\mapsto u\circ f$ where $u\circ \phi$ denotes the composition of two continuous functions,
reversing the arrows in the maps (i.e.\ the morphisms). On the other hand, $\mathcal{M}$ is a covariant functor and $\mathcal{M}(\phi)\colon \mathcal{M}(\mathcal{X})\to \mathcal{M}(\mathcal{Y})$, $\mu\mapsto \mu\circ \phi^{-1}$, where $\mu\circ \phi^{-1}(\mathcal{B}):=\mu(\phi^{-1}(\mathcal{B}))$ for any Borel measurable $\mathcal{B}\subseteq \mathcal{Y}$.
``Duality'' itself is a contra-variant functor between the category of topological spaces (note the reversal of arrows in Fig.~\ref{fig_unification}).
Moreover, $C_b(\mathcal{X})^*=\mathcal{M}(\mathcal{X})$ and $C_b(\phi)^*=\mathcal{M}(\phi)$ if $\mathcal{X}$ and $\mathcal{Y}$ are compact metric spaces and $\phi\colon\mathcal{X}\to\mathcal{Y}$ is continuous.
Definition~\ref{exp_defn1} can therefore be viewed as the special case where $\phi$ is the projection map:
\end{rem}

\begin{defn}\label{exp_defn1}
Suppose $\phi\colon\mathcal{Z}_1\times\mathcal{Z}_2\to \mathcal{Z}_1,\,(z_1,z_2)\mapsto z_1$ is the projection to the first coordinate.
\begin{itemize}
  \item $C_b(\phi)\colon C_b(\mathcal{Z}_1)\to C_b(\mathcal{Z}_1\times \mathcal{Z}_2)$ is called a \emph{canonical map}, whose action is almost trivial: it sends a function of $z_i$ to itself, but viewed as a function of $(z_1,z_2)$.
  \item $\mathcal{M}(\phi)\colon \mathcal{M}(\mathcal{Z}_1\times\mathcal{Z}_2)\to \mathcal{M}(\mathcal{Z}_1)$ is called \emph{marginalization}, which simply takes a joint distribution to a marginal distribution.
\end{itemize}
\end{defn}

\subsection{Compact $\mathcal{X}$}
We first state a duality theorem for the case of compact alphabets to streamline the proof. Later we show that the argument can be extended to a particular non-compact case.\footnote{Theorem~\ref{thm_unification}
is not included in the conference paper \cite{ISIT_lccv2016}, but was announced in the conference presentation.}
Our proof based on the Legendre-Fenchel duality (Theorem~\ref{thm_FRduality} in Appendix~\ref{app_fr}) was inspired by the proof of the Kantorovich duality in the theory of optimal transportation (see \cite[Chapter~1]{villani2003topics}, where the idea was credited to Brenier).

\begin{thm}[Dual formulation of forward-reverse Brascamp-Lieb inequality]\label{thm_unification}
Assume that
\begin{itemize}
\item $m$ and $l$ are positive integers, $d\in\mathbb{R}$, $\mathcal{X}$ is a compact metric space (hence also a Polish space);
\item For each $i=1,\dots,l$, $b_i\in (0,\infty)$; $\nu_i$ is a finite Borel measure on a Polish space $\mathcal{Z}_i$, and $Q_{Z_i|X}=S_i^*$ is a random transformation;
\item For each $j=1,\dots,m$, $c_j\in (0,\infty)$, $\mu_j$ is a finite Borel measure on a Polish space $\mathcal{Y}_j$, and $Q_{Y_j|X}=T_j^*$ is a random transformation.
\item For any $P_{Z_i}$ such that $D(P_{Z_i}\|\nu_i)<\infty$, $i=1,\dots,l$, there exists $P_X\in\bigcap_i(S_i^*)^{-1}P_{Z_i}$ such that $\sum_{j=1}^m c_j D(P_{Y_j}\|\mu_j)< \infty$, where $P_{Y_j}:=T_j^*P_X$.
\end{itemize}
Then the following two statements are equivalent:
\begin{enumerate}
  \item If nonnegative continuous functions $(g_i)$, $(f_j)$ are bounded away from $0$ and such that
      \begin{align}
      \sum_{i=1}^l b_iS_i\log g_i\le \sum_{j=1}^m c_jT_j\log f_j
      \label{e34}
      \end{align}
      then (see \eqref{e_shorthand} for the notation of the integral)
      \begin{align}
      \prod_{i=1}^l\nu_i^{b_i}(g_i)\le
      \exp(d)\prod_{j=1}^m\mu_j^{c_j}(f_j).
      \label{e_frbl_func}
      \end{align}
  \item For any $(P_{Z_i})$ such that $D(P_{Z_i}\|\nu_i)<\infty$\footnote{Of course, this assumption is not essential (once we adopt the convention that the infimum in \eqref{e_frbl_entropic} is $+\infty$ when it runs over an empty set).}, $i=1,\dots,l$,
      \begin{align}
      \sum_{i=1}^l b_iD(P_{Z_i}\|\nu_i)+d
      \ge \inf_{P_X}\sum_{j=1}^m c_j D(P_{Y_j}\|\mu_j)
      \label{e_frbl_entropic}
      \end{align}
      where the infimum is over $P_X$ such that $S_i^*P_X=P_{Z_i}$, $i=1,\dots,l$, and $P_{Y_j}:=T_j^*P_X$, $j=1,\dots,m$.
\end{enumerate}
\end{thm}
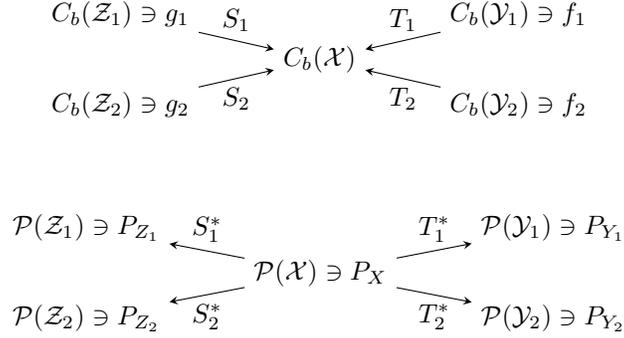
\begin{figure}[h!]
  \centering
\begin{tikzpicture}
[scale=2,
      dot/.style={draw,fill=black,circle,minimum size=0.7mm,inner sep=0pt},arw/.style={->,>=stealth}]
    \node[rectangle] (X) {$C_b(\mathcal{X})$};
    \node[rectangle] (Z1) [above left=of X, xshift=0mm, yshift=-10mm] {$C_b(\mathcal{Z}_1)\ni g_1$};
    \node[rectangle] (Z2) [below left=of X, xshift=0mm, yshift=10mm] {$C_b(\mathcal{Z}_2)\ni g_2$};
    \node[rectangle] (Y1) [above right=of X, xshift=0mm, yshift=-10mm] {$C_b(\mathcal{Y}_1)\ni f_1$};
    \node[rectangle] (Y2) [below right=of X, xshift=0mm, yshift=10mm] {$C_b(\mathcal{Y}_2)\ni f_2$};

  \draw [arw] (Z1) to node[midway,above]{$S_1$} (X);
  \draw [arw] (Z2) to node[midway,below]{$S_2$} (X);
  \draw [arw] (Y1) to node[midway,above]{$T_1$} (X);
  \draw [arw] (Y2) to node[midway,below]{$T_2$} (X);
\end{tikzpicture}
\\
\vspace{30pt}
\begin{tikzpicture}
[scale=2,
      dot/.style={draw,fill=black,circle,minimum size=0.7mm,inner sep=0pt},arw/.style={->,>=stealth}]
    \node[rectangle] (X) {$\mathcal{P}(\mathcal{X})\ni P_X$};
    \node[rectangle] (Z1) [above left=of X, xshift=0mm, yshift=-10mm] {$\mathcal{P}(\mathcal{Z}_1)\ni P_{Z_1}$};
    \node[rectangle] (Z2) [below left=of X, xshift=0mm, yshift=10mm] {$\mathcal{P}(\mathcal{Z}_2)\ni P_{Z_2}$};
    \node[rectangle] (Y1) [above right=of X, xshift=0mm, yshift=-10mm] {$\mathcal{P}(\mathcal{Y}_1)\ni P_{Y_1}$};
    \node[rectangle] (Y2) [below right=of X, xshift=0mm, yshift=10mm] {$\mathcal{P}(\mathcal{Y}_2)\ni P_{Y_2}$};

  \draw [arw] (X) to node[midway,above]{$S_1^*$} (Z1);
  \draw [arw] (X) to node[midway,below]{$S_2^*$} (Z2);
  \draw [arw] (X) to node[midway,above]{$T_1^*$} (Y1);
  \draw [arw] (X) to node[midway,below]{$T_2^*$} (Y2);
\end{tikzpicture}
\caption{Diagrams for Theorem~\ref{thm_unification}.}
\label{fig_unification}
\end{figure}
\begin{proof} We can safely assume $d=0$ below without loss of generality (since otherwise we can always substitute $\mu_1\leftarrow\exp\left(\frac{d}{c_1}\right)\mu_1$).
\begin{description}
  \item [1)$\Rightarrow$2)]
%
%
%
This is the nontrivial direction which relies on certain (strong) min-max type results.
In Theorem~\ref{thm_FRduality} in Appendix~\ref{app_fr}, put\footnote{In \eqref{e37}, $u\le 0$ means that $u$ is pointwise non-positive.}
\begin{align}
\Theta_0\colon u\in C_b(\mathcal{X})
\mapsto
\left\{
\begin{array}{cc}
0 & u\le0; \\
+\infty & \textrm{otherwise}.
\end{array}
\right.
\label{e37}
\end{align}
Then,
\begin{align}
\Theta_0^*
\colon
\pi\in \mathcal{M}(\mathcal{X})
\mapsto\left\{
\begin{array}{cc}
0 & \pi\ge0; \\
+\infty & \textrm{otherwise}.
\end{array}
\right.
\end{align}
For each $j=1,\dots,m$, set
\begin{align}
\Theta_j(u)
=c_j\inf\log \mu_j \left(\exp\left(\frac{1}{c_j}v\right)\right)
\label{e39}
\end{align}
where the infimum is over $v\in C_b(\mathcal{Y})$ such that $u=T_jv$; if there is no such $v$ then $\Theta_j(u):=+\infty$ as a convention. Observe that
\begin{itemize}
  \item $\Theta_j$ is convex: indeed given arbitrary $u^0$ and $u^1$, suppose that $v^0$ and $v^1$ respectively achieve the infimum in \eqref{e39} for $u^0$ and $u^1$ (if the infimum is not achievable, the argument still goes through by the approximation and limit argument).
      Then for any $\alpha\in [0,1]$, $v^{\alpha}:=(1-\alpha)v^0+\alpha v^1$ satisfies $u^{\alpha}=T_jv^{\alpha}$ where $u^{\alpha}:=(1-\alpha)u^0+\alpha u^1$.
      Thus the convexity of $\Theta_j$ follows from the convexity of the functional in \eqref{e38};
  \item $\Theta_j(u)>-\infty$ for any $u\in C_b(\mathcal{X})$. If otherwise, for any $P_X$ and $P_{Y_j}:=T_j^*P_X$ we have
      \begin{align}
      D(P_{Y_j}\|\mu_j)
      &=\sup_v\{P_{Y_j}(v)-\log \mu_j(\exp(v))\}
      \label{e40}
      \\
      &=\sup_v\{P_X(T_jv)-\log \mu_j(\exp(v))\}
      \\
      &=\sup_{u\in C_b(\mathcal{X})}\left\{P_X(u)-\frac{1}{c_j}\Theta_j(c_ju)\right\}
      \label{e42}
      \\
      &=+\infty
      \end{align}
      which contradicts the assumption that $\sum_{j=1}^m c_j D(P_{Y_j}\|\mu_j)< \infty$ in the theorem;
  \item From the steps \eqref{e40}-\eqref{e42}, we see $\Theta_j^*(\pi)=c_jD(T_j^*\pi\|\mu_j)$ for any $\pi\in\mathcal{M}(\mathcal{X})$, where the definition of $D(\cdot\|\mu_j)$ is extended using the Donsker-Varadhan formula (that is, it is infinite when the argument is not a probability measure).
\end{itemize}
Finally, for the given $(P_{Z_i})_{i=1}^l$, choose
\begin{align}
\Theta_{m+1} \colon u\in C_b(\mathcal{X})
\mapsto
\left\{
\begin{array}{cc}
\sum_{i=1}^l P_{Z_i}(w_i)  & \textrm{if $u=\sum_{i=1}^l S_iw_i$ for some $w_i\in C_b(\mathcal{Z}_i)$}; \\
+\infty & \textrm{otherwise}.
\end{array}
\right.
\label{e44}
\end{align}
Notice that
\begin{itemize}
  \item $\Theta_{m+1}$ is convex;
  \item $\Theta_{m+1}$ is well-defined (that is, the choice of $(w_i)$ in \eqref{e44} is inconsequential). Indeed if $(w_i)_{i=1}^l$ is such that $\sum_{i=1}^lS_i w_i=0$, then
      \begin{align}
      \sum_{i=1}^l P_{Z_i}(w_i)
      &=
      \sum_{i=1}^l S_i^*P_X(w_i)
      \\
      &=
      \sum_{i=1}^l P_X(S_iw_i)
      \\
      &=0,
      \end{align}
  where $P_X$ is such that $S_i^*P_X=P_{Z_i}$, $i=1,\dots,l$, whose existence is guaranteed by the assumption of the theorem. This also shows that $\Theta_{m+1}>-\infty$.
  \item
  \begin{align}
  \Theta_{m+1}^*(\pi)
  &:=\sup_u\{\pi(u)-\Theta_{m+1}(u)\}
  \\
  &=\sup_{w_1,\dots,w_l}
  \left\{\pi\left(\sum_{i=1}^lS_iw_i\right)-\sum_{i=1}^l P_{Z_i}(w_i)\right\}
  \\
  &=\sup_{w_1,\dots,w_l}
  \left\{\sum_{i=1}^lS_i^*\pi(w_i)-\sum_{i=1}^l P_{Z_i}(w_i)\right\}
  \\
  &=\left\{
  \begin{array}{cc}
    0 & \textrm{if }S_i^*\pi=P_{Z_i},\quad i=1,\dots,l; \\
    +\infty & \textrm{otherwise}.
  \end{array}
  \right.
  \end{align}
\end{itemize}
Invoking Theorem~\ref{thm_FRduality} in Appendix~\ref{app_fr} (where the $u_j$ in Theorem~\ref{thm_FRduality} can be chosen as the constant function $u_j\equiv1$, $j=1,\dots,m+1$):
\begin{align}
&\quad\inf_{
\pi\colon \pi\ge0,\,
S_i^*\pi=P_{Z_i}}
\sum_{j=1}^m c_j D(T_j^*\pi\|\mu_j)\\
&=-\inf_{v^m,w^l\colon \sum_{j=1}^mT_jv_j
+\sum_{i=1}^lS_iw_i\ge0}\left[\sum_{j=1}^mc_j\log\mu_j\left(\exp\left(
\frac{1}{c_j}v_j\right)\right)
+\sum_{i=1}^lP_{Z_i}(w_i)\right]
\label{e_52}
\end{align}
where $v^m$ denotes the collection of the functions $v_1,\dots,v_m$, and similarly for $w^l$. Note that the left side of \eqref{e_52} is exactly the right side of \eqref{e_frbl_entropic}. For any $\epsilon>0$, choose $v_j\in C_b(\mathcal{Y}_j)$, $j=1,\dots,m$ and $w_i\in C_b(\mathcal{Z}_i)$, $i=1,\dots,l$ such that $\sum_{j=1}^mT_jv_j
+\sum_{i=1}^lS_iw_i\ge0$ and
\begin{align}
\epsilon-\sum_{j=1}^mc_j\log\mu_j\left(\exp\left(
\frac{1}{c_j}v_j\right)\right)
-\sum_{i=1}^lP_{Z_i}(w_i)
>\inf_{
\pi\colon \pi\ge0,\,
S_i^*\pi=P_{Z_i}}
\sum_{j=1}^mc_j D(T_j^*\pi\|\mu_j)
\label{e_54}
\end{align}
Now invoking \eqref{e_frbl_func} with $f_j:=\exp\left(\frac{1}{c_j}v_j\right)$, $j=1,\dots,m$ and $g_i:=\exp\left(-\frac{1}{b_i}w_i\right)$, $i=1,\dots,l$,
we upper bound the left side of \eqref{e_54} by
\begin{align}
\epsilon-\sum_{i=1}^lb_i\log\nu_i(g_i)+\sum_{i=1}^lb_iP_{Z_i}(\log g_i)
\le
\epsilon+\sum_{i=1}^lb_iD(P_{Z_i}\|\nu_i)
\end{align}
where the last step follows by the Donsker-Varadhan formula.
Therefore \eqref{e_frbl_entropic} is established since $\epsilon>0$ is arbitrary.

  \item [2)$\Rightarrow$1)] Since $\nu_i$ is finite and $g_i$ is bounded by assumption, we have $\nu_i(g_i)<\infty$, $i=1,\dots,l$. Moreover \eqref{e_frbl_func} is trivially true when $\nu_i(g_i)=0$ for some $i$, so we will assume below that $\nu_i(g_i)\in (0,\infty)$ for each $i$. Define $P_{Z_i}$ by
      \begin{align}
      \frac{{\rm d}P_{Z_i}}{{\rm d}\nu_i}=\frac{g_i}{\nu_i(g_i)},\quad i=1,\dots,l.
      \end{align}
      Then for any $\epsilon>0$,
      \begin{align}
      \sum_{i=1}^lb_i\log \nu_i(g_i)
      &=
      \sum_{i=1}^lb_i[P_{Z_i}(\log g_i)-D(P_{Z_i}\|\nu_i)]
      \label{e64}
      \\
      &<
      \sum_{j=1}^mc_jP_{Y_j}(\log f_j)
      +\epsilon-\sum_{j=1}^m c_jD(P_{Y_j}\|\mu_j)
      \label{e_58}
      \\
      &\le \epsilon+\sum_{j=1}^mc_j\log\mu_j(f_j)
      \label{e_59}
      \end{align}
      where
      \begin{itemize}
        \item \eqref{e_58} uses the Donsker-Varadhan formula, and we have chosen $P_X$, $P_{Y_j}:=T_j^*P_X$, $j=1,\dots,m$ such that
            \begin{align}
            \sum_{i=1}^l b_iD(P_{Z_i}\|\nu_i)
      > \sum_{j=1}^m c_j D(P_{Y_j}\|\mu_j)-\epsilon
            \end{align}
        \item \eqref{e_59} also follows from the Donsker-Varadhan formula.
      \end{itemize}
      The result follows since $\epsilon>0$ can be arbitrary.
\end{description}
\end{proof}

\begin{rem}
The infimum in \eqref{e_frbl_entropic} is in fact achievable: For any $(P_{Z_i})$,
there exists a $P_X$ that minimizes $\sum_{j=1}^m c_j D(P_{Y_j}\|\mu_j)$ subject to the constraints $S_i^*P_X=P_{Z_i}$, $i=1,\dots m$, where $P_{Y_j}:=T_j^*P_X$, $j=1,\dots,m$.
Indeed, since the singleton $\{P_{Z_i}\}$ is weak$^*$-closed and $S^*_i$ is weak$^*$-continuous\footnote{Generally, if $T\colon A\to B$ is a continuous map between two topologically vector spaces, then $T^*\colon B^*\to A^*$ is a weak$^*$ continuous map between the dual spaces. Indeed, if $y_n\to y$ is a weak$^*$-convergent subsequence in $B^*$, meaning $y_n(b)\to y(b)$ for any $b\in B$, then we must have $T^* y_n(a)=y_n(Ta)\to y(Ta)=T^*y(a)$ for any $a\in A$, meaning that $T^* y_n$ converges to $T^*y$ in the weak$^*$ topology.}, the set $\bigcap_{i=1}^l(S^*_i)^{-1}P_{Z_i}$ is weak$^*$-closed in $\mathcal{M}(X)$;
hence its intersection with $\mathcal{P}(\mathcal{X})$ is weak$^*$-compact in $\mathcal{P}(\mathcal{X})$,
because $\mathcal{P}(\mathcal{X})$ is weak$^*$-compact by (a simple version for the setting of a compact underlying space $\mathcal{X}$ of) the Prokhorov theorem \cite{prokhorov1956convergence}.
Moreover, by the weak$^*$-lower semicontinuity of $D(\cdot\|\mu_j)$ (easily seen from the variational formula/Donsker-Varadhan formula of the relative entropy, cf.~\cite{verdubook}) and the weak$^*$-continuity of $T_j^*$, $j=1,\dots,m$, we see $\sum_{j=1}^m c_j D(T_j^*P_X\|\mu_j)$ is weak$^*$-lower semicontinuous in $P_X$, and hence the existence of a minimizing $P_X$ is established.
\end{rem}

\begin{rem}
Abusing the terminology from the min-max theory, Theorem~\ref{thm_unification} may be interpreted as a ``strong duality'' result which establishes the equivalence of two optimization problems. The 1)$\Rightarrow$2) part is the non-trivial direction which  requires regularity on the spaces. In contrast, the 2)$\Rightarrow$1) direction can be thought of as a ``weak duality'' which establishes only a partial relation but holds for more general spaces.
\end{rem}

\begin{rem}
The equivalent formulations of the forward Brascamp-Lieb inequality (Theorem~\ref{thm_1}) can be recovered from Theorem~\ref{thm_unification} by taking $l=1$, $\mathcal{Z}_1=\mathcal{X}$, and letting $S_1$ be the identity map/isomorphism, except that Theorem~\ref{thm_1} is established for completely general alphabets.
In other words, the forward Brascamp-Lieb inequality is the special case of the forward-reverse Brascamp-Lieb inequality when there is only one \emph{reverse channel} which is the identity.
\end{rem}

\subsection{Noncompact $\mathcal{X}$}
Our proof of 1)$\Rightarrow$2) in Theorem~\ref{thm_unification} makes use of the Hahn-Banach theorem, and hence relies crucially on the fact that the measure space is the dual of the function space. Naively, one might want to extend the the proof to the case of \emph{locally compact} $\mathcal{X}$ by considering $C_0(\mathcal{X})$ instead of $C_b(\mathcal{X})$, so that the dual space is still $\mathcal{M}(\mathcal{X})$. However, this would not work: consider the case when $\mathcal{X}=\mathcal{Z}_1\times,\dots,\times\mathcal{Z}_l$ and each $S_i$ is the canonical map. Then $\Theta_{m+1}(u)$ as defined in \eqref{e44} is $+\infty$ unless $u\equiv0$ (because $u\in C_0(\mathcal{X})$ requires that $u$ vanishes at infinity), thus $\Theta_{m+1}^*\equiv0$. Luckily, we can still work with $C_b(\mathcal{X})$; in this case $\ell\in C_b(\mathcal{X})^*$ may not be a measure, but we can decompose it into $\ell=\pi+R$ where $\pi\in\mathcal{M}(\mathcal{X})$ and $R$ is a linear functional ``supported at the infinity''.
Below we use the techniques in \cite[Chapter 1.3]{villani2003topics}
to prove a particular extension of Theorem~\ref{thm_unification} to a non-compact case.
\begin{thm}\label{thm_noncpt}
Theorem~\ref{thm_unification} still holds if
\begin{itemize}
  \item The assumption that $\mathcal{X}$ is a compact metric space is relaxed to the assumption that it is a locally compact and $\sigma$-compact Polish space;
  \item $\mathcal{X}=\prod_{i=1}^l\mathcal{Z}_i$ and $S_i\colon C_b(\mathcal{Z}_i)\to C_b(\mathcal{X})$, $i=1,\dots,l$ are canonical maps (see Definition~\ref{exp_defn1}).
\end{itemize}
\end{thm}

\begin{proof}
The proof of the ``weak duality'' part 2)$\Rightarrow$1) still works in the noncompact case, so we only need to explain what changes need to be made in the proof of 1)$\Rightarrow$2) part. Let $\Theta_0$ be defined as before, in \eqref{e37}. Then for any $\ell\in C_b(\mathcal{X})^*$,
\begin{align}
\Theta_0^*(\ell)=\sup_{u\le 0}\ell(u)
\end{align}
which is $0$ if $\ell$ is nonnegative (in the sense that $\ell(u)\ge 0$ for every $u\ge 0$), and $+\infty$ otherwise. This means that when computing the infimum on the left side of \eqref{e62}, we only need to take into account of those nonnegative $\ell$.

Next, let $\Theta_{m+1}$ be also defined as before. Then directly from the definition we have
\begin{align}
\Theta_{m+1}^*(\ell)
&=\left\{
\begin{array}{cc}
0 & \textrm{if }\ell(\sum_iS_iw_i)=\sum_iP_{Z_i}(w_i),\quad \forall w_i\in C_b(\mathcal{Z}_i),\,i=1,\dots l; \\
+\infty & \textrm{otherwise}.
\end{array}
\right.
\label{e72}
\end{align}
For any $\ell\in C_b^*(\mathcal{X})$.
Generally, the condition in the first line of \eqref{e72} does not imply that $\ell$ is a measure.
However, if $\ell$ is also nonnegative, then using a technical result in \cite[Lemma~1.25]{villani2003topics} we can further simplify:
\begin{align}
\Theta_{m+1}^*(\ell)
  &=\left\{
  \begin{array}{cc}
    0 & \textrm{if }\ell\in\mathcal{M}(\mathcal{X})\textrm{ and }S_i^*\ell=P_{Z_i},\quad i=1,\dots,l; \\
    +\infty & \textrm{otherwise}.
  \end{array}
  \right.
\label{e65}
\end{align}
This further shows that when we compute the left side of \eqref{e62} the infimum can be taken over $\ell$ which is a coupling of $(P_{Z_i})$. In particular, if $\ell$ is a probability measure, then $\Theta_j^*(\ell)=c_jD(T_j^*\ell\|\mu_j)$ still holds with the $\Theta_j$ defined in \eqref{e39}, $j=1,\dots,m$. Thus the rest of the proof can proceed as before.
\end{proof}
\begin{rem}
The second assumption is made in order to achieve \eqref{e65} in the proof.
\end{rem}
\begin{rem}\label{rem_rbl}
In \cite{ISIT_lccv2016} we studied a version of ``reverse Brascamp-Lieb inequality''
which is a special case of Theorem~\ref{thm_noncpt} when there is only one
\emph{forward channel}:
in the setting of Theorem~\ref{thm_unification} consider $m=1$, $c_1=1$, $\mathcal{X}=\mathcal{Z}_1\times,\dots,\times\mathcal{Z}_l$. Let $S_i$ be the canonical map, $g_i\leftarrow g_i^{\frac{1}{b_i}}$, $i=1,\dots,l$.
Then \eqref{e_frbl_func} becomes
\begin{align}
\prod_{i=1}^l\|g_i\|_{\frac{1}{b_i}}
\le
\exp(d)\mu_1(f_1)
\label{e_rbl}
\end{align}
for any nonnegative continuous $(g_i)_{i=1}^l$ and $f_1$ bounded away from $0$ and $+\infty$ such that
\begin{align}
\sum_{i=1}^l \log g_i(z_i)\le \mathbb{E}[\log f_1(Y_1)|Z^l=z^l],\quad\forall z^l.
\label{e_rbl_cond}
\end{align}
Note that \eqref{e_rbl_cond} can be simplified in the deterministic special case:
let $\phi\colon \mathcal{X}\to\mathcal{Y}_1$ be any continuous function, and $T_1\leftarrow C_b(\phi)$ (that is, $T_1$ sends a function $f$ on $\mathcal{Y}_1$ to the function $f\circ \phi$ on $\mathcal{X}$).
Then \eqref{e_rbl_cond} becomes
\begin{align}
\prod_{i=1}^l g_i(z_i)\le f_1(\phi(z_1,\dots,z_l)),\quad\forall z_1,\dots,z_l.
\end{align}
Then the optimal choice of $f_1$ admits an explicit formula,
since for any given $(g_i)$, to verify \eqref{e_rbl} we only need to consider
\begin{align}
f_1(y)
:= \sup_{\phi(z_1,\dots,z_l)=y}\left\{\prod_i g_i(z_i)\right\},
\quad\forall y.
\label{e_rv_f}
\end{align}
Thus when $\phi$ is a linear function, \eqref{e_rbl} is essentially Barthe's formulation of reverse BL \eqref{e_barthe} (the exception being that Theorem~\ref{thm_noncpt}, in contrast to \eqref{e_barthe}, restricts attention to finite $\nu_i$, $i=1,\dots,l$ and $\mu_1$).
The more straightforward part of the duality (entropic inequality$\Rightarrow$functional inequality) has essentially been proved by Lehec \cite[Theorem~18]{lehec2010representation} in a special setting.
\end{rem}

\subsection{Extension to General Convex Functionals}
For certain applications (e.g.\ the transportation-cost inequalities, see Section~\ref{sec_transportation} ahead),
we may be interested in convex functionals beyond the relative entropy.
Recall that given a lower semicontinuous, proper convex function $\Lambda(\cdot)$, its Legendre-Fenchel transform (Definition~\ref{defn_LFT}) is denoted as $\Lambda^*(\cdot)$.
From convex analysis (see for example \cite[Lemma~4.5.8]{dembo2009large}) we have
\begin{align}
\Lambda(u)=\sup_{\ell\in C_b(\mathcal{X})^*}[\ell(u)-\Lambda^*(\ell)],
\label{e75}
\end{align}
for any $u\in C_b(\mathcal{X})$. Moreover, if $\Lambda^*(\ell)=+\infty$ for any $\ell\notin \mathcal{P}(\mathcal{X})$, then from \eqref{e75} we must also have
\begin{align}
\Lambda(u)=\sup_{\ell\in \mathcal{P}(\mathcal{X})}[\ell(u)-\Lambda^*(\ell)].
\label{e76}
\end{align}
For example, the function $\Lambda$ defined in \eqref{e38} satisfies the property in \eqref{e76}.
We need this property in the proof of Theorem~\ref{thm_unification} because of step \eqref{e64}.
From the proof of Theorem~\ref{thm_unification} we see that we can obtain the following generalization to convex functionals with no additional cost. An application of this generalization to transportation-cost inequalities is given in Section~\ref{sec_transportation}.

\begin{thm}\label{thm_GeneralConvex}
Assume that
\begin{itemize}
\item $m$ and $l$ are positive integers, $d\in\mathbb{R}$, $\mathcal{X}$ is a compact metric space (hence also a Polish space);
\item For each $i=1,\dots,l$, $\mathcal{Z}_i$ is a Polish space, $\Lambda_i\colon C_b(\mathcal{Z}_i)\to\mathbb{R}\cup\{+\infty\}$ is proper convex such that $\Lambda_i^*(\ell)=+\infty$ for $\ell\notin \mathcal{P}(\mathcal{Z}_i)$, and $S_i\colon C_b(\mathcal{Z}_i)\mapsto C_b(\mathcal{X})$ is a conditional expectation operator;
\item For each $j=1,\dots,m$, $\mathcal{Y}_j$ is a Polish space, $\Theta_j\colon C_b(\mathcal{Y}_j)\to\mathbb{R}\cup\{+\infty\}$ is proper convex such that $\Theta_j(u)<\infty$ for some $u\in C_b(\mathcal{Y}_j)$ which is bounded below,
    and $T_j\colon C_b(\mathcal{Y}_j)\to C_b(\mathcal{X})$ is a conditional expectation operator;
\item For any $\ell_{Z_i}\in \mathcal{M}(\mathcal{Z}_i)$ such that $\Lambda_i^*(\ell_{Z_i})<\infty$, $i=1,\dots,l$, there exists $\ell_X\in\bigcap_i(S_i^*)^{-1}\ell_{Z_i}$ such that $\sum_{j=1}^m \Theta_j^*(\ell_{Y_j})< \infty$, where $\ell_{Y_j}:=T_j^*\ell_X$.
\end{itemize}
Then the following two statements are equivalent:
\begin{enumerate}
  \item If $g_i\in C_b(\mathcal{Z}_i)$, $f_j\in C_b(\mathcal{Y}_j)$, $i=1,\dots,l$, $j=1,\dots,m$ satisfy
      \begin{align}
      \sum_{i=1}^l S_ig_i\le \sum_{j=1}^m T_jf_j
      \label{e_GConvex1}
      \end{align}
      then
      \begin{align}
      \sum_{i=1}^l\Lambda_i(g_i)\le
      \sum_{j=1}^m\Theta_j(f_j).
      \label{e_GConvex_func}
      \end{align}
  \item For any\footnote{Since by assumption $\Lambda_i^*(\ell_{Z_i})=+\infty$ when $\ell\notin \mathcal{P}(\mathcal{Z}_i)$, in which case \eqref{e_GConvex_entropic} is trivially true, it is equivalent to assume here that $\ell\in \mathcal{P}(\mathcal{Z}_i)$.} $\ell_{Z_i}\in\mathcal{M}(\mathcal{Z}_i)$, $i=1,\dots,l$,
      \begin{align}
      \sum_{i=1}^l \Lambda_i^*(\ell_{Z_i})
      \ge \inf_{\ell_X}\sum_{j=1}^m \Theta_j^*(\ell_{Y_j})
      \label{e_GConvex_entropic}
      \end{align}
      where the infimum is over $\ell_X$ such that $S_i^*\ell_X=P_{Z_i}$, $i=1,\dots,l$, and $\ell_{Y_j}:=T_j^*\ell_X$, $j=1,\dots,m$.
\end{enumerate}
\end{thm}
\begin{rem}\label{rem_ncpt}
Just like Theorem~\ref{thm_noncpt}, it is possible to extend Theorem~\ref{thm_GeneralConvex} to the case of noncompact $\mathcal{X}$, provided that $\mathcal{X}=\mathcal{Z}_1\times,\dots,\times\mathcal{Z}_l$ and $S_i$, $i=1,\dots,l$ are canonical maps.
\end{rem}

\section{Some Special Cases of the Forward-Reverse Brascamp-Lieb Inequality}\label{sec_special}
In this section we discuss some notable special cases of the duality results for the forward-reverse Brascamp-Lieb Inequality
(Theorems~\ref{thm_1}-\ref{thm_GeneralConvex}).
Some of these special cases have been noticed in the literature (some proved using different methods that crucially rely on the finiteness of the alphabet).
\subsection{Variational Formula of R\'{e}nyi Divergence}
\label{sec_lpcb}
As the first example, we show how \eqref{e_func} recovers the variational formula of R\'{e}nyi divergence \cite{dvijotham} \cite{atarRobust} in a special case.
A prototype of the variational formula of R\'{e}nyi divergence appeared in the context of control theory \cite{dvijotham} as a technical lemma. Its utility in information theory was then noticed by \cite{atarRobust} \cite{atar2014information}, which further developed the result and elaborated on its applications in other areas of probability theory. 
Suppose $R$ and $Q$ are nonnegative measures on $\mathcal{X}$, $\alpha\in(0,1)\cup(1,\infty)$, and $g\colon\mathcal{X}\to\mathbb{R}$ is a bounded measurable function. Also, let $T$ be a probability measure such that $R,Q\ll T$. Define the R\'{e}nyi divergence
\begin{align}
D_{\alpha}(Q\|R)
:=\frac{1}{\alpha-1}\log\left(\mathbb{E}\left[\exp\left(
\alpha\imath_{Q\|T}(\bar{X})+
(1-\alpha)\imath_{R\|T}(\bar{X})\right)\right]\right)
\end{align}
where $\bar{X}\sim T$,
which is independent of the particular choice of the reference measure $T$ \cite{erven2014}.
Then the variational formula of R\'{e}nyi divergence \cite[Remark~2.2]{atarRobust} can be equivalently stated as the functional inequality\footnote{Note that our definition of R\'{e}nyi divergence is different from \cite{atarRobust} by a factor of $\alpha$.}
\begin{align}
\frac{1}{\alpha-1}\log\mathbb{E}[ \exp((\alpha-1)g(\hat{X}))]-
\frac{1}{\alpha}\log\mathbb{E}[\exp(\alpha g(X))]
\le \frac{1}{\alpha}D_{\alpha}(Q\|{R})
\label{e_62}
\end{align}
where $X\sim {R}$ and $\hat{X}\sim{Q}$,
with equality achieved when
\begin{align}
\frac{{\rm d}{Q}}{{\rm d}{R}}(x)=\frac{\exp( g(x))}{\mathbb{E}[\exp( g(X))]}.
\end{align}
The well-known variational formula of the relative entropy (see e.g.~\cite{verdubook}) can be recovered by taking $\alpha\to1$.
In the $\alpha\in(1,\infty)$ case,
we can choose $\exp(g)$ to be the indicator function of an arbitrary measurable set $\mathcal{A}\subseteq\mathcal{X}$, to obtain the \emph{logarithmic probability comparison bound} (LPCB) \cite{atarRobust}\footnote{\cite{atarRobust} focuses on the case of probability measure, but \eqref{e_lpcb} continues to hold if ${Q}$ and ${R}$ are replaced by any unnormalized nonnegative measures.}
\begin{align}
\frac{1}{\alpha-1}\log{Q}(\mathcal{A})
-\frac{1}{\alpha}\log{R}(\mathcal{A})
\le \frac{1}{\alpha}D_{\alpha}({Q}\|{R}).
\label{e_lpcb}
\end{align}

Now we give a new proof of the functional inequality \eqref{e_62} using a well-known entropic inequality in information theory.
First consider $\alpha\in(1,\infty)$.
In Theorem~\ref{thm_1}, set $m\leftarrow 1$, $\mathcal{Y}_1=\mathcal{X}$, $c\leftarrow \frac{\alpha-1}{\alpha}$, $P_{Y_1|X}={\sf id}$ (the identity mapping).
We may assume without loss of generality that ${Q}\ll{R}$, since otherwise $D_{\alpha}({Q}\|{R})=\infty$ and \eqref{e_62} always holds. Thus, setting the cost function as
\begin{align}
d(x)\leftarrow
-\imath_{{Q}\|{R}}(x)
+\frac{\alpha-1}{\alpha}D_{\alpha}({Q}\|{R})
\label{e_d}
\end{align}
we see that \eqref{e_info} is reduced to
\begin{align}
D(P\|{R})+\mathbb{E}\left[
-\imath_{{Q}\|{R}}(\hat{X})\right]
+\frac{\alpha-1}{\alpha}D_{\alpha}({Q}\|{R})
\ge\frac{\alpha-1}{\alpha}D(P\|{R})
\label{e_64}
\end{align}
which, by our convention in Remark~\ref{rem6},
can be simplified to
\begin{align}
D(P\|{Q})+\frac{\alpha-1}{\alpha}D_{\alpha}({Q}\|{R})
\ge \frac{\alpha-1}{\alpha}D(P\|{R}).
\label{e_65}
\end{align}
It is a well-known result that \eqref{e_65} holds for all $P$ absolutely continuous with respect to ${Q}$ and ${R}$ (see for example \cite[Theorem~30]{erven2014}, \cite[Theorem~1]{shayevitz2010}, \cite[Corollary~2]{sason2015}), due to its relation to the fundamental problem of characterizing the error exponents in  binary hypothesis testing. By Theorem~\ref{thm_1} with $m=1$ we translate \eqref{e_64} into the functional inequality:
\begin{align}
\mathbb{E}\left[\exp\left(\log f(\hat{X})
+\imath_{{Q}\|{R}}(\hat{X})-\frac{\alpha-1}{\alpha}D_{\alpha}
({Q}\|{R})\right)\right]
\le\left(\mathbb{E}\left[ f^{\frac{\alpha}{\alpha-1}}(X)\right]\right)
^{\frac{\alpha-1}{\alpha}}
\end{align}
for all nonnegative measurable $f$. Finally, by taking the logarithms and dividing by $\alpha-1$ on  both sides and setting
$f\leftarrow \exp((\alpha-1)g)$, the functional inequality \eqref{e_62}
is recovered.

Note that the choice of $d(\cdot)$ in \eqref{e_d} is simply for the purpose of change-of-measure, since the two relative entropy terms in \eqref{e_65} have different reference measures $Q$ and $R$.
Thus, an alternative proof is to take $d(\cdot)=0$ but invoke the extension of Theorem~\ref{thm_1} in Remark~\ref{thm_1_ext} with $\mu\leftarrow Q$ and $\nu\leftarrow R$.

The case of $\alpha\in(0,1)$ can be proved in a similar fashion using Theorem~\ref{thm_unification} with $m\leftarrow2$, $l\leftarrow 1$, $\mathcal{X}=\mathcal{Y}_1=\mathcal{Y}_2\leftarrow \mathcal{X}$, $Q_{Y_1|X}=Q_{Y_2|X}={\sf id}$, $Q_{Y_1}\leftarrow Q$, $Q_{Y_2}\leftarrow R$, and $|\mathcal{Z}|=1$; we omit the details here.

Note that the original proofs of the functional inequality \eqref{e_62} in \cite{dvijotham}\cite{atarRobust} were based on H\"{o}lder's inequality, whereas the present proof relies on the duality between functional inequalities and entropic inequalities, and the property \eqref{e_65} (which amounts to the nonnegativity of relative entropy). A third proof of \eqref{e_62} based on the non-negativity of R\'{e}nyi divergence by the first and fourth named authors will be given in \cite{verdubook}. Moreover, the weaker probability version
\eqref{e_lpcb} can be easily proved by a data processing argument; see for example \cite[Section~II.B]{polyanskiy2010arimoto}\cite{sasonverdu2017}.

\subsection{Strong Data Processing Constant}
The strong data processing inequality (SDPI) \cite{ahlswede1976spreading}\cite{csiszar2011information}\cite{anan_13}
has received considerable interests recently. It has been proved fruitful in providing impossibility bounds in various problems; see \cite{pw_2015} for a recent list of its applications.
It generally refers to an inequality of the form
\begin{align}
D(P_X\|Q_X)\ge c D(P_Y\|Q_Y),\quad \textrm{for all $P_X\ll Q_X$}
\label{e_sdp}
\end{align}
where $P_X\to Q_{Y|X}\to P_Y$, and we have fixed $Q_{XY}=Q_XQ_{Y|X}$. The conventional data processing inequality corresponds to the case of $c=1$. The study of the best (largest) constant $c$ for \eqref{e_sdp} to hold can be traced to Ahlswede and G\'{a}cs \cite{ahlswede1976spreading}, who
showed, among other things, its equivalence to the functional inequality
\begin{align}
\mathbb{E}[\exp(\mathbb{E}[\log f(Y)|X])]\le \|f\|_{\frac{1}{c}}.\quad\textrm{for all nonnegative $f$}.
\label{e_sdpi_func}
\end{align}
The strong data processing inequality can be viewed as a special case of the forward-reverse Brascamp-Lieb inequality where there is only one forward and one identity reverse channel.
In other words, the equivalence between \eqref{e_sdp}
and \eqref{e_sdpi_func}
can be readily seen from either Theorem~\ref{thm_1} or Remark~\ref{rem_rbl}.
Its original proof of such an equivalence \cite[Theorem~5]{ahlswede1976spreading}, on the other hand, relies on a limiting property of hypercontractivity, which relies heavily on the finiteness of the alphabet and the proof is quite technical even in that case.

As we saw in Section~\ref{sec_lpcb}, a functional inequality often implies an inequality of the probabilities of sets when specialized to the indicator functions. In the case of \eqref{e_sdpi_func}, however, a more rational choice is
\begin{align}
f(y)=(1_{\mathcal{A}}(y)+ Q_Y(\mathcal{A}) 1_{\mathcal{\bar{A}}}(y))^c
\end{align}
where $\mathcal{A}$ is an arbitrary measurable subset of $\mathcal{Y}$ and $\mathcal{\bar{A}}:=\mathcal{Y}\setminus \mathcal{A}$. Then using \eqref{e_sdpi_func},
\begin{align}
Q_Y(\mathcal{A})^{c\epsilon}Q_X(x\colon Q_{Y|X=x}(\mathcal{A})\ge 1-\epsilon)
&=
Q_Y(\mathcal{A})^{c\epsilon}Q_X(x\colon Q_{Y|X=x}(\mathcal{\bar{A}})\le \epsilon)
\\
&\le \int\exp\left(
\log Q_Y(\mathcal{A})^c \cdot Q_{Y|X=x}(\mathcal{\bar{A}})\right){\rm d}Q_X(x)
\\
&=\mathbb{E}[\exp(\mathbb{E}[\log f(Y)|X])]
\\
&\le \mathbb{E}^c[f^{\frac{1}{c}}(Y)]
\\
&\le 2^cQ_Y^{c}(\mathcal{A}).
\end{align}
Rearranging, we obtain the following bound on conditional probabilities:
\begin{align}
Q_X(x\colon Q_{Y|X=x}(\mathcal{A})\ge 1-\epsilon)\le 2^c Q_Y^{c(1-\epsilon)}(\mathcal{A})
\label{e54}
\end{align}
which, by a blowing-lemma argument (cf.~\cite{ahlswede1976bounds}), would imply the asymptotic result of \cite[Theorem~1]{ahlswede1976bounds}, a useful tool in establishing strong converses in source coding problems. Note that \cite[Section~2]{ahlswede1976bounds} proved a result essentially the same as \eqref{e54} by working on \eqref{e_sdp} rather than \eqref{e_sdpi_func}.\footnote{Another difference is that \cite[Theorem~1]{ahlswede1976bounds} involves an auxiliary r.v.~$U$ with $|\mathcal{U}|\le 3$, where the cardinality bound comes from convexifying a subset in $\mathbb{R}^2$. Here \eqref{e54} holds if \eqref{e_sdp}, which is slightly simpler involving only relative entropy terms, because we are essentially working with the supporting lines of the convex hull,
and the supporting line of a set is the same as the supporting line of its convex hull.}

%



\subsection{Loomis-Whitney Inequality and Shearer's Lemma}
The duality between Loomis-Whitney Inequality and Shearer's Lemma is yet another special case of Theorem~\ref{thm_1}. This is already contained in the duality theorem of Carlen and Cordero-Erausquin \cite{carlen2009subadditivity}, but we briefly discuss it here.

The combinatorial Loomis-Whitney inequality \cite[Theorem~2]{loomis1949} says that if $A$
is a subset of $\mathcal{A}^m$, where $\mathcal{A}$ is a finite or countably infinite set, then
\begin{align}
|A|\le \prod_{j=1}^m |\pi_j(A)|^{\frac{1}{m-1}}
\label{e_clw}
\end{align}
where we defined the projection
\begin{align}
\pi_j\colon \mathcal{A}^m&\to \mathcal{A}^{m-1},
\label{e_proj}
\\
(x_1,\dots,x_m)&\mapsto (x_1,\dots,x_{j-1},x_{j+1},\dots,x_m).
\label{e_proj2}
\end{align}
for each $j\in\{1,\dots,m\}$.
The combinatorial inequality \eqref{e_clw} can be recovered from the following integral inequality: let $\mu$ be the counting measure on $\mathcal{A}^m$, then
\begin{align}
\int_{\mathcal{A}^m}\prod_{j=1}^m f_{j}(\pi_j(x))
{\rm d}\mu(x)
\le \prod_{j=1}^m \|f_{j}\|_{m-1}
\label{e_ilw}
\end{align}
for all nonnegative $f_j$'s, where the norm on the right side is with respect to the counting measure on $\mathcal{A}^{m-1}$.
This is an inequality of the form \eqref{e_func}.
To see how \eqref{e_ilw} recovers \eqref{e_clw},
let $f_j$ be the indicator function of $\pi_j(A)$ for each $j$.
Then the left side of \eqref{e_ilw} upper-bounds the left side of \eqref{e_clw},
while the right side of \eqref{e_ilw} is equal to the right side of \eqref{e_clw}.
Now we invoke Remark~\ref{thm_1_ext} with $X\leftarrow\mathcal{A}^m$, $\nu$ and $\mu_j$ being the counting measure on $\mathcal{A}^m$ and $\mathcal{A}^{m-1}$, respectively, $Q_{Y_j|X}$ being the projection mappings in \eqref{e_proj}-\eqref{e_proj2}, and let $(X_1,\dots,X_m)$ be distributed according to a given $P_X$, to obtain
\begin{align}
-H(X_1,\dots,X_m)&=D(P_X\|\nu)
\\
&\ge \sum_{j=1}^m\frac{1}{m-1}D(P_{Y_j}\|\mu_j)
\\
&\ge -\sum_{j=1}^m\frac{1}{m-1}
H(X_1,\dots,X_{j-1},X_{j+1},\dots,X_{m})
\end{align}
where $H(\cdot)$ is the Shannon entropy.
This is Shearer's Lemma \cite{shearer} \cite{shearer2} when the cardinality of the subset is one less than the cardinality of the whole set of random variables.

Similarly, the continuous Loomis-Whitney inequality for Lebesgue measure, that is,
\begin{align}
\int_{\mathbb{R}^m}\prod_{j=1}^m f_{j}(\pi_j(x^n))
{\rm d}x^n
\le \prod_{j=1}^m \|f_{j}\|_{m-1}
\label{e_lwc}
\end{align}
is the dual of a continuous version of Shearer's lemma involving differential entropies:
\begin{align}
h(X^m)\le \sum_{j=1}^m\frac{1}{m-1}h(X_1,\dots,X_{j-1},X_{j+1},\dots,X_m).
\end{align}

\subsection{Hypercontractivity}\label{sec_hc}
\begin{figure}[h!]
  \centering
\begin{tikzpicture}
[scale=2,
      dot/.style={draw,fill=black,circle,minimum size=0.7mm,inner sep=0pt},arw/.style={->,>=stealth}]
    \node[rectangle] (X) {$\mathcal{P}(\mathcal{Y}_1\times\mathcal{Y}_2)$};
    \node[rectangle] (Z1) [left=of X, xshift=0mm, yshift=0mm] {$\mathcal{P}(\mathcal{Z}_1)$};
    \node[rectangle] (Y1) [above right=of X, xshift=0mm, yshift=-10mm] {$\mathcal{P}(\mathcal{Y}_1)$};
    \node[rectangle] (Y2) [below right=of X, xshift=0mm, yshift=10mm] {$\mathcal{P}(\mathcal{Y}_2)$};

  \draw [arw] (X) to node[midway,above]{$\cong$} (Z1);
  \draw [arw] (X) to node[midway,above]{$T_1^*$} (Y1);
  \draw [arw] (X) to node[midway,below]{$T_2^*$} (Y2);
\end{tikzpicture}
\caption{Diagram for hypercontractivity (HC)}
\end{figure}
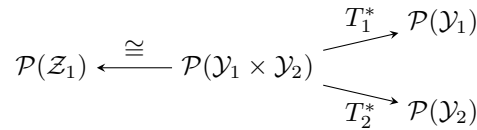
Fix a joint probability distribution $Q_{Y_1Y_2}$ and nonnegative continuous functions $F_1$ and $F_2$ on $\mathcal{Y}_1$ and $\mathcal{Y}_2$, respectively,
both bounded away from 0.
In Theorem~\ref{thm_unification}, take $l\leftarrow1$, $m\leftarrow2$, $b_1\leftarrow1$, $d\leftarrow0$, $f_1\leftarrow F_1^{\frac{1}{c_1}}$, $f_2\leftarrow F_2^{\frac{1}{c_2}}$, $\nu_1\leftarrow Q_{Y_1Y_2}$, $\mu_1\leftarrow Q_{Y_1}$, $\mu_2\leftarrow Q_{Y_2}$.
Also, put $Z_1=X=(Y_1,Y_2)$, and let $T_1$ and $T_2$ be the canonical maps (Definition~\ref{exp_defn1}).
The constraint \eqref{e34} translates to
\begin{align}
g_1(y_1,y_2)\le F_1(y_1)F_2(y_2),\quad\forall y_1,y_2
\end{align}
and the optimal choice of $g_1$ is when the equality is achieved.
We thus obtain the equivalence between\footnote{By a standard dense-subspace argument, we see that it is inconsequential that $F_1$ and $F_2$ in \eqref{e_hc} are not assumed to be continuous nor bounded away from zero.
It is also easy to see that the nonnegativity of $F_1$ and $F_2$ is inconsequential for \eqref{e_hc}.}
\begin{align}
\|F_1\|_{\frac{1}{c_1}} \|F_2\|_{\frac{1}{c_2}}
\ge
\mathbb{E}[F_1(Y_1)F_2(Y_2)],\quad \forall F_1\in L^{\frac{1}{c_1}}(Q_{Y_1}),\,
F_2\in L^{\frac{1}{c_2}}(Q_{Y_2})
\label{e_hc}
\end{align}
and
\begin{align}
\forall P_{Y_1Y_2},\quad
D(P_{Y_1Y_2}\|Q_{Y_1Y_2}) \ge c_1D(P_{Y_1}\|Q_{Y_1})
+ c_2D(P_{Y_2}\|Q_{Y_2}).
\end{align}
This equivalence can also be obtained from Theorem~\ref{thm_1}.
By H\"{o}lder's inequality, \eqref{e_hc} is equivalent to saying that the norm of the linear operator sending $F_1\in L^{\frac{1}{c_1}}(Q_{Y_1})$ to $\mathbb{E}[F_1(Y_1)|Y_2=\cdot]\in L^{\frac{1}{1-c_2}}(Q_{Y_2})$ does not exceed 1.
The interesting case is $\frac{1}{1-c_2}>\frac{1}{c_1}$, hence the name hypercontractivity.
The equivalent formulation of hypercontractivity was shown in \cite{nair} using a different proof via the method of types/typicality, which relies on the finite nature of the alphabet. In contrast, the proof based on the nonnegativity of relative entropy removes this constraint, allowing one to prove Nelson's Gaussian hypercontractivity from the information-theoretic formulation (see Section~\ref{sec_hypercontractivity}).

\subsection{Reverse Hypercontractivity (Positive Parameters\protect\footnote{By ``positive parameters'' we mean the $b_1$ and $b_2$ in \eqref{e_rhc_entropic} are positive.})}\label{sec_rhc}
\begin{figure}[h!]
  \centering
\begin{tikzpicture}
[scale=2,
      dot/.style={draw,fill=black,circle,minimum size=0.7mm,inner sep=0pt},arw/.style={->,>=stealth}]
    \node[rectangle] (X) {$\mathcal{P}(\mathcal{Z}_1\times\mathcal{Z}_2)$};
    \node[rectangle] (Z1) [above left=of X, xshift=0mm, yshift=-10mm] {$\mathcal{P}(\mathcal{Z}_1)$};
    \node[rectangle] (Z2) [below left=of X, xshift=0mm, yshift=10mm] {$\mathcal{P}(\mathcal{Z}_2)$};
    \node[rectangle] (Y1) [right=0.5cm of X, xshift=0mm, yshift=0mm] {$\mathcal{P}(\mathcal{Y}_1)$};

  \draw [arw] (X) to node[midway,above]{$S_1^*$} (Z1);
  \draw [arw] (X) to node[midway,below]{$S_2^*$} (Z2);
  \draw [arw] (X) to node[midway,above]{$\cong$} (Y1);
\end{tikzpicture}
\caption{Diagram for reverse HC}
\end{figure}
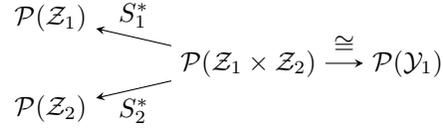
Let $Q_{Z_1Z_2}$ be a given joint probability distribution, and let $G_1$ and $G_2$ be nonnegative functions on $\mathcal{Z}_1$ and $\mathcal{Z}_2$, respectively, both bounded away from 0.
In Theorem~\ref{thm_unification}, take $l\leftarrow2$, $m\leftarrow1$, $c_1\leftarrow1$, $d\leftarrow0$, $g_1\leftarrow G_1^{\frac{1}{b_1}}$, $g_2\leftarrow G_2^{\frac{1}{b_2}}$, $\mu_1\leftarrow Q_{Z_1Z_2}$, $\nu_1\leftarrow Q_{Z_1}$, $\nu_2\leftarrow Q_{Z_2}$.
Also, put $Y_1=X=(Z_1,Z_2)$, and let $S_1$ and $S_2$ be the canonical maps (Definition~\ref{exp_defn1}).
Note that the constraint
\eqref{e34}
translates to
\begin{align}
f_1(z_1,z_2)\ge G_1(z_1)G_2(z_2),
\quad \forall z_1,z_2.
\end{align}
and the equality case yields the optimal choice of $f_1$ for \eqref{e_frbl_func}.
By Theorem~\ref{thm_unification} we thus obtain the equivalence between
\begin{align}
\|G_1\|_{\frac{1}{b_1}} \|G_2\|_{\frac{1}{b_2}}
\le
\mathbb{E}[G_1(Z_1)G_2(Z_2)],\quad \forall G_1,G_2
\label{e_rhc}
\end{align}
and
\begin{align}
\forall P_{Z_1},P_{Z_2},\,\exists P_{Z_1Z_2},\quad
D(P_{Z_1Z_2}\|Q_{Z_1Z_2}) \le b_1D(P_{Z_1}\|Q_{Z_1})
+ b_2D(P_{Z_2}\|Q_{Z_2}).
\label{e_rhc_entropic}
\end{align}
Note that in this set-up, if $\mathcal{Z}_1$ and $\mathcal{Z}_2$ are finite, then the condition in the last bullet in Theorem~\ref{thm_unification} is equivalent to $Q_{Z_1Z_2}\ll Q_{Z_1}\times Q_{Z_2}$.
The equivalent formulations of reverse hypercontractivity were observed in \cite{kamath_reverse}, where the proof is based on the method of types argument.

\subsection{Reverse Hypercontractivity (One Negative Parameter\protect\footnote{By ``one negative parameter'' we mean the $b_1$ is positive and $-c_2$ is negative in \eqref{e_rhcn_entropic}.})}\label{sec_rhcn}
\begin{figure}[h!]
  \centering
\begin{tikzpicture}
[scale=2,
      dot/.style={draw,fill=black,circle,minimum size=0.7mm,inner sep=0pt},arw/.style={->,>=stealth}]
    \node[rectangle] (X) {$\mathcal{P}(\mathcal{Z}_1\times\mathcal{Y}_2)$};
    \node[rectangle] (Z1) [left=of X, xshift=0mm, yshift=0mm] {$\mathcal{P}(\mathcal{Z}_1)$};
    \node[rectangle] (Y1) [above right=of X, xshift=0mm, yshift=-10mm] {$\mathcal{P}(\mathcal{Y}_1)$};
    \node[rectangle] (Y2) [below right=of X, xshift=0mm, yshift=10mm] {$\mathcal{P}(\mathcal{Y}_2)$};

  \draw [arw] (X) to node[midway,above]{$S_1^*$} (Z1);
  \draw [arw] (X) to node[midway,above]{$\cong$} (Y1);
  \draw [arw] (X) to node[midway,below]{$T_2^*$} (Y2);
\end{tikzpicture}
\caption{Diagram for reverse HC with one negative parameter}
\end{figure}
In Theorem~\ref{thm_unification}, take $l\leftarrow1$, $m\leftarrow2$, $c_1\leftarrow1$, $d\leftarrow0$.
Let $Y_1=X=(Z_1,Y_2)$, and let $S_1$ and $T_2$ be the canonical maps (Definition~\ref{exp_defn1}).
Suppose that $Q_{Z_1Y_2}$ is a given joint probability distribution, and set $\mu_1\leftarrow Q_{Z_1Y_2}$, $\nu_1\leftarrow Q_{Z_1}$, $\mu_2\leftarrow Q_{Y_2}$
in Theorem~\ref{thm_unification}.
Suppose that $F$ and $G$ be arbitrary nonnegative continuous functions on $\mathcal{Y}_2$ and $\mathcal{Z}_1$, respectively,
which are bounded away from $0$.
Take $g_1\leftarrow G^{\frac{1}{b_1}}$,
$f_2\leftarrow F^{-\frac{1}{c_2}}$.
in Theorem~\ref{thm_unification}.
The constraint \eqref{e34}
translates to
\begin{align}
f_1(z_1,y_2)\ge G(z_1)F(y_2),\quad\forall z_1,y_2.
\label{e_constraint3}
\end{align}
Note that
\eqref{e_frbl_func}
translates to
\begin{align}
\|G\|_{\frac{1}{b_1}}\le Q_{Y_2Z_1}(f_1)Q_{Y_2}^{c_2}(F^{-\frac{1}{c_2}})
\label{e_nfunc}
\end{align}
for all $F$, $G$, and $f_1$ satisfying \eqref{e_constraint3}.
It suffices to verify \eqref{e_nfunc} for the optimal choice $f_1=GF$, so \eqref{e_nfunc} is reduced to
\begin{align}
\|F\|_{\frac{1}{-c_2}} \|G\|_{\frac{1}{b_1}}
\le
\mathbb{E}[F(Y_2)G(Z_1)],\quad \forall F,G.
\label{e_rhcn}
\end{align}
By Theorem~\ref{thm_unification}, \eqref{e_rhcn} is equivalent to
\begin{align}
\forall P_{Z_1},\,\exists P_{Z_1Y_2},\quad
D(P_{Z_1Y_2}\|Q_{Z_1Y_2}) \le b_1D(P_{Z_1}\|Q_{Z_1})
+ (-c_2)D(P_{Y_2}\|Q_{Y_2}).
\label{e_rhcn_entropic}
\end{align}
Inequality \eqref{e_rhcn} is called reverse hypercontractivity with a negative parameter in \cite{beigi2016}, where the entropic version
\eqref{e_rhcn_entropic}
is established for finite alphabets using the method of types.
Multiterminal extensions of \eqref{e_rhcn} and \eqref{e_rhcn_entropic} (called reverse Brascamp-Lieb type inequality with negative parameters in \cite{beigi2016}) can also be recovered from Theorem~\ref{thm_unification} in the same fashion, i.e., we move all negative parameters to the other side of the inequality so that all parameters become positive.

In summary, from the viewpoint of Theorem~\ref{thm_unification},
the results in~\ref{sec_hc},\ref{sec_rhc} and \ref{sec_rhcn} are degenerate special cases, in the sense that in any of the three cases the optimal choice of one of the functions in \eqref{e_frbl_func} can be explicitly expressed in terms of the other functions, hence this ``hidden function'' disappears in \eqref{e_hc}, \eqref{e_rhc} or \eqref{e_rhcn}.

\subsection{Transportation-Cost Inequalities}\label{sec_transportation}
\begin{defn}[see for example \cite{villani2008optimal}]\label{defn_tp}
We say that a probability measure $Q$ on a metric space $(\mathcal{Z},d)$ satisfies ${\rm T}_p(\lambda)$ inequality, $p\in[1,\infty)$, $\lambda\in(0,\infty)$, if
\begin{align}
\inf_{\pi}\mathbb{E}^{\frac{1}{p}}[d^p(X,Y)]\le \sqrt{2\lambda D(P\|Q)}
\label{e_tp_defn}
\end{align}
for every $P\ll Q$, where the infimum is over all coupling $\pi$ of $P$ and $Q$, and $(X,Y)\sim \pi$.
It suffices to focus on the case of $\lambda=1$,
since results for general $\lambda\in (0,\infty)$ can
usually be obtained by a scaling argument.
\end{defn}
As a consequence of Theorem~\ref{thm_GeneralConvex}
and Remark~\ref{rem_ncpt},
we have
\begin{cor}\label{cor_tp}
Let $(\mathcal{Z},d)$ be a locally compact, $\sigma$-compact Polish space.
\begin{description}
  \item[(a)]A probability measure $Q$ on $\mathcal{Z}$ satisfies ${\rm T}_2(1)$ inequality if and only if for any $f\in C_b(\mathcal{Z})$,
      \begin{align}
      \log Q\left(\exp\left(\inf_{z\in \mathcal{Z}}\left[f(z)+\frac{d^2(\cdot,z)}{2}
      \right]\right)\right)
      \le
      Q(f).
      \label{e_t2}
      \end{align}
  \item[(b)]A probability measure $Q$ on $\mathcal{Z}$ satisfies ${\rm T}_p(1)$ inequality, $p\in [1,2)$, if and only if:
      \begin{align}
      \log Q\left(t\inf_{z\in\mathcal{Z}}\left[f(z)
      +\frac{d^p(\cdot,z)}{p}\right]\right)
      \le \left(\frac{1}{p}-\frac{1}{2}\right)t^{\frac{2}{2-p}}
      +tQ(f),\quad \forall t\in [0,\infty), \, f\in C_b(\mathcal{Z}).
      \label{e_tp}
      \end{align}
\end{description}
\end{cor}
\begin{proof}
\begin{description}
\item[(a)]In Theorem~\ref{thm_GeneralConvex}, put $l=2$, $m=1$, $\mathcal{Z}_1=\mathcal{Z}_2\leftarrow\mathcal{Z}$,
$\mathcal{X}=\mathcal{Y}_1\leftarrow\mathcal{Z}\times\mathcal{Z}$,
and
\begin{align}
\Lambda_1(u)&:=2\log Q\left(\exp\left(\frac{u}{2}\right)\right);
\\
\Lambda_2(u)&:=Q(u);
\\
\Theta_1(u)&:=
\left\{
\begin{array}{cc}
  0 & u\le d^2; \\
  +\infty & \textrm{otherwise}.
\end{array}
\right.
\end{align}
For $\ell\in \mathcal{M}(\mathcal{X})$, we can compute
\begin{align}
\Lambda_1^*(\ell)&=2D(\ell\|Q);
\\
\Lambda_2^*(\ell)&:=
\left\{\begin{array}{cc}
         0 & \ell=Q; \\
         +\infty & \textrm{otherwise}  ;
       \end{array}
\right.
\label{e_119}
\\
\Theta_1^*(\ell)&=
\left\{
\begin{array}{cc}
  \ell(d^2) & \ell\ge 0; \\
  +\infty & \textrm{otherwise}.
\end{array}
\right.
\end{align}
We also have $\Lambda_i^*(\ell)=+\infty$ for any $\ell\notin \mathcal{P}(\mathcal{X})$, $i=1,2$.
Thus by Theorem~\ref{thm_GeneralConvex}, \eqref{e_t2} is equivalent to the following: for any $f_1\in C_b(\mathcal{Z}_1\times \mathcal{Z}_2)$, $g_1\in C_b(\mathcal{Z}_1)$, $g_2\in C_b(\mathcal{Z}_2)$ such that
$
g_1+g_2\le f_1
$,
it holds that
\begin{align}
\Lambda_1(g_1)+\Lambda_2(g_2)\le \Theta_1(f_1).
\end{align}
By the monotonicity of $\Lambda_1$, this is equivalent to
\begin{align}
\Lambda_1\left(\inf_z[d^2(\cdot,z)-g_2(z)]\right)+\Lambda_2(g_2)\le 0
\end{align}
for any $g_2\in C_b(\mathcal{Z})$,
which is the same as \eqref{e_t2}.

\item[(b)]The proof is similar to Part~(a), except that we now pick
    \begin{align}
    \Theta_1(f):=
    \left\{
    \begin{array}{cc}
      2^{-\frac{2}{2-p}}p^{\frac{p}{2-p}}(2-p)
      \sup^{\frac{2}{2-p}}\left(\frac{f}{d^p}\right) & \textrm{ if }\sup f\ge 0; \\
      0 & \textrm{otherwise},
    \end{array}
    \right.
    \end{align}
    so that for any $\ell\ge0$,
    \begin{align}
    \Theta_1^*(\ell)=[\ell(d^p)]^{\frac{2}{p}}.
    \end{align}
\end{description}
\end{proof}
\begin{rem}
Actually, the proof of Corollary~\ref{cor_tp} does not use the assumption that $d$ is a metric (other than that it is a continuous function which is bounded below). The equivalent formulation of ${\rm T}_1$ inequality (special case of \eqref{e_tp}) was known to Rachev \cite{rachev1991} and Bobkov and G\"otze \cite{bobkov1999} (who actually slightly simplified the formula using the fact that $d$ is a metric).
The equivalent formulation of ${\rm T}_2$ inequality in \eqref{e_t2} also appeared in \cite{bobkov1999}, and was employed in \cite{bobkov2000}\cite{bobkov2001} to show a connection to the logarithmic Sobolev inequality.
The equivalent formulation of ${\rm T}_p$ inequality, $p\in [1,2)$ in \eqref{e_tp} appeared in \cite[Proposition~22.3]{villani2008optimal}.
\end{rem}
Transportation-cost inequalities have been fruitful in obtaining measure concentration results (since \cite{marton1986simple}\cite{marton1996bounding}).
We discuss more on ${\rm T}_2$ inequalities in the Gaussian case in Section~\ref{sec_talagrand}.

\section{Data Processing, Tensorization and Convexity}
\label{sec_ele}
Given $Q_X$ and $(Q_{Y_j|X})$, denote by $\BL(Q_X,(Q_{Y_j|X}))$ the set of $(d,(c_j))$ in Theorem~\ref{thm_1} (forward Brascamp-Lieb inequality) such that either \eqref{e_func} or \eqref{e_info} holds. In this section we show that some elementary properties of $\BL(Q_X,(Q_{Y_j|X}))$ follows conveniently from the information-theoretic characterization \eqref{e_info}.

\subsection{Data Processing}
Loosely speaking, the set $\BL(Q_X,(Q_{Y_j|X}))$ characterizes the level of ``uncorrelatedness'' between $X$ and $(Y_1,\dots,Y_m)$. The following data processing property captures this intuition:
\begin{prop}
\begin{enumerate}
  \item Given $Q_W$, $Q_{X|W}$ and $(Q_{Y_j|X})_{j=1}^m$, assume that $Q_{WXY_j}=Q_WQ_{X|W}Q_{Y_j|X}$ for each $j$.
      If $(0,(c_j))\in\BL(Q_X,(Q_{Y_j|X}))$, then $(0,(c_j))\in\BL(Q_W,(Q_{Y_j|W}))$.
  \item Given $Q_X$, $(Q_{Y_j|X})_{j=1}^m$ and $(Q_{Z_j|Y_j})_{j=1}^m$, assume that $Q_{XY_jZ_j}=Q_XQ_{Y_j|X}Q_{Z_j|Y_j}$ for each $j$. Then
      $\BL(Q_X,(Q_{Y_j|X}))\subset\BL(Q_X,(Q_{Z_j|X}))$.
\end{enumerate}
\end{prop}
The proof is omitted since it follows immediately from the monotonicity of the relative entropy and \eqref{e_info}.

\subsection{Tensorization}
The term "tensorization" refers to the phenomenon of additivity/multiplicativity in certain functional inequalities under tensor products.
In information theory this is a central feature of many converse proofs, and is closely related to the fact that some operational problems admit single-letter solutions. In functional analysis, this provides a ``particularly cute'' \cite{tao_tens} tool for proving many inequalities in arbitrary dimensions. As a close example, Lieb's proof \cite{lieb1990gaussian} of the Brascamp Lieb inequality relies on a special case of Proposition~\ref{prop_tens} below, where the proof uses the (functional version of) Brascamp-Lieb inequality and the Minkowski inequality. The original proof of Brascamp-Lieb inequality \cite{brascamp1976best} is also based on a tensor power construction.
\begin{prop}\label{prop_tens}
Suppose $(d^{(i)},(c_j))\in
\BL(Q_X^{(i)},(Q^{(i)}_{Y_j|X}))$
for $i=1,2$.
Then
$$\left(
d^{(1)}+d^{(2)},
(c_j)\right)
\in
\BL\left(Q_X^{(1)}\times Q_X^{(2)},
\left(Q_{Y_j|X}^{(1)}\times Q_{Y_j|X}^{(2)}\right)\right)$$
where $d\1+d\2$ is defined as the function
\begin{align}
\mathcal{X}\1\times\mathcal{X}\2&\to\mathbb{R};
\\
(x_1,x_2)&\mapsto d\1(x_1)+d\2(x_2).
\end{align}
\end{prop}
We provide a simple information-theoretic proof using the chain rules of the relative entropy. Note that the algebraic expansions here are similar to the ones in the proof of Gaussian optimality in Section~\ref{sec_gaussian} or the converse proof for the key generation problem in Section~\ref{sec_key}.
\begin{proof}
For any arbitrary $P_{X^{(1)}X^{(2)}}$, define $P_{X\1X\2Y_j\1Y_j\2}:=P_{X\1X\2}Q_{Y_j|X}\1Q_{Y_j|X}\2$.
Observe that
\begin{align}
D(P_{X^{(1)}X^{(2)}}\|Q_X^{(1)}\times Q_X^{(2)})
=D(P_{X\1}\|Q_X^{(1)})
+D(P_{X\2|X\1}\|Q_X\2|P_{X\1}).
\label{e_51}
\end{align}
\begin{align}
D(P_{Y_j\1Y_j\2}\|Q_{Y_j}\1\times Q_{Y_j}\2)
&=D(P_{Y_j\1}\|Q_{Y_j}\1)+D(P_{Y_j\2|Y_j\1}\|Q_{Y_j}\2|P_{Y_j\1})
\\
&\le D(P_{Y_j\1}\|Q_{Y_j}\1)
+D(P_{Y_j\2|X\1Y_j\1}\|Q_{Y_j}\2|P_{X\1Y_j\1})\label{e52}
\\
&=D(P_{Y_j\1}\|Q_{Y_j}\1)
+D(P_{Y_j\2|X\1}\|Q_{Y_j}\2|P_{X\1})
\label{e53}
\end{align}
where \eqref{e52} uses Jensen's inequality, and \eqref{e53} is from the Markov chain $\hat{Y}_j\2-\hat{X}\1-\hat{Y}_j\1$,
wherein $(\hat{X}^{(i)},\hat{Y}_j^{(i)})\sim P_{X^{(i)}Y_j^{(i)}}$ for $i=1,2$, $j=1,\dots,m$.
By the assumption and the law of total expectation,
\begin{align}
D(P_{X\1}\|Q_X\1)+\mathbb{E}[d(\hat{X}\1)]
&\ge \sum_{j=1}^m c_jD(P_{Y_j\1}\|Q_{Y_j}\1);
\label{e56}
\\
D(P_{X\2|X\1}\|Q_X\2|P_{X\1})+\mathbb{E}[d(\hat{X}\2)]
&\ge \sum_{j=1}^m c_jD(P_{Y_j\2|X\1}\|Q_{Y_j}\2|P_{X\1}).
\label{e57}
\end{align}
Adding up \eqref{e56} and \eqref{e57} and applying \eqref{e_51} and \eqref{e53}, we obtain
\begin{align}
D(P_{X^{(1)}X^{(2)}}\|Q_X^{(1)}\times Q_X^{(2)})
+\mathbb{E}[(d\1+d\2)(X\1,X\2)]
\ge \sum_{j=1}^m c_jD(P_{Y_j\1Y_j\2}\|Q_{Y_j}\1\times Q_{Y_j}\2)
\end{align}
as desired.
\end{proof}
A functional proof of the tensorization of reverse Brascamp-Lieb inequalities can be given by generalizing the proof of the tensorization of the Pr\'ekopa-Leindler inequality (see for example \cite{tao_leindler}).
Alternatively, information-theoretic proofs of the tensorization of these reverse-type inequalities can be extracted from the proof of the Gaussian optimality in Theorem~\ref{thm_fr_optimality} ahead, and we omit the repetition here.

\subsection{Convexity}
Another property which follows conveniently from the information-theoretic characterization of $\BL(\cdot)$ is convexity:
\begin{prop}\label{prop_conv}
If $(d^i,(c_j^i))\in\BL(Q_X,(Q_{Y_j|X}))$ for $i=0,1$, then
$(d^{\theta},(c_j^{\theta}))\in\BL(Q_X,(Q_{Y_j|X}))$ for $\theta\in[0,1]$, where we have defined
\begin{align}
d^{\theta}&:=(1-\theta)d^0+\theta d^1,
\\
c_j^{\theta}&:=(1-\theta)c_j^0+\theta c_j^1,\quad\forall j\in\{1,\dots,m\}.
\end{align}
\end{prop}
\begin{proof}
Follows immediately from the \eqref{e_info} and taking convex combinations.
\end{proof}

Note that by taking $m=2$, $X=(Y_1,Y_2)$, $d^i(\cdot)=0$ and $Q_{Y_j|X}$ to be the projection to the coordinates, we recover the Riesz-Thorin theorem on the interpolation of operator norms in the special case of nonnegative kernels. This information-theoretic proof (for this special case) is much simpler than the common proof of the Riesz-Thorin theorem in functional analysis based on the three-lines lemma, because the $c_j$'s only affect the right side of \eqref{e_info} as linear coefficients, rather than as tilting of the distributions or functions.

\section{Gaussian Optimality Associated with the Forward inequality}\label{sec_gaussian}
In this section we prove the Gaussian extremality in several information-theoretic inequalities related to the forward Brascamp-Lieb inequality. Specifically, we first establish this for an inequality involving conditional differential entropies, which immediately implies the variants involving conditional mutual informations or differential entropies; the latter is directly connected to the Brascamp-Lieb inequality, as Theorem~\ref{thm_1} showed. These extremal inequalities have implications for certain operational problems in information theory, and quite interestingly, the essential steps in the proofs of these extremal inequalities follow the same patterns as the converse proofs for the corresponding operational problems.

Roughly speaking, the proof method is essentially based on the fact that two independent random variables are both Gaussian if their sum is independent of their difference (i.e.~Cramer's theorem \cite{cramer1936}). This rotation invariance argument\footnote{This argument was referred to as ``$O(2)$-invariance'' in \cite{bennett2008brascamp} and ``doubling trick'' in \cite{carlen1991superadditivity}.} has been used in establishing Gaussian extremality by Lieb \cite{lieb1990gaussian}, Carlen \cite{carlen1991superadditivity} and recently
in information theory by Geng-Nair \cite{geng2014yanlin} \cite{nairextremal}, Courtade-Jiao \cite{courtade2014extremal} and Courtade \cite{ISIT_courtade2016}.
Some related ideas have also appeared in the literature on the Brascamp-Lieb inequality, such as the observation that convolution preserves the extremizers of Brascamp-Lieb inequality \cite[Lemma~2]{barthe1998optimal} due to Ball.
However, as keenly noted in \cite{geng2014yanlin}, applying the rotation invariance/doubling trick on the information-theoretic formulation has certain advantages. For example, the chain rules provide convenient tools, and the establishment of the extremality usually follows similar steps as the converse proofs of the corresponding operational problems in information theory.
Since the optimization problems we consider involve many information-theoretic terms, we introduce a simplification/strengthening of the Geng-Nair approach by perturbing the coefficients in the objective function (see Remark~\ref{rem_expansion}), thus giving rise to some identities which become handy in the proof. A similar idea was used in \cite{courtade2014extremal}, and this should be applicable to a wide range of other problems.


In this section, $\mathcal{X},\mathcal{Y}_1,\dots,\mathcal{Y}_m$ are assumed to be Euclidean spaces of dimensions $n,n_1,\dots,n_m$. To be specific about the notions of Gaussian optimality, we adopt some terminologies from \cite{bennett2008brascamp}:
\begin{defn}\label{defn_extremisability}
\begin{itemize}
  \item Extremisability: a certain supremization/infimization is finitely attained by some argument.
  \item Gaussian extremisability: a certain supremization/infimization is finitely attained by Gaussian function/Gaussian distributions.
  \item Gaussian exhaustibility: the value of a certain supremization/infimization does not change when the arguments are restricted to the subclass of Gaussian functions/Gaussian distributions.
\end{itemize}
\end{defn}
Most of the times, we will be able to prove Gaussian extremisability in a certain non-degenerate case, while showing Gaussian exhaustibility in general.

\subsection{Optimization of Conditional Differential Entropies}\label{sec_cond}
Fix $\mb{M}\succeq 0$,  $c_0\in[0,\infty)$, $c_1,\dots,c_m\in(0,\infty)$, and Gaussian random transformations $Q_{\mb{Y}_j|\mb{X}}$ for $j\in\{1,\dots,m\}$. For each $P_{\mb{X}U}$\footnote{In the case of standard Borel space, the conditional distribution $P_{{\bf X}|U=\cdot}$ can be uniquely defined from the joint distribution $P_{{\bf X}U}$, $P_U$-almost surely; see e.g.~\cite{verdubook}.}, define
\begin{align}
F(P_{\mb{X}U} ) := h(\mb{X}|U)-\sum_{j=1}^m c_j h(\mb{Y}_j|U)-c_0 \Tr [\mb{M} \mb{\Sigma}_{\mb{X}|U}],\label{gFunc}
\end{align}
where $\mb{X}\sim P_{\mb{X}}$ (the marginal of $P_{\mb{X}U}$) and $\mb{Y}_j$ has distribution induced by $P_{\mb{X}} \to Q_{\mb{Y}_j|\mb{X}}\to P_{\mb{Y}_j}$.
We have defined the differential entropy and the conditional differential entropies as
\begin{align}
h(\mb{X})&:=-D(P_{\mb{X}}\|\lambda);
\\
h(\mb{X}|U=u)&:=-D(P_{\mb{X}|U=u}\|\lambda),\quad\forall u\in\mathcal{U};
\\
h(\mb{X}|U)&:=\int h(\mb{X}|U=u){\rm d}P_U,\label{e79}
\end{align}
where $\lambda$ is the Lebesgue measure (with the same dimension as $\mb{X}$), and \eqref{e79} is defined whenever the integral exists.
Moveover, we have used the notation $\bsigma_{\mb{X}|U}:=\mathbb{E}[\Cov(\mb{X}|U)]$ for the expectation of the conditional covariance matrix.

\begin{defn}\label{defn_ND}
We say $(Q_{\mb{Y}_1|\mb{X}},\dots,Q_{\mb{Y}_m|\mb{X}})$ is \emph{non-degenerate} if each $Q_{\mb{Y}_j|\mb{X=0}}$ is a $n_j$-dimensional Gaussian distribution with invertible covariance matrix.
\end{defn}

In the non-degenerate case, we can show an \emph{extremal} result for the following optimization with a regularization on the covariance of the input.

\begin{thm}\label{thm:GaussEntropy}
If $(Q_{\mb{Y}_1|\mb{X}},\dots,Q_{\mb{Y}_m|\mb{X}})$ is non-degenerate, then $\sup_{P_{\mb{X}U}} \{ F(P_{\mb{X}U})   \colon  \bsigma_{\mb{X}|U}\preceq \mb{\Sigma} \}$ is finite and is attained by a Gaussian $\mb{X}$ and constant $U$. Moreover, the covariance of such $\mb{X}$ is unique.
\end{thm}

In Theorem~\ref{thm:GaussEntropy},
we assume that the supremum is over $P_{\mb{X}U}$
such that $P_{\mb{X}|U=u}$ is absolutely continuous with respect to the Lebesgue measure (hence having a density function) for almost every $u$\footnote{Since the integral in \eqref{e79} may not be well-defined in general, some authors have restricted the attention to finite $\mathcal{U}$ in the optimization problems. In this paper, we are allowed to drop this restriction as long as the integral in \eqref{e79} is well-defined. These distinctions do not appear to make an essential difference for our purpose; see Footnote~\ref{f18}.}. Additionally, we adopt the following convention in all the optimization problems in Section~\ref{sec_gaussian} and Section~\ref{sec_consequence}, unless otherwise specified. This eliminates situations such as $\infty+\infty$ or $a+\infty$ which can be considered as legitimate calculations but are technically difficult to deal with.
\begin{conv}\label{conv1}
The $\sup$ or $\inf$ are taken over all arguments such that \emph{each term} in the objective function (e.g.~\eqref{gFunc}) is well-defined and finite.
\end{conv}
\begin{proof}[Proof of Theorem~\ref{thm:GaussEntropy}]
Assume that both $P_{\mb{X}^{(1)}U^{(1)}}$ and $P_{\mb{X}^{(2)}U^{(2)}}$ are maximizers of \eqref{gFunc} subject to $ \bsigma_{\mb{X}|U}\preceq \bsigma$; the proof of the existence of maximizer is deferred to Appendix~\ref{app_exist}. Let $(U^{(1)}, \mb{X}^{(1)}, \mb{Y}_1^{(1)}, \dots, \mb{X}_m^{(1)})\sim P_{\mb{X}^{(1)}U^{(1)}}Q_{\mb{Y}_1|\mb{X}}\dots Q_{\mb{Y}_m|\mb{X}}$ and $(U^{(2)}, \mb{X}^{(2)}, \mb{Y}_1^{(2)}, \dots, \mb{X}_m^{(2)})\sim P_{\mb{X}^{(2)}U^{(2)}}Q_{\mb{Y}_1|\mb{X}}\dots Q_{\mb{Y}_m|\mb{X}}$ be mutually independent. Define
\begin{align}
&{ \mb X}^+ = \frac{1}{\sqrt{2}}\left( \mb{X}^{(1)} +   \mb{X}^{(2)}\right) &{ \mb{X}}^-  = \frac{1}{\sqrt{2}}\left(  \mb{X}^{(1)} -   \mb{X}^{(2)}\right) .
\end{align}
Define $\mb{Y}_j^{+}$ and $\mb{Y}_j^{-}$ similarly for $j=1, \dots, m$, and put $\hat{U}=(U^{(1)},U^{(2)})$.  We now make three important observations:
\begin{enumerate}
\item First, due to the Gaussian nature of $Q_{\mb{Y}_j|\mb{X}}$, it is easily seen that $\mb{Y}_j^{+}|\{\mb{X}^+=\mb{x}^+,\mb{X}^-=\mb{x}^-,
    \hat{U}=\hat{u}\}\sim Q_{\mb{Y}_j|\mb{X}=\mb{x}^+}$ is independent of $\mb{x}^-$.
    Thus $\mb{Y}_j^{+}|\{\mb{X}^+=\mb{x},
    \hat{U}=\hat{u}\}\sim Q_{\mb{Y}_j|\mb{X}=\mb{x}}$ as well.
Similarly, $\mb{Y}_j^{-}|\{\mb{X}^-=\mb{x},\hat{U}=u\}\sim Q_{\mb{Y}_j|\mb{X}=\mb{x}}$ for $j=1, \dots, m$.
\item  Second, observe that
for each $\hat{u}=(u_1,u_2)$ we can verify the algebra
\begin{align}
\bsigma_{\mb{X}^+|\uh=\hat{u}}
&=\mathbb{E}[(\mb{X}^+-\mu_{\mb{X}^+|\uh})
(\mb{X}^+-\mu_{\mb{X}^+|\uh})^{\top}|\hat{U}=\hat{u}]
\nonumber
\\
&=
\half\mathbb{E}[(\mb{X}^{(1)}-\mu_{\mb{X}^{(1)}|\uh})
(\mb{X}^{(1)}-\mu_{\mb{X}^{(1)}|\uh})^{\top}|\uh=\hat{u}]
\nonumber
\\
&\quad+
\half\mathbb{E}[(\mb{X}^{(2)}-\mu_{\mb{X}^{(2)}|\uh})
(\mb{X}^{(2)}-\mu_{\mb{X}^{(2)}|\uh})^{\top}|\uh=\hat{u}]
\nonumber
\\
&\quad+
\mathbb{E}[(\mb{X}^{(1)}-\mu_{\mb{X}^{(1)}|\uh})
(\mb{X}^{(2)}-\mu_{\mb{X}^{(2)}|\uh})^{\top}|\uh=\hat{u}]
\nonumber
\\
&=
\half\mathbb{E}[(\mb{X}^{(1)}-\mu_{\mb{X}^{(1)}|U\1})
(\mb{X}^{(1)}-\mu_{\mb{X}^{(1)}|U\1})^{\top}|\uh=\hat{u}]
\nonumber
\\
&\quad+
\half\mathbb{E}[(\mb{X}^{(2)}-\mu_{\mb{X}^{(2)}|U\2})
(\mb{X}^{(2)}-\mu_{\mb{X}^{(2)}|U\2})^{\top}|\uh=\hat{u}]
\nonumber
\\
&\quad+
\mathbb{E}[(\mb{X}^{(1)}-\mu_{\mb{X}^{(1)}|U\1})
(\mb{X}^{(2)}-\mu_{\mb{X}^{(2)}|U\2})^{\top}|\uh=\hat{u}].
\end{align}
The last term above vanishes upon averaging over $(u_1,u_2)$ because of the independence $(U\1,\mb{X}\1)\perp(U\2,\mb{X}\2)$. Thus
\begin{align}
\bsigma_{\mb{X}^+|\uh}=\half \bsigma_{\mb{X}\1|U\1}
+\half \bsigma_{\mb{X}\2|U\2}
\preceq\bsigma.
\label{e+}
\end{align}
By the same token,
\begin{align}
\bsigma_{\mb{X}^-|\uh}=\half \bsigma_{\mb{X}\1|U\1}
+\half \bsigma_{\mb{X}\2|U\2}
\preceq\bsigma.
\label{e-}
\end{align}
These combined with $\bsigma_{\mb{X}^{-}|\mb{X}^+\hat{U}}\preceq  \bsigma_{\mb{X}^{-}|\hat{U}}$ (which is a consequence of the convexity of the square function) justify that both $P_{\mb{X}^+,\hat{U}}$ and $P_{\mb{X}^-,\hat{U}\mb{X}^+}$ satisfy the covariance constraint
$\bsigma_{\mb{X}|U}\preceq \mb{\Sigma}$ in the theorem.

\item Third, we have
\begin{align}
&\quad\sum_{k=1}^2 \left[h(\mb{X}^{(k)}|U^{(k)}) - \sum_{j=1}^m c_j h(\mb{Y}^{(k)}_j|U^{(k)})\right]
\\
& =h(\mb{X}\1,\mb{X}\2|\hat{U}) - \sum_{j=1}^m c_j h(\mb{Y}\1_j,\mb{Y}\2_j|\hat{U})
\\
& =h(\mb{X}^{+},\mb{X}^-|\hat{U}) - \sum_{j=1}^m c_j h(\mb{Y}^{^+}_j,\mb{Y}^-_j|\hat{U})
\\
&=h(\mb{X}^{+} |\hat{U}) - \sum_{j=1}^m c_j h(\mb{Y}^{^+}_j|\hat{U}) +h(\mb{X}^{-}|\mb{X}^+,\uh) - \sum_{j=1}^m c_j h(\mb{Y}^{-}_j|\mb{Y}^+_j,\uh)\\
&\leq h(\mb{X}^{+} |\hat{U}) - \sum_{j=1}^m c_j h(\mb{Y}^{^+}_j|\hat{U}) +h(\mb{X}^{-}|\mb{X}^+,\uh) - \sum_{j=1}^m c_j h(\mb{Y}^{-}_j|\mb{X}^+,\uh),
\label{e_h_tens}
\end{align}
where the final inequality follows from the Markov chain $\mb{Y}^{-}_j-\mb{X}^+\uh-\mb{Y}^+_j$,
which is because the joint distribution factorizes as $P_{\uh\mb{X}^+\mb{X}^-\mb{Y}_j^+\mb{Y}_j^-}
=P_{\uh\mb{X}^+\mb{X}^-}Q_{\mb{Y}_j|\mb{X}}Q_{\mb{Y}_j|\mb{X}}$.
\end{enumerate}
Thus, we can conclude that
\begin{align}
\sum_{i=1}^2F(P_{\mb{X}^{(i)}U^{(i)}})
&= \sum_{i=1}^2 \left[h(\mb{X}^{(k)}|U^{(k)}) - \sum_{j=1}^m c_j h(\mb{Y}^{(k)}_j|U^{(k)}) -c_0 \Tr [M \bsigma_{\mb{X}^{(k)}|U^{(k)}}]\right]
\label{e86}
\\
&\le h(\mb{X}^{+}|\uh) - \sum_{j=1}^m c_j h(\mb{Y}^{+}_j|\uh) -c_0 \Tr [\mb{M}\bsigma_{\mb{X}^{+}|\uh}]\notag\\
&~~+h(\mb{X}^-|\mb{X}^{+},\uh) - \sum_{j=1}^m c_j h(\mb{Y}^-_j|\mb{X}^{+},\uh) -c_0 \Tr [\mb{M} \bsigma_{\mb{X}^-|\mb{X}^{+},\uh}]
\label{e82}
\\
&\le \sum_{i=1}^2F(P_{\mb{X}^{(i)}U^{(i)}}),
\label{e83}
\end{align}
where
\begin{itemize}
  \item \eqref{e82} follows from \eqref{e+},\eqref{e-} and \eqref{e_h_tens};
  \item \eqref{e83} follows since $P_{\mb{X}^+,\uh}$ and $P_{\mb{X}^-,\uh \mb{X}^+}$ are candidate optimizers of \eqref{gFunc} subject to the given covariance constraint
whereas $P_{\mb{X}^{(i)},U^{(i)}}$ are the optimizers by assumption ($i=1,2$).
\end{itemize}
Then, the equalities in \eqref{e86}-\eqref{e83} must be achieved throughout, so
both $P_{\mb{X}^+,\hat{U}}$ and $P_{\mb{X}^-,\hat{U}\mb{X}^+}$ (and also $P_{\mb{X}^+,\hat{U}\mb{X}^-}$, by symmetry of the argument) are maximizers of \eqref{gFunc} subject to $ \bsigma_{\mb{X}|U}\preceq \bsigma$.

So far, we have considered fixed coefficients $c_0^m=(c_0,\dots,c_m)$. The same argument applies for coefficients on a line:
\begin{align}
c_0^m(t):= ta_0^m,\quad t>0
\label{e_line}
\end{align}
for any fixed $a_0\in[0,\infty)$, $a_1,\dots,a_m\in(0,\infty)$, and we next show several properties for a dense subset of this line.
Applying Lemma~\ref{lem:invariance} (following this proof) with
\begin{align}
p(P_{\mb{X}U})&\leftarrow h(\mb{X}|U);
\\
q(P_{\mb{X}U})&\leftarrow -\sum_{j=1}^m a_j h(\mb{Y}_j|U)-a_0 \Tr [\mb{M} \mb{\Sigma}_{\mb{X}|U}],
\\
f(t)&\leftarrow \max_{P_{\mb{X}U}}[p(P_{\mb{X}U})+tq(P_{\mb{X}U})],
\end{align}
we obtain from the optimality of $P_{\mb{X}^+,\hat{U}}$ $P_{\mb{X}^-,\hat{U}\mb{X}^+}$ and $P_{\mb{X}^+,\hat{U}\mb{X}^-}$
that
\begin{align}
h(\mb{X}^{+}|\uh) = h(\mb{X}^{-}|\mb{X}^{+},\uh)
=h(\mb{X}^{+}|\mb{X}^{-},\uh)
\label{e_92}
\end{align}
for almost all $t\in(0,\infty)$, where $P_{\mb{X}^+\mb{X}^-\hat{U}}$ depends implicitly on $t$.
Note that \eqref{e_92} implies that $I(\mb{X}^+;\mb{X}^-|\uh)=0$ hence $\mb{X}^+$ and $\mb{X}^-$ are independent conditioned on $\uh$. Recall the following Skitovic-Darmois characterization of Gaussian distributions (with the extension to the vector Gaussian case in \cite{geng2014yanlin}):
\begin{lem}\label{SDcharacterization}
Let $\mathbf{A}_1$ and $\mathbf{A}_2$ be mutually independent $n$-dimensional random vectors.  If $\mathbf{A}_1+\mathbf{A}_2$ is independent of $\mathbf{A}_1-\mathbf{A}_2$, then  $\mathbf{A}_1$ and $\mathbf{A}_2$ are normally distributed with identical covariances.
\end{lem}
Using Lemma \ref{SDcharacterization}, we can conclude that for almost all $t\in(0,\infty)$, $\mb{X}^{(i)}$ must be Gaussian, with covariance not depending on $U^{(i)}$, thus $U^{(i)}$ can be chosen as a constant ($i=1,2$).
Thus for all such $t$,
\begin{align}
f(t)=\max_{P_{\mb{X}U}\colon U=\textrm{const.},\,{\bf X}\textrm { Gaussian}}
[p(P_{\mb{X}U})+tq(P_{\mb{X}U})].
\label{e139}
\end{align}
Since both sides of \eqref{e139} are concave in $t$, hence continuous on $(0,\infty)$, we see \eqref{e139} actually holds for all $t\in(0,\infty)$.
The proof is completed since $a_0^m$ can be arbitrarily chosen.
\end{proof}

\begin{lem}\label{lem:invariance}
Let $p$ and $q$ be real-valued functions on an arbitrary set $\mathcal{D}$.  If
$f(t):= \max_{x\in \mathcal{D}} \{p(x) + t q(x)\}$
is always attained,
then for almost all $t$, $f'(t)$ exists and
\begin{align}
f'(t)=q(x^{\star}), ~~~\forall x^{\star}\in \arg \max_{x\in \mathcal{D}} \{p(x) + t q(x)\}.
\end{align}
In particular, for all such $t$, $
q(x^{\star}) = q(\tilde{x}^{\star})$ and
$p(x^{\star}) = p(\tilde{x}^{\star})$ for all $x^{\star},\tilde{x}^{\star}\in \arg \max_{x\in \mathcal{D}} \{p(x) + t q(x)\}$.
\end{lem}
Geometrically, $f(t)$ is the \emph{support function} \cite{rockafellar2015convex} of the set $\mathcal{S}:=\{(p(x),q(x))\}_{x\in\mathcal{D}}$ evaluated at $(1,t)$. Hence $f(\cdot)$ is convex, and the left and the right derivatives are determined by the two extreme points of the intersection between $\mathcal{S}$ and the supporting hyperplane.
\begin{proof}[Proof of Lemma~\ref{lem:invariance}]
The function $f$ is convex since it is a pointwise supremum of linear functions, and is therefore differentiable almost everywhere.  Moreover,  $f'(t)$ (which is well-defined in the a.e.\ sense) is monotone increasing by convexity, and is therefore continuous almost everywhere.

Let $x^{\star}_t$ denote an arbitrary element of $\arg \max_{x\in \mathcal{D}} \{p(x) + t q(x)\}$. By definition, for any $s \in \mathbb{R}$,
$f(s) \ge p( x^{\star}_t) + s q( x^{\star}_t)$.  Thus, for $\delta>0$,
\begin{align}
&\frac{f(t+\delta) - f(t)}{\delta} \ge q(x^{\star}_{t}) \mbox{~~~and~~~}
\frac{f(t) - f(t-\delta)}{\delta} \le q(x^{\star}_{t}).
\end{align}
If $f'(t)$ exists, that is, the left sides of the two inequalities above have the same limit $f'(t)$ as $\delta\downarrow0$, then $f'(t)=q(x^{\star}_{t})$. The second claim of the lemma follows immediately from the first.
\end{proof}
\begin{rem}
Let us remark on an interesting connection between the above proof of the optimality of Gaussian random variable and Lieb's proof of that Gaussian functions maximize Gaussian kernels. Recall that \cite[Theorem~3.2]{lieb1990gaussian} wants to show that for an operator $G$ given by a two-variate Gaussian kernel function and $p,q>1$, the ratio $\frac{\|Gf\|_q}{\|f\|_p}$ is maximized by Gaussian $f$. First, a tensorization property is proved, implying that $f^*(x_1)f^*(x_2)$ is a maximizer for $G\otimes G$ if $f^*$ is any maximizer of $G$. Then, Lieb made two important observations:
\begin{enumerate}
  \item By a rotation invariance property of the Lebesgue measure/isotropic Gaussian measure, $f^*(\frac{x_1+x_2}{\sqrt{2}})f^*(\frac{x_1-x_2}
      {\sqrt{2}})$ is also a maximizer of $G\otimes G$.
  \item An examination of the equality condition in the proof of tensorization property reveals that any maximizer for $G\otimes G$ must be of a product form.
\end{enumerate}
Thus Lieb concluded that $f^*(\frac{x_1+x_2}{\sqrt{2}})f^*(\frac{x_1-x_2}
      {\sqrt{2}})=\alpha(x_1)\beta(x_2)$ for some functions $\alpha$ and $\beta$, and $f^*$ must be a Gaussian function. This is very similar to the above proof, once we think of $f^*$ as the density function of $P^*_{X|U=u}$ in our proof.
\end{rem}
\begin{rem}\label{rem_expansion}
Our proof technique is essentially following ideas of
Geng and Nair \cite{geng2014yanlin}\cite{nairextremal} who established the Gaussian optimality for several information-theoretic regions. However, we also added the important ingredient of Lemma~\ref{lem:invariance}.\footnote{The similar idea of differentiating the coefficients has been used in another paper of the second named author \cite{courtade2014extremal}.} That is, by differentiating with respect to the linear coefficients, we can conveniently obtain information-theoretic identities which helps us to conclude the conditional independence of $\mb{X}^+$ and $\mb{X}^-$ quickly.
For fixed $m$, in principle, this may be avoided by trying various expansions of the two-letter quantities manually (e.g.~as done in \cite{nairextremal}), but that approach will become increasingly complicated and unstructured as $m$ increases.
Finally, we also note that a simple rotational invariance argument/doubling trick has been used for proving that the capacity achieving distribution for an additive Gaussian channel is Gaussian (cf.~\cite[P36]{PW_lec}), which does not involve Lemma~\ref{SDcharacterization} and whose extension to problems involving auxiliary random variables is not clear.
\end{rem}

If we do not have the non-degenerate assumption and the regularization $\bsigma_{\mb{X}|U}\preceq \mb{\Sigma}$, it is very well possible that the optimization in Theorem~\ref{thm:GaussEntropy} is nonfinite and/or not attained by any $P_{U\mb{X}}$. In this case, we can show that the optimization is \emph{exhausted} by Gaussian distributions.
To state the result conveniently, for any $P_{\mb{X}}$, define
\begin{align}
F_0(P_{\mb{X}}):=h(\mb{X})-\sum_{j=1}^m c_j h(\mb{Y}_j)-c_0\Tr[\mb{M}\bsigma_{\mb{X}}],
\label{e_103}
\end{align}
where $(\mb{X},\mb{Y}_j)\sim P_{\mb{X}}Q_{\mb{Y}_j|\mb{X}}$. Apparently, $F(P_{\mb{X}U})=F_0(P_{\mb{X}})$ when $U$ is constant.

\begin{thm}\label{thm:exhaustibility}
In the general (possibly degenerate) case,
\begin{enumerate}
\item For any given positive semidefinite $\mb{\Sigma}$,
\begin{align}
\sup_{P_{\mb{X}U},\,\bsigma_{\mb{X}|U}\preceq \bsigma} F(P_{\mb{X}U})
=
\sup_{P_{\mb{X}}\textrm{ Gaussian},\,\bsigma_{\mb{X}}\preceq \bsigma} F_0(P_{\mb{X}})
\label{e_78}
\end{align}
where the left side of \eqref{e_78} follows Convention~\ref{conv1}.
\item
\begin{align}
\sup_{P_{\mb{X}U}} F(P_{\mb{X}U})
=
\sup_{P_{\mb{X}}\textrm{ Gaussian}} F_0(P_{\mb{X}}).
\label{e_78_1}
\end{align}
\end{enumerate}
\end{thm}
\begin{rem}
Theorem~\ref{thm:exhaustibility} reduces an infinite dimensional optimization problem to a finite dimensional one. In particular, in the degenerate case we can verify that the left side of \eqref{e_78} to be extremisable (Definition~\ref{defn_extremisability}) if the the right side of \eqref{e_78} is extremisable.
\end{rem}
\begin{proof}[Proof of Theorem~\ref{thm:exhaustibility}]
Note that the Gaussian random transformation $Q_{\mb{Y}_j|\mb{X}}$ can be realized as a linear transformation $\mb{B}_j$ to $\mathbb{R}^{n_j}$ followed by adding an independent Gaussian noise of covariance $\bsigma_j$. We will assume that $\bsigma_j$ is non-degenerate in the orthogonal complement of the image of $\mb{B}_j$, since otherwise both sides of \eqref{e_78} are $+\infty$. In particular, if $Q_{\mb{Y}_j|\mb{X}}$ is a deterministic linear transform $\mb{B}_j$, then it must be onto $\mathbb{R}^{n_j}$.
\begin{enumerate}
  \item\label{p1} We can use the argument at the beginning of Appendix~\ref{app_exist} to show that the left side of \eqref{e_78} equals
\begin{align}
\sup_{P_{\mb{X}U}\colon \mathcal{U}=\{1,\dots,D\},\bsigma_{\mb{X}|U}\preceq \mb{\Sigma}} F(P_{\mb{X}U})
\label{e_79}
\end{align}
where $D$ is a fixed integer (depending only on the dimension of $\bsigma$).

Next, we argue that it is without loss of generality to add a restriction to \eqref{e_79} that $P_{\mb{X}|U=u}$ has a smooth density with compact support for each $u\in\{1,\dots,D\}$, in which case each $P_{\mb{Y}_j|U=u}$ will have a smooth density by the assumption made at the beginning of the proof.
For this, define $F_0(\cdot)$ as in \eqref{e_103}, and we will show that for any $P_{\mb{X}}$ absolutely continuous with respect to the Lebesgue measure and any $\epsilon>0$, we can choose a distribution $P_{\mb{X}''}$ whose density is smooth with compact support such that
\begin{align}
F_0(P_{\mb{X}''})\ge F_0(P_{\mb{X}})-\epsilon
\label{e80}
\end{align}
and
\begin{align}
\bsigma_{\mb{X}''}\preceq (1+\epsilon)\bsigma_{\mb{X}}.
\label{e81}
\end{align}
First, by conditioning $\bf X$ on a large enough compact set and applying dominated convergence theorem, there exists $P_{\mb{X}'}$ which is supported on a compact set and whose density $f_{\mb{X}'}$ is bounded, such that each term in the definition of $F_0(P_{\mb{X}})$ is well-approximated when $\mb{X}$ is replaced with $\mb{X}'$, so that
\begin{align}
F_0(P_{\mb{X}'})\ge F_0(P_{\mb{X}})-\frac{\epsilon}{2}
\label{e_82}
\end{align}
and
\begin{align}
\bsigma_{\mb{X}'}\preceq \sqrt{1+\epsilon}\bsigma_{\mb{X}}.
\label{e_83}
\end{align}
Second, we can pick $P_{\mb{X}''}$ with smooth density with compact support such that $\|f_{\mb{X}''}-f_{\mb{X}'}\|_1$ can be made arbitrarily small, which implies that $\|f_{\mb{Y}_j''}-f_{\mb{Y}_j'}\|_1$ is small by Jensen's inequality, and that $\Tr[\mb{M}\bsigma_{\mb{X}''}]-\Tr[\mb{M}\bsigma_{\mb{X}'}]$ is small by the fact that the densities are supported on a compact set. Since $x\mapsto x\log x$ is Lipschitz on a bounded interval, the differential entropy terms can be made small as well (here we used the boundedness of $f_{\bf X'}$). As a result, we can ensure that
\begin{align}
F_0(P_{\mb{X}''})\ge F_0(P_{\mb{X}}')-\frac{\epsilon}{2}
\end{align}
and
\begin{align}
\bsigma_{\mb{X}''}\preceq \sqrt{1+\epsilon}\bsigma_{\mb{X}'},
\end{align}
which, combined with \eqref{e_82} and \eqref{e_83}, yield \eqref{e80}-\eqref{e81}. This and the observation in \eqref{e_79} imply that
\begin{align}
\inf_{\epsilon>0}\,\sup_{P_{\mb{X}U}\in\mathcal{C}_{\epsilon}} F(P_{\mb{X}U})
\le \sup_{P_{\mb{X}U}\colon \bsigma_{\mb{X}|U}\preceq \mb{\Sigma}} F(P_{\mb{X}U})
\label{e_86}
\end{align}
where $\mathcal{C}_{\epsilon}$ denotes the set of $P_{\mb{X}U}$ such that  $\mathcal{U}=\{1,\dots,D\}$, $\bsigma_{\mb{X}|U}\preceq \mb{(1+\epsilon)\Sigma}$, and the density of $P_{\mb{X}|U=u}$ is smooth with compact support for each $u$. In the final steps, we write $F$ as $F^Q$ to indicate its dependence on $(Q_{\mb{Y}_1|\mb{X}},\dots,Q_{\mb{Y}_m|\mb{X}})$, and define $\mathcal{\tilde{Q}}$ the set of non-degenerate $(Q_{\mb{\tilde{Y}}_1|\mb{X}},\dots,
Q_{\mb{\tilde{Y}}_m|\mb{X}})$ where each $\mb{\tilde{Y}}_j$ is obtained by adding an independent Gaussian noise with covariance $\tilde{\bsigma}_j$ to $\mb{Y}_j$ (here the random transformation from $\bf X$ to ${\bf Y}_j$ is fixed and $\tilde{\bsigma}_j$ runs over the set of positive definite matrices of the given dimension, $j=1,\dots,m$).
Then
\begin{align}
\sup_{P_{\mb{X}U}\in\mathcal{C}_{\epsilon}} F^{Q}(P_{\mb{X}U})
&=
\sup_{P_{\mb{X}U}\in\mathcal{C}_{\epsilon}}
\sup_{\tilde{Q}\in\mathcal{\tilde{Q}}}
F^{\tilde{Q}}(P_{\mb{X}U})
\label{e_87}
\\
&=\sup_{\tilde{Q}\in\mathcal{\tilde{Q}}}
\sup_{P_{\mb{X}U}\in\mathcal{C}_{\epsilon}}
F^{\tilde{Q}}(P_{\mb{X}U})
\label{e88}
\\
&=\sup_{\tilde{Q}\in\mathcal{\tilde{Q}}}
\sup_{P_{\mb{X}}\,\textrm{Gaussian},\bsigma_{\mb{X}}
\preceq (1+\epsilon)\bsigma}F_0^{\tilde{Q}}
(P_{\mb{X}})
\label{e89}
\\
&=
\sup_{P_{\mb{X}}\,\textrm{Gaussian},\bsigma_{\mb{X}}
\preceq (1+\epsilon)\bsigma}\sup_{\tilde{Q}\in\mathcal{\tilde{Q}}}
F_0^{\tilde{Q}}
(P_{\mb{X}})
\\
&=
\sup_{P_{\mb{X}}\,\textrm{Gaussian},\bsigma_{\mb{X}}
\preceq (1+\epsilon)\bsigma}
F_0^{Q}
(P_{\mb{X}})
\label{e_91}
\end{align}
where \begin{itemize}
        \item \eqref{e_87} and \eqref{e_91} are because, first,  $h(\tilde{\mb{Y}}_j)\ge h(\mb{Y}_j)$ by the entropy power inequality; second, $\inf_{\tilde{\bsigma}_j\succ \mb{0}}
            h(\tilde{\mb{Y}}_j)=h(\mb{Y}_j)$. The proof of the second claim is standard, since $\mb{\tilde{Y}}_j$ has smooth density with fast (Gaussian like) decay, and we can obtain pointwise convergence of the density function and apply dominated convergence theorem.\footnote{Such a continuity in the variance of the additive noise can fail terribly when $f_{\mb{X}}$ does not have the decay properties; in \cite[Proposition~4]{bobkov2015} Bobkov and Chistyakov provided an example of random vector $\mb{X}$ with finite differential entropy such that $h(\mb{X}+\mb{Z})=\infty$ for each $\mb{Z}$ independent of $\mb{X}$ and having finite differential entropy. In their example, $\sup_{\mb{x}:\|\mb{x}\|\ge r}f_{\mb{X}}(\mb{x})=1$ for every $r>0$, and we cannot obtain a dominating function for $|f_{\mb{X}+\mb{Z}}\log f_{\mb{X}+\mb{Z}}|$ to apply the dominated convergence theorem.}
        \item \eqref{e89} is from Gaussian extremality in non-degenerate case.
      \end{itemize}
      Then \eqref{e_86} and \eqref{e_91} gives
      \begin{align}
      \sup_{P_{\mb{X}U}\colon \bsigma_{\mb{X}|U}\preceq \mb{\Sigma}} F(P_{\mb{X}U})
      &\le
      \sup_{\epsilon>0}\,\sup_{P_{\mb{X}}\,
      \textrm{Gaussian},\bsigma_{\mb{X}}
        \preceq (1+\epsilon)\bsigma}
        F_0^{Q}
        (P_{\mb{X}})
        \label{e92}
         \\
         &=\sup_{P_{\mb{X}}\,\textrm{Gaussian},\bsigma_{\mb{X}}
        \preceq \bsigma}
        F_0^{Q}
        (P_{\mb{X}})
        \label{e93}
      \end{align}
      where \eqref{e93} is a property that can be verified without much difficulty for Gaussian distributions. Thus the $\le$ part of \eqref{e_78} is established. The other direction is immediate from the definition.
  \item First, observe that it is without loss of generality to assume that $U$ is constant, that is,
      \begin{align}
      \sup_{P_{\mb{X}U}}F(P_{\mb{X}U})
      =\sup_{P_{\mb{X}}}F_0(P_{\mb{X}}).
      \label{eq94}
      \end{align}
      Next, by the same argument as \ref{p1}), for any $P_{\mb{X}}$ and $\epsilon>0$, there exists $P_{\mb{X}'}$ such that $\|\mb{X}'\|<M$ with probability one for some finite $M>0$, and that
      \begin{align}
      F_0(P_{\mb{X}'})
      \ge F_0(P_{\mb{X}})-\epsilon.
      \label{eq95}
      \end{align}
      Then
      \begin{align}
      F_0(P_{\mb{X}'})
      &\le \sup_{\bsigma\succeq \mb{0}}
      \,\sup_{P_{\mb{X}U}\colon \bsigma_{\bf X}\preceq \bsigma}F(P_{\mb{X}U})
      \\
      &\le \sup_{\bsigma\succeq \mb{0}}
      \,\sup_{P_{\mb{X}}\,\textrm{Gaussian},\bsigma_{\mb{X}}
        \preceq \bsigma}
        F_0
        (P_{\mb{X}})
        \label{e97}
            \\
      &= \sup_{P_{\mb{X}}\textrm{ Gaussian}}
        F_0
        (P_{\mb{X}})
        \label{eq98}
      \end{align}
    where \eqref{e97} was established in \eqref{e93}. Finally, \eqref{eq94}, \eqref{eq95}, \eqref{eq98} and arbitrariness of $\epsilon$ give
    \begin{align}
    \sup_{P_{\mb{X}U}}F(P_{\mb{X}U})
    \le\sup_{P_{\mb{X}}\textrm{ Gaussian}}
        F_0
        (P_{\mb{X}}).
    \end{align}
    Thus the $\le$ part of \eqref{e_78_1} is established. The other direction is trivial from the definition.
\end{enumerate}
\end{proof}

\subsection{Optimization of Mutual Informations}
Let $\mb{X},\mb{Y}_1,\dots,\mb{Y}_m$ be jointly Gaussian vectors, $a_0,\dots,a_m$ be nonnegative real numbers, and $\mb{M}$ be a positive-semidefinite matrix. We are interested in minimizing $I(U;\mb{X})$ over $P_{U|\mb{X}}$ (that is, $U-{\bf X-Y}^m$ must hold) subject to $I(U;\mb{Y}_i)\ge a_i$, $i\in\{1,\dots,m\}$
and $\Tr[\mb{M}\bsigma_{\mb{X}|U}]\le a_0$.
This is relevant to many problems in information theory including the Gray-Wyner network \cite{wyner1975common} and some common randomness/key generation problems \cite{liu2015key} (to be discussed in  Section~\ref{sec_key}). We have the following result for the Lagrange dual of such an optimization problem:
\begin{thm}\label{thm1}
Fix $\mb{M}\succeq \mb{0}$,  positive constants $c_0,c_1, \dots, c_m$, and jointly Gaussian vectors $(\mb{X}, \mb{Y}_1, \dots, \mb{Y}_m)\sim Q_{\mb{X}\mb{Y}_1,\dots,\mb{Y}_m}$.  Define the function
\begin{align}
G(P_{U|\mb{X}} ) := \sum_{j=1}^m c_j I(\mb{Y}_j;U)-I(\mb{X};U)-c_0 \Tr[\mb{M} \bsigma_{\mb{X}|U}].\label{fFunc}
\end{align}
Then
\begin{enumerate}
  \item\label{p2} If $(Q_{\mb{Y}_1|\mb{X}},\dots,Q_{\mb{Y}_m|\mb{X}})$ is non-degenerate, then $\sup_{P_{U|\mb{X}}}{G(P_{U|\mb{X}})}$
is achieved by $U^{\star}$ for which $\mb{X}|\{U^{\star}=u\}$ is normal for $P_{U^{\star}}$-a.e. $u$, with covariance not depending on $u$. This can be realized by a Gaussian vector $\mb{U}^*$ with the same dimension as $\mb{X}$.
    \item In general, $\sup_{P_{U|\mb{X}}}G(P_{U|\mb{X}})
        =\sup_{P_{\mb{U}|\mb{X}}\colon\textrm{$(\mb{U},\mb{X})$ is Gaussian}}G(P_{U|\mb{X}})$.
\end{enumerate}
\end{thm}
\begin{proof}
Set $\mb{\Sigma}:=\bsigma_{\mb{X}}$.
Note that for any $P_{U|\mb{X}}$, we have
\begin{align}
    &  \sum_{j=1}^m c_j I(\mb{Y}_j;U)
    -I(\mb{X};U)
    -c_0 \Tr [\mb{M} \bsigma_{\mb{X}|U}] \\
 &=  \sum_{j=1}^m c_j h(\mb{Y}_j)
  -h(\mb{X})
  +\left( h(\mb{X}|U) -\sum_{j=1}^m c_j h(\mb{Y}_j|U)-c_0 \Tr [\mb{M} \bsigma_{\mb{X}|U}]\right)\\
 &\le \sum_{j=1}^m c_j h(\mb{Y}_j) - h(\mb{X}) +\sup_{P_{\mb{X}'U'}\colon \bsigma_{\mb{X}'|U'}\preceq \mb{\Sigma}} F(P_{\mb{X}'U'}).
 \label{e_105}
\end{align}
In the non-degenerate case, by Theorem \ref{thm:GaussEntropy}, the supremum is attained in the last line by constant $U'$ and $\mb{X}'\sim\mathcal{N}(\mb{0},\mb{\Sigma}')$, where $\mb{\Sigma}'\preceq \mb{\Sigma}$.
Hence, taking $\mb{U}^{\star}\sim \mathcal{N}(\mb{0},\mb{\Sigma-\Sigma}')$ and $P_{\mb{X}|\mb{U}^{\star}=\mb{u}}\sim \mathcal{N}(\mb{u},\bsigma')$, we have $P_{\mb{X}}\sim \mathcal{N}(\mb{0},\bsigma)$ as required, and \eqref{e_105} reads as
\begin{align}
 \sum_{j=1}^m c_j I(\mb{Y}_j;U)
 -I(\mb{X};U) - c_0 \Tr [\mb{M} \bsigma_{\mb{X}|U}] \le
 \sum_{j=1}^m c_j I(\mb{Y}_j;\mb{U}^{\star})
 -I(\mb{X};\mb{U}^{\star})
 -c_0 \Tr [\mb{M} \bsigma_{\mb{X}|\mb{U}^{\star}}] ,
\end{align}
and \ref{p2}) follows since $P_{U|\mb{X}}$ is arbitrary. The claim for the general (possibly degenerate) case follows by invoking Theorem~\ref{thm:exhaustibility} and using a similar argument.
\end{proof}

\subsection{Optimization of Differential Entropies}
The $F_0(\cdot)$ defined in Section~\ref{sec_cond} is a linear combination of differential entropies and second order moments, which is closely related to the Brascamp-Lieb inequality. Immediately from Theorem~\ref{thm:GaussEntropy} and Theorem~\ref{thm:exhaustibility}, we have the following result regarding maximization of $F_0(\cdot)$:
\begin{cor}\label{cor_dif}
For any $\bsigma\succeq \mb{0}$,
\begin{align}
\sup_{P_{\mb{X}}\colon
\bsigma_{\mb{X}}\preceq \bsigma
}F_0(P_{\mb{X}})
=\sup_{P_{\mb{X}}\colon\textrm{ Gaussian},\bsigma_{\mb{X}}\preceq \bsigma}F_0(P_{\mb{X}}).
\label{e107}
\end{align}
\begin{align}
\sup_{P_{\mb{X}}
}F_0(P_{\mb{X}})
=\sup_{P_{\mb{X}}\colon\textrm{ Gaussian}}F_0(P_{\mb{X}}),
\label{e108}
\end{align}
where $F_0(\cdot)$ is defined in \eqref{e_103}.
In addition, in the non-degenerate case \eqref{e107} is always finite and (up to a translation) uniquely achieved by a Gaussian $P_{\mb{X}}$; \eqref{e108} is finite and (up to a translation) uniquely achieved by a Gaussian $P_{\mb{X}}$ if $\mb{M}$ in \eqref{e_103} is positive definite.
\end{cor}
\begin{rem}
Since the left side of \eqref{e_78} is obviously concave as a function of $\bsigma$, \eqref{e_78} and \eqref{e107} combined shows that the left side of \eqref{e107} is concave in $\bsigma$.
This is not obvious from the definition; in particular it is not clear wether the concavity still holds for non-gaussian $Q_{\mb{X}}$ and $Q_{\mb{Y}_j|\mb{X}}$.
\end{rem}

\section{Gaussian Optimality in the Forward-Reverse Brascamp-Lieb Inequality}\label{sec_frg}
In this section, some notations and terminologies from Sections~\ref{sec_fr} and \ref{sec_gaussian} will be used.
Moreover, consider the following parameters/data:
\begin{itemize}
\item Fix Lebesgue measures $(\mu_j)_{j=1}^m$ and Gaussian measures $(\nu_i)_{i=1}^l$
on $\mathbb{R}$;
\item non-degenerate (Definition~\ref{defn_ND}) linear Gaussian random transformation $(P_{Y_j|\mb{X}})_{j=1}^m$ (where ${\bf X}:=(X_1,\dots,X_l)$) associated with conditional expectation operators $(T_j)_{j=1}^m$;
\item positive $(c_j)$ and $(b_i)$.
\end{itemize}
Given Borel measures $P_{X_i}$ on $\mathbb{R}$, $i=1,\dots,l$,
define
\begin{align}
F_0((P_{X_i}))
:=\inf_{P_{\mb{X}}}\sum_{j=1}^m c_jD(P_{Y_j}\|\mu_j)
-\sum_{i=1}^lb_iD(P_{X_i}\|\nu_i)
\label{e_f0_2}
\end{align}
where the infimum is over Borel measures $P_{\bf X}$ that has $(P_{X_i})$ as marginals.
The aim of this section is to prove the following:
\begin{thm}\label{thm_fr_optimality}
$\sup_{(P_{X_i})}F_0((P_{X_i}))$,
where the supremum is over Borel measures $P_{X_i}$ on $\mathbb{R}$, $i=1,\dots,l$, is achieved by some Gaussian $(P_{X_i})_{i=1}^l$.
\end{thm}
Naturally, one would expect that Gaussian optimality can be established when $(\mu_j)_{j=1}^m$ and $(\nu_i)_{i=1}^l$ are either Gaussian or Lebesgue. We made the assumption that the former is Lebesgue and the latter is Gaussian so that certain technical conditions can be justified conveniently, while the crux of the matter--the tensorization steps can still be demonstrated. More precisely, we have the following observation:
\begin{prop}\label{prop13}
$\sup_{(P_{X_i})}F_0((P_{X_i}))$ is finite and there exist $\sigma_i^2\in (0,\infty)$, $i=1,\dots,l$ such that it equals
\begin{align}
\sup_{(P_{X_i})\colon \mathbb{E}[X_i^2]\le \sigma_i^2}F_0((P_{X_i})).
\label{e159}
\end{align}
\end{prop}
\begin{proof}
when $\mu_j$ is Lebesgue and $P_{Y_j|{\bf X}}$ is non-degenerate,
$D(P_{Y_j}\|\mu_j)=-h(P_{Y_j})\le -h(P_{Y_j}|{\bf X})$ is bounded above (in
terms of the variance of additive noise of $P_{Y_j|{\bf X}}$).
Moreover, $D(P_{X_i}\|\nu_i)\ge 0$ when $\nu_i$ is Gaussian, so $\sup_{(P_{X_i})}F_0((P_{X_i}))<\infty$. Further, choosing $(P_{X_i})=(\nu_i)$ and applying a covariance argument to lower bound the first term in \eqref{e_f0_2} shows that $\sup_{(P_{X_i})}F_0((P_{X_i}))>-\infty$.

To see \eqref{e159}, notice that
\begin{align}
D(P_{X_i}\|\nu_i)&=D(P_{X_i}\|\nu_i')
+\mathbb{E}[\imath_{\nu_i'\|\nu_i}(X)]
\\
&=
D(P_{X_i}\|\nu_i')+D(\nu_i'\|\nu_i)
\\
&\ge D(\nu_i'\|\nu_i)
\end{align}
where $\nu_i'$ is a Gaussian distribution with the same first and second moments as $X_i\sim P_{X_i}$. Thus $D(P_{X_i}\|\nu_i)$ is bounded below by some function of the second moment of $X_i$ which tends to $\infty$ as the second moment of $X_i$ tends to $\infty$. Moreover, as argued in the preceding paragraph the first term in \eqref{e_f0_2} is bounded above by some constant depending only on $(P_{Y_j|{\bf X}})$.
Thus, we can choose $\sigma_i^2>0$, $i=1,\dots,l$ large enough such that if $\mathbb{E}[X_i^2]>\sigma_i^2$ for some of $i$ then $F_0((P_{X_i}))<\sup_{(P_{X_i})}F_0((P_{X_i}))$, irrespective of the choices of $P_{X_1},\dots,P_{X_{i-1}},P_{X_{i+1}},\dots,P_{X_l}$.
Then these $\sigma_1,\dots,\sigma_l$ are as desired in the proposition.
\end{proof}

As we saw in Section~\ref{sec_gaussian}, the regularization ensures that the supremum is achieved, although it might be possible to prove Gaussian \emph{exhaustibility} results by taking limits.
\begin{prop}\label{prop_exist2}
\begin{enumerate}
\item For any $(P_{X_i})_{i=1}^l$, the infimum in \eqref{e_f0_2} is attained.
\item If $(P_{Y_j|X^l})_{j=1}^m$ are non-degenerate (Definition~\ref{defn_ND}), then the supremum in \eqref{e159} is achieved by some $(P_{X_i})_{i=1}^l$.
\end{enumerate}
\end{prop}
\begin{proof}
\begin{enumerate}
\item
For any $\epsilon>0$, by the continuity of measure there exists $K>0$ such that
\begin{align}
P_{X_i}([-K,K])\ge 1-\epsilon/l,\quad i=1,\dots,l.
\end{align}
By the union bound,
\begin{align}
P_{\bf X}([-K,K]^l)\ge 1-\epsilon
\label{e_tight}
\end{align}
wherever $P_{\bf X}$ is a coupling of $(P_{X_i})$.
Now let $P_{\bf X}^{(n)}$, $n=1,2,\dots$ be a such that
\begin{align}
\lim_{n\to\infty}\sum_{j=1}^m c_jD(P_{Y_j}^{(n)}\|\mu_j)
=\inf_{P_{\mb{X}}}\sum_{j=1}^m c_jD(P_{Y_j}\|\mu_j)
\end{align}
where $P_{Y_j}:=T_j^* P_{\bf X}$, $j=1,\dots,m$.
The sequence $(P_{\bf X}^{(n)})$ is tight by \eqref{e_tight},
Thus invoking Prokhorov theorem and by passing to a subsequence, we may assume that $(P_{\bf X}^{(n)})$ converges weakly to some $P_{\bf X}^{\star}$.
Therefore $P_{Y_j}^{(n)}$ converges to $P_{Y_j}^{\star}$ weakly, and by the semicontinuity property in Lemma~\ref{lem28} we have
\begin{align}
\sum_{j=1}^m c_jD(P_{Y_j}^{\star}\|\mu_j)
\le
\lim_{n\to\infty}\sum_{j=1}^m c_jD(P_{Y_j}^{(n)}\|\mu_j)
\end{align}
establishing that $P_{\bf X}^{\star}$ is an infimizer.

\item Suppose $(P_{X_i}\n)_{1\le i\le l,n\ge 1}$ is such that
$\mathbb{E}[X_i^2]\le \sigma_i^2$, $X_i\sim P^{(n)}_{X_i}$,
where $(\sigma_i)$ is as in Proposition~\ref{prop13}
and
\begin{align}
\lim_{n\to\infty}F_0\left((P_{X_i}\n)_{i=1}^l\right)=
\sup_{(P_{X_i})\colon \bsigma_{X_i}\preceq \sigma_i^2}F_0((P_{X_i})_{i=1}^l).
\end{align}
The regularization on the covariance implies that for each $i$, $(P_{X_i}\n)_{n\ge 1}$ is a tight sequence. Thus upon the extraction of subsequences, we may assume that for each $i$, $(P_{X_i}\n)_{n\ge 1}$ converges to some $P_{X_i}^{\star}$,
and a simple truncation and min-max inequality argument (see e.g.~\eqref{e_mono2}) shows that $\mathbb{E}[X_i^2]\le \sigma_i^2$, $X_i\sim P^{\star}_{X_i}$.
Then by Lemma~\ref{lem28},
\begin{align}
\sum_ib_iD(P_{X_i}^{\star}\|\nu_i)
\le \lim_{n\to\infty}\sum_ib_iD(P_{X_i}\n\|\nu_i)
\label{e_semi}
\end{align}
Under the covariance regularization and the non-degenerateness assumption, we showed in Proposition~\ref{prop13} that the value of \eqref{e159} cannot be $+\infty$ or $-\infty$. This implies that we can assume (by passing to a subsequence) that
$P_{X_i}\n\ll \lambda$, $i=1,\dots,l$ since otherwise $F((P_{X_i}))=-\infty$. Moreover, since $\left( \sum_jc_jD(P_{Y_j}\n\|\mu_j)\right)_{n\ge 1}$ is bounded above under the non-degenerateness assumption, the sequence $\left( \sum_ib_iD(P_{X_i}\n\|\nu_i)\right)_{n\ge 1}$ must also be bounded from above, which implies, using \eqref{e_semi}, that
\begin{align}
\sum_ib_iD(P_{X_i}^{\star}\|\nu_i)<\infty.
\end{align}
In particular, we have $P_{X_i}^{\star}\ll \lambda$ for each $i$.
Let $P_{\bf X}^{\star}$ be an infimizer as in Part~1) for marginals $(P_{X_i}^{\star})$.
In view of Lemma~\ref{lem_couple}, by possibly passing to subsequences we can assume that each $(P_{X_i}\n)_{i=1}^l$ admits a coupling $P_{\bf X}\n$, such that $P_{\bf X}\n\to P_{\bf X}^{\star}$ weakly as $n\to\infty$.
In the non-degenerate case the output differential entropy is weakly continuous in the input distribution under the covariance constraint (see for example \cite[Proposition~18]{geng2014yanlin}), which establishes that
\begin{align}
\inf_{P_{\bf X}\colon S_i^*P_{\bf X}=P_{X_i}^{\star}}\sum_j c_jD(T_j^*P_{\bf X}\|\mu_j)
&=\sum_j c_jD(P_{Y_j}^{\star}\|\mu_j)
\\
&=\lim_{n\to\infty} \sum_j c_jD(P_{Y_j}\n\|\mu_j)
\\
&\ge \lim_{n\to\infty}\inf_{P_{\bf X}\colon S_i^*P_{\bf X}=P_{X_i}\n} \sum_j c_jD(T_j^*P_{\bf X}\|\mu_j)
\label{e169}
\end{align}
Thus \eqref{e_semi} and \eqref{e169} show that $(P_{X_i}^{\star})$ is in fact a maximizer.
\end{enumerate}
\end{proof}

\begin{lem}\label{lem_couple}
Suppose that for each $i=1,\dots,l$ ($l\ge 2$), $P_{X_i}$ is a Borel measure on $\mathbb{R}$ and $P^{(n)}_{X_i}$ converges weakly to $P_{X_i}$ as $n\to\infty$. If $P_{\bf X}$ is a coupling of $(P_{X_i})_{1\le i\le l}$, then, upon extraction of a subsequence, there exist couplings $P^{(n)}_{{\bf X}}$ for $(P^{(n)}_{X_i})_{1\le i\le l}$ which converge weakly to $P_{\bf X}$ as $n\to\infty$.
\end{lem}

\begin{rem}
We will use Lemma~\ref{lem_couple} to establish that the infimum of an upper semicontinuous (w.r.t.~the joint distribution) functional over couplings is also upper semicontinuous (w.r.t~the marginal distributions),
which is the key to the proof of Proposition~\ref{prop_exist2}.
Another application is to prove the weak continuity of the optimal transport cost (the lower semicontinuity part being trivial) in the theory of optimal transportation, which may also be proved using the ``stability of optimal transport'' in \cite[Theorem~5.20]{villani2008optimal}.
However, we note that the approach in \cite[Theorem~5.20]{villani2008optimal} relies on cyclic monotonicity, and hence cannot be extended to the setting of general upper semicontinuous functionals such as in Corollary~\ref{cor_stability}.
\end{rem}
\begin{cor}[Weak stability of optimal coupling]\label{cor_stability}
(In the case of non-degenerate $(P_{Y_j|\mb{X}})$) Suppose for each $n\ge 1$,
$P_{X_i}\n$ is a Borel measure on $\mathbb{R}$, $i=1,\dots,l$,
whose second moment is bounded by $\sigma_i^2<\infty$.
Assume that $P_{\bf X}\n$ is a coupling of $(P_{X_i}\n)$ that minimizes $\sum_{j=1}^lc_jD(P_{Y_j}\n\|\mu_j)$. If $P_{\bf X}\n$ converges weakly to some $P_{\bf X}^{\star}$, then $P_{\bf X}^{\star}$ minimizes $\sum_{j=1}^lc_jD(P_{Y_j}^{\star}\|\mu_j)$ given the marginals $(P_{X_i}^{\star})$.
\end{cor}
\begin{proof}
Lemma~\ref{lem_couple} and
the fact that the differential entropy of a non-degenerate Gaussian channel is weakly semicontinuous with respect to the input distribution under a moment constraint
(see e.g.~\cite[Proposition~18]{geng2014yanlin}, \cite[Theorem~7]{wu2012functional}, or \cite[Theorem~1, Theorem~2]{godavarti2004convergence})
imply that there exists a strictly increasing sequence of positive integers $(n_k)_{k\ge 1}$ such that
\begin{align}
\min_{P_{\bf X}\colon S_i^*P_{\bf X}=P_{X_i}^{\star}}
\sum_j c_j D(T_j^*P_{\bf X}\|\mu_j)
\ge
\limsup_{k\to\infty}
\min_{P_{\bf X}\colon S_i^*P_{\bf X}=P_{X_i}^{(n_k)}}
\sum_j c_j D(T_j^*P_{\bf X}\|\mu_j).
\label{e170}
\end{align}
On the other hand, the assumption on the convergence of the optimal couplings and
Lemma~\ref{lem28} imply that
\begin{align}
\sum_j c_j D(T_j^*P_{\bf X}^{\star}\|\mu_j)
\le
\liminf_{k\to\infty}
\min_{P_{\bf X}\colon S_i^*P_{\bf X}=P_{X_i}^{(n_k)}}
\sum_j c_j D(T_j^*P_{\bf X}\|\mu_j)
\label{e171}
\end{align}
so the left side of \eqref{e171} is no larger than the left side of \eqref{e170}.
\end{proof}
\begin{proof}[Proof of Lemma~\ref{lem_couple}]
For each integer $k\ge 1$, define the random variable $W_i^{[k]}:=\phi_k(X_i)$ where $\phi_k$ is the following ``dyadic quantization function'':
\begin{align}
\phi_k\colon x\in\mathbb{R}\mapsto
\left\{
\begin{array}{cc}
  \lfloor 2^kx\rfloor & |x|\le k,\,x\notin 2^{-k}\mathbb{Z}; \\
  {\sf e} & \textrm{otherwise},
\end{array}
\right.
\label{e_quant}
\end{align}
and let ${\bf W}^{[k]}:=(W_i^{[k]})_{i=1}^l$.
Denote by $\mathcal{W}^{[k]}:=\{-k2^k,\dots,k2^k-1,{\sf e}\}$ the alphabet of $W_i^{[k]}$.

For simplicity of the presentation, we shall assume that the set of ``dyadic points'' has measure zero:
\begin{align}
P_{X_i}(\bigcup_{k=1}^{\infty}2^{-k}\mathbb{Z})=0,\quad i=1,\dots,l.
\label{e173}
\end{align}
This is not an essential restriction, since the property of $2^{-k}\mathbb{Z}$ used in the proof is that it is a countable dense subset of $\mathbb{R}$. In general, $P_{X_i}$ can have a positive mass only on countably many points, and in particular there exists a point of zero mass in any open interval. Thus, we can always choose a countable dense subset of $\mathbb{R}$ with $P_{X_i}$ measure zero ($i=1$) and appropriately modify \eqref{e_quant} with points in this set for quantization instead.

Since $P_{X_i}^{(n)}\to P_{X_i}$ weakly and the assumption in the preceding paragraph precluded any positive mass on the quantization boundaries under $P_{X_i}$, for each $k\ge 1$ there exists some $n:=n_k$ large enough such that
\begin{align}
P^{(n)}_{W_i^{[k]}}(w)
&\ge (1-\frac{1}{k})
P_{W_i^{[k]}}(w),\label{e_162}
\end{align}
for each $i$ and $w\in\mathcal{W}^{[k]}$.
Now define a coupling $P_{{\bf W}^{[k]}}^{(n)}$ compatible with the $\left(P_{W_i^{[k]}}^{(n)}\right)_{i=1}^l$ induced by $\left(P_{X_i}^{(n)}\right)_{i=1}^l$, as follows:
\begin{align}
P_{{\bf W}^{[k]}}^{(n)}
:=(1-\frac{1}{k})P_{{\bf W}^{[k]}}
+k^{l-1}\prod_{i=1}^l
\left(
P_{W_i^{[k]}}^{(n)} - (1-\frac{1}{k}) P_{W_i^{[k]}}
\right).
\label{e_163}
\end{align}
Observe that this is a well-defined probability measure because of \eqref{e_162}, and indeed has $\left(P_{W_i^{[k]}}^{(n)}\right)_{i=1}^l$ as the marginals.
Moreover, by triangle inequality we have the following bound on the total variation distance
\begin{align}
\left|P_{{\bf W}^{[k]}}^{(n)}-P_{{\bf W}^{[k]}}\right|
\le \frac{2}{k}.
\label{e164}
\end{align}
Next, construct\footnote{We use $P|_{\mathcal{A}}$ to denote the restriction of a probability measure $P$ on measurable set $\mathcal{A}$, that is, $P|_{\mathcal{A}}(\mathcal{B}):=P(\mathcal{A}\cap\mathcal{B})$ for any measurable $\mathcal{B}$.} $P_{{\bf X}}^{(n)}$:
\begin{align}
P_{{\bf X}}^{(n)}
:=\sum_{w^l\in\mathcal{W}^{[k]}\times\dots\times\mathcal{W}^{[k]}}\frac{P_{{\bf W}^{[k]}}^{(n)}\left(
w^l\right)}
{\prod_{i=1}^l P_{W^{[k]}_i}^{(n)}(w_i)}
\prod_{i=1}^l P_{X_i}^{(n)}|_{\phi^{-1}_k(w_i)}.
\end{align}
Observe that the $P_{{\bf X}}^{(n)}$ so defined is compatible with the $P_{{\bf W}^{[k]}}^{(n)}$ defined in \eqref{e_163}, and indeed has $(P_{X_i}^{(n)})_{i=1}^l$ the marginals.
Since $n:=n_k$ can be made increasing in $k$, we have constructed the desired sequence $(P_{\bf X}^{(n_k)})_{k=1}^{\infty}$ converging weakly to $P_{\bf X}$. Indeed, for any bounded open \emph{dyadic cube}\footnote{That is, a cube whose corners have coordinates being multiples of $2^{-k}$ where $k$ is some integer.} $\mathcal{A}$, using \eqref{e164} and the assumption \eqref{e173}, we conclude
\begin{align}
\liminf_{k\to\infty}P_{\bf X}^{(n_k)}(\mathcal{A})
\ge P_{\bf X}(\mathcal{A}).
\label{e166}
\end{align}
Moreover, since bounded open dyadic cubes form a countable basis of the topology in $\mathbb{R}^l$, we see \eqref{e166} actually holds for any open set $\mathcal{A}$ (by writing $\mathcal{A}$ as a countable union of dyadic cubes, using the continuity of measure to pass to a finite disjoint union, and then apply \eqref{e166}), as desired.
\end{proof}
Next, we need the following tensorization result:
\begin{lem}\label{lem_inner_tens}
Fix $(P_{X_i\1})$, $(P_{X_i\2})$, $(\mu_j)$, $(T_j)$, $(c_j)\in[0,\infty)^m$, and let $S_j$ be induced by coordinate projections. Then
\begin{align}
\inf_{P_{{\bf X}^{(1,2)}}\colon S_i^{*\otimes2}P_{{\bf X}^{(1,2)}}
=P_{X_i\1}\times P_{X_i\2}}
\sum_{j=1}^m c_jD(P_{Y_j^{(1,2)}}\|\mu_j^{\otimes 2})
= \sum_{t=1,2}\sum_{j=1}^m c_j \inf_{P_{\mb{X}^{(t)}}\colon S_i^*P_{\mb{X}^{(t)}}=P_{X_i^{(t)}}}D(P_{Y_j^{(t)}}\|\mu_j)
\end{align}
where for each $j$,
\begin{align}
P_{Y_j^{(1,2)}}:=T_j^{*\otimes 2}P_{{\bf X}^{(1,2)}}
\end{align}
on the left side and
\begin{align}
P_{Y_j^{(t)}}:=T_j^{*\otimes 2}P_{{\bf X}^{(t)}}
\end{align}
on the right side, $t=1,2$.
\end{lem}
\begin{proof}
We only need to prove the nontrivial $\ge$ part. For any $P_{{\bf X}^{(1,2)}}$ on the left side, choose $P_{\mb{X}^{(t)}}$ on the right side by marginalization. Then
\begin{align}
D(P_{Y_j^{(1,2)}}\|\mu_j^{\otimes 2})
-\sum_t D(P_{Y_j^{(t)}}\|\mu_j)
&=I(Y_j\1;Y_j\2)
\\
&\ge0
\end{align}
for each $j$.
\end{proof}
We are now in the position of proving the main result of this section.
\begin{proof}[Proof of Theorem~\ref{thm_fr_optimality}]
\begin{enumerate}
  \item Assume that $(P_{X_i\1})$ and $(P_{X_i\2})$ are maximizers of $F_0$ (possibly equal).
      Let $P_{X_i^{1,2}}:=P_{X_i\1}\times P_{X_i\2}$.
      Define \begin{align}
      \mb{X}^+&:=\frac{1}{\sqrt{2}}\left(\mb{X}\1
      +\mb{X}\2\right);
      \\
      \mb{X}^-&:=\frac{1}{\sqrt{2}}\left(\mb{X}\1
      -\mb{X}\2\right).
      \end{align}
      Define $(Y_j^+)$ and $(Y_j^-)$ analogously. Then
      $Y_j^+|\{\mb{X}^+=\mb{x}^+,\mb{X}^-=\mb{x}^-\}
      \sim Q_{Y_j|\mb{X}=\mb{x}^+}$ is independent of $\mb{x}^-$ and
      $Y_j^-|\{\mb{X}^+=\mb{x}^+,\mb{X}^-=\mb{x}^-\}
      \sim Q_{Y_j|\mb{X}=\mb{x}^-}$ is independent of $\mb{x}^+$.
  \item Next we perform the same algebraic expansion as in the proof of tensorization:
      \begin{align}
      \sum_{t=1}^2F_0(\left(P_{X_i^{(t)}}\right))
      &=\inf_{P_{\mb{X}^{(1,2)}}\colon S_j^{*\otimes2}P_{\mb{X}^{(1,2)}}=P_{X_j^{(1,2)}}}
      \sum_j c_jD(P_{Y_j^{(1,2)}}\|\mu_j^{\otimes2})
      -\sum_ib_iD(P_{X_i^{(1,2)}}\|\nu_i^{\otimes2})
      \label{e_rbl_inner_tens}
      \\
      &=\inf_{P_{\mb{X}^+\mb{X}^-}\colon S_j^{*\otimes2}P_{\mb{X}^+\mb{X}^-}
      =P_{X_j^+X_j^-}}
      \sum_j c_jD(P_{Y_j^+Y_j^-}\|\mu_j^{\otimes2})
      -\sum_ib_iD(P_{X_i^+X_i^-}\|\nu_i^{\otimes2})
      \\
      &\le\inf_{P_{\mb{X}^+\mb{X}^-}\colon S_j^{*\otimes2}P_{\mb{X}^+\mb{X}^-}
      =P_{X_j^+X_j^-}}
      \sum_j c_j\left[D(P_{Y_j^+}\|\mu_j)
      +D(P_{Y_j^-|\mb{X}^+}\|\mu_j|P_{\mb{X}^+})
      \right]
      \nonumber\\
      &\quad-\sum_ib_i
      \left[D(P_{X_i^+}\|\nu_i)
      +D(P_{X_i^-|X_i^+}\|\nu_i|P_{X_i^+})
      \right]
      \label{e_rbl_markov}
      \\
      &\le
      \sum_j c_j\left[D(P^{\star}_{Y_j^+}\|\mu_j)
      +D(P^{\star}_{Y_j^-|\mb{X}^+}\|\mu_j|P^{\star}_{\mb{X}^+})
      \right]
      \nonumber\\
      &\quad-\sum_ib_i
      \left[D(P^{\star}_{X_i^+}\|\nu_i)
      +D(P^{\star}_{X_i^-|\mb{X}^+}\|\nu_i|P^{\star}_{\mb{X}^+})
      \right]
      \label{e_rbl_inst}
      \\
      &=F_0(\left(P^{\star}_{X_i^+}\right))
      +\int F_0(\left(P^{\star}_{X_i^-|\mb{X}^+}\right)){\rm d}P^{\star}_{\mb{X}^+}
      \label{e_rbl64}
      \\
      &\le\sum_{t=1}^2F_0(\left(P_{X_i^{(t)}}\right))
      \label{e_rbl65}
      \end{align}
  where
  \begin{itemize}
    \item   \eqref{e_rbl_inner_tens} uses Lemma~\ref{lem_inner_tens}.
    \item  \eqref{e_rbl_markov} is because of the Markov chain $Y_j^+-\mb{X}^+-Y_j^-$ (for any coupling).
    \item In \eqref{e_rbl_inst} we selected a particular instance of coupling $P_{\bf X^+X^-}$, constructed as follows: first we select an optimal coupling $P_{\bf X^+}$ for given marginals $(P_{X_i^+})$. Then, for any $\mb{x}^+=(x_i^+)_{i=1}^l$,
        let $P_{\mb{X}^-|\mb{X}^+=x^+}$ be an optimal coupling of $(P_{X_i^-|X_i^+=x_i^+})$.
        \footnote{Here we need to justify that we can select optimal coupling $P_{\mb{X}^-|\mb{X}^+=\mb{x}^+}$ in a way that $P_{\mb{X}^-|\mb{X}^+}$ is indeed a regular conditional probability distribution, or equivalently, $P_{\mb{X}^-|\mb{X}^+=\mb{x}^+}$ is Borel measurable, where we endow the weak topology on the space of probability measures.
        This is justified by two observations: (a) $\mb{x}^+\mapsto (P_{X_i^-|\mb{X}^+=\mb{x}^+})$ is Borel measurable; (b) the marginalization map $\phi\colon \mathcal{O}\to \prod_i\mathcal{P}(\mathcal{X}_i),P_{\bf X}\mapsto (P_{X_i})$ admits a Borel right inverse, where $\mathcal{O}$ is the set of optimal couplings (w.r.t.~the relative entropy functional) whose marginals satisfy the second moment constraint. Part (a) is equivalent to the regularity of $P_{X_i^-|\mb{X}^+}$, and the existence of such a regular conditional distribution is guaranteed for joint distributions on Polish spaces (whose measurable space structure is isomorphic to a standard Borel space); see e.g.~\cite{verdubook}. Part (b) is justified by measurable selection theorems (see e.g.~references in \cite[Corollary~5.22]{villani2008optimal}). In particular, a similar argument as \cite[Corollary~5.22]{villani2008optimal} can also be applied here, since Corollary~\ref{cor_stability} implies that $\mathcal{O}$ is closed and the pre-images $\phi^{-1}(\left(P_{X_i}\right))$ are all compact, and Proposition~\ref{prop_exist2}.1 justifies that $\phi$ is onto.
        }
        With this construction, it is apparent that $X_i^+-\mb{X}^+-X_i^-$ and hence
        \begin{align}
        D(P_{X_i^-|X_i^+}\|\nu_i|P_{X_i^+})
        =D(P_{X_i^-|\mb{X}^+}\|\nu_i|P_{\mb{X}^+}).
        \end{align}
  \item  \eqref{e_rbl64} is because in the above we have constructed the coupling optimally.
  \item \eqref{e_rbl65} is because $(P_{X_i}^{(t)})$ maximizes $F_0$, $t=1,2$.
  \end{itemize}
  \item Thus in the expansions above, equalities are attained throughout. Using the differentiation technique as in the case of forward inequality, for almost all $(b_i)$, $(c_j)$, we have
      \begin{align}
      D(P_{X_i^-|X_i^+}\|\nu_i|P_{X_i^+})
      &=D(P_{X_i^+}\|\nu_i)
      \\
      &=D(P_{X_i^-}\|\nu_i),\quad \forall i
      \end{align}
      where the last equality is because by symmetry we can perform the algebraic expansions in a different way to show that $(P_{X_i^-})$ is also a maximizer of $F_0$. Then $I(X_i^+;X_i^-)=0$, which, combined with $I(X_i\1;X_i\2)$, shows that $X_i\1$ and $X_i\2$ are Gaussian with the same covariance. Lastly, using Lemma~\ref{lem_inner_tens} and the doubling trick one can show that the optimal coupling is also Gaussian.
\end{enumerate}
\end{proof}

\section{Consequences of Gaussian Optimality}\label{sec_consequence}
In this section, we demonstrate several implications of the entropic Gaussian optimality results in Sections~\ref{sec_gaussian}-\ref{sec_frg}
to functional inequalities, transportation-cost inequalities, and network information theory.
\subsection{Brascamp-Lieb Inequality with Gaussian Random transformations: an Information-Theoretic Proof}\label{sec_consequence_BL}
We give a simple proof of an extension of the Brascamp-Lieb inequality using Corollary~\ref{cor_dif}. That is, we give an information-theoretic proof of the following result:

\begin{thm}\label{thm6}
Suppose $\mu$ is either a Gaussian measure or the Lebesgue measure on $\mathcal{X}=\mathbb{R}^n$. Let $Q_{\mb{Y}_j|\mb{X}}$ be Gaussian random transformations where $\mathcal{Y}_j=\mathbb{R}^{n_j}$, and $c_j\in(0,\infty)$ for $j\in\{1,\dots,m\}$.
For any non-negative measurable functions $f_j\colon\mathcal{Y}_j\to \mathbb{R}$, $j\in\{1,\dots,m\}$, define
\begin{align}
H(f_1,\dots,f_m)
=\log\int\exp\left(\sum_{j=1}^m\mathbb{E}[\log f_j(\mb{Y}_j)|\mb{X}=\mb{x}]\right){\rm d}\mu(\mb{x})
-\log\prod_{j=1}^m\|f_j\|_{\frac{1}{c_j}}
\label{e109}
\end{align}
where the norm $\|f_j\|_{\frac{1}{c_j}}$ is with respect to the Lebesgue measure and the expectation is with respect to $Q_{\mb{Y}_j|\mb{X=x}}$. Then
\begin{align}
\sup_{f_1,\dots,f_m}H(f_1,\dots,f_m)
=\sup_{\textrm{Gaussian }f_1,\dots,f_m}H(f_1,\dots,f_m).
\label{e110}
\end{align}
Moreover, if $\mu$ is Gaussian, $(Q_{\mb{Y}_j|\mb{X}})$ is non-degenerate in the sense of Definition~\ref{defn_ND}, then \eqref{e110} is finite and uniquely attained by a set of Gaussian functions.
\end{thm}
\begin{proof}
We only prove the case of Gaussian $\mu$, since the proof for the Lebesgue case is similar.
By the translation and scaling invariance,
it suffices to consider the case where $\mu$ is centered Gaussian with density $\exp(-\mb{x}^{\top}\mb{M}\mb{x})$ for some $w\in\mathbb{R}$ and ${\bf M}\succeq \mb{0}$.
Then
\begin{align}
\imath_{\lambda\|\mu}(\mb{x})=\mb{x}^{\top}\mb{Mx}
\end{align}
where $\lambda$ denotes the Lebesgue measure on $\mathbb{R}^n$.
Therefore,
\begin{align}
D(P_{\bf X}||\mu)
+\sum_{j=1}^m c_j h(P_{{\bf Y}_j})
&=-h(P_{\bf X})+\sum_{j=1}^m c_j h(P_{{\bf Y}_j})
+\mathbb{E}[\imath_{\lambda\|\mu}(\mb{\hat{X}})]
\\
&=-h(P_{\bf X})+\sum_{j=1}^m c_j h(P_{{\bf Y}_j})
+\mathbb{E}[\mb{\hat{X}^{\top}M\hat{X}}].
\label{e94}
\end{align}
where $\hat{X}\sim P_{\mb{X}}$. Then \eqref{e110} is established by
\begin{align}
\sup_{f_1,\dots,f_m}H(f_1,\dots,f_m)
&=-\inf_{P_{\mb{X}}}\left\{D(P_{\bf X}||\mu)
+\sum_{j=1}^m c_j h(P_{{\bf Y}_j})\right\}
\label{e114}
\\
&=-\inf_{\textrm{Gaussian }P_{\mb{X}}}\left\{D(P_{\bf X}||\mu)
+\sum_{j=1}^m c_j h(P_{{\bf Y}_j})\right\}
\label{e115}
\\
&=\sup_{\textrm{Gaussian }f_1,\dots,f_m}H(f_1,\dots,f_m)
\label{e116}
\end{align}
where
\begin{itemize}
  \item \eqref{e114} is from Remark~\ref{thm_1_ext};
  \item \eqref{e115} is from \eqref{e94} and Corollary~\ref{cor_dif};
  \item \eqref{e116} is essentially a ``Gaussian version'' of the equivalence of the two inequalities in Theorem~\ref{thm_1}, which is easily shown with the same steps in the proof of Theorem~\ref{thm_1}, noting that
      \eqref{e_f} sends Gaussian distributions to Gaussian functions and \eqref{e12} sends Gaussian functions to Gaussian distributions.
\end{itemize}
In the case of Gaussian $\mu$ and non-degenerate $(Q_{\mb{Y}_j|\mb{X}})$, Corollary~\ref{cor_dif} implies that \eqref{e115} is finitely attained by a unique Gaussian $P_{\mb{X}}$. By Remark~\ref{thm_1_ext}, \eqref{e110} is finitely attained by a unique set of Gaussian functions.
\end{proof}

\begin{rem}
The proof of the Brascamp-Lieb inequality by Carlen and Erausquin \cite{carlen2009subadditivity} also relies on dual information-theoretic formulation. However, their proof of the differential entropy inequality uses a different approach based on superadditivity of Fisher information. That approach applies to the case where each $Q_{\mb{Y}_j|\mb{X}}$ is a deterministic rank-one linear map, and it requires the problem to be first reduced to a special case called the geometric Brascamp-Lieb inequality (proposed by K.~Ball \cite{ball1989volumes}).
\end{rem}

\subsection{Multi-variate Gaussian Hypercontractivity}\label{sec_hypercontractivity}
In this section we show a multivariate extension of Gaussian hypercontractivity.
An $m$-tuple of random variables $(X_1,\dots,X_m)\sim Q_{X^m}$ is said to be $(p_1,\dots,p_m)$-hypercontractive for $p_l\in[1,\infty]$, $l\in\{1,\dots,m\}$ if
\begin{align}\label{hyper1}
\mathbb{E}\left[\prod_{l=1}^m f_l(X_l)\right]\le \prod_{l=1}^m\|f_l(X_l)\|_{p_l}
\end{align}
for all bounded real-valued measurable functions $f_l$ defined on $\mathcal{X}_l$, $l\in\{1,\dots,m\}$. Define the \emph{hypercontractivity region}\footnote{Note that this definition is similar to the hypercontractivity ribbon defined in \cite{anan_13} but without taking the H\"{o}lder conjugate of one of the two exponent, which is more symmetrical and convenient in the multivariate setting.}
\begin{align}
\mathcal{R}(Q_{X^m}):=\{(p_1,\dots,p_m)\in\mathbb{R}_+^m\colon(X_1,\dots,X_m)\textrm{ is } (p_1,\dots,p_m)\textrm{ hypercontractive}\}.
\end{align}

By Theorem~\ref{thm_1}, the inequality \eqref{hyper1} is true if and only if
\begin{align}\label{e85}
D(P_{X^m}||Q_{X^m})\ge \sum_{j=1}^m \frac{1}{p_j}D(P_{X_j}||Q_{X_j})
\end{align}
holds for any $P_{X^m}\ll Q_{X^m}$. In fact, \cite{anan_13} showed that \eqref{e85} is equivalent to the following (which hinges on the fact that the constant term $d=0$ in \eqref{e85}):
\begin{align}\label{e121}
I(U;X^m)\ge \sum_{j=1}^m \frac{1}{p_j}I(U;X_j),\quad
\forall U.
\end{align}

In the case of Gaussian $Q_{X^m}$, by Theorem~\ref{thm1} the inequality \eqref{e121} holds if it holds for all $U$ jointly Gaussian with $X^m$ and having dimension at most $m$. When restricted to such $U$, \eqref{e85} becomes an inequality involving the covariance matrices, and some elementary computations show that:
\begin{prop}\label{prop8}
Suppose $Q_{X^m}=\mathcal{N}(\mb{0},\mb{\Sigma})$ where $\bf \Sigma$ is a positive semidefinite matrix whose diagonal values are all $1$. Then $p^m\in \mathcal{R}(Q_{X^m})$ if and only if
\begin{align}\label{e87}
\mb{P}\succcurlyeq\mb{\Sigma}
\end{align}
where $\mb{P}$ is a diagonal matrix with $p_j$, $j\in\{1,\dots,m\}$ as its diagonal entries.
\end{prop}
\begin{proof}
Let $\mb{A}$ be the covariance matrix of $X^m$ conditioned on $U$, and put $\mb{C}:=\mb{P}^{-1}$. We will use the lowercase letters such as $a_1,\dots,a_m$ to denote the diagonal entries of the corresponding matrices. Then in view of \eqref{e121}, we see that the goal is to show that \eqref{e87} is a necessary and sufficient condition for
\begin{align}
\log\frac{|\bsigma|}{|{\bf A}|}
\ge \sum_j c_j \log\frac{1}{a_i},
\quad\forall\, \mb{0}\preceq\mb{A}\preceq\bsigma.
\label{e123}
\end{align}
Let us first assume that $\mb{\Sigma}$ is invertible.
Define $\bf D:=\bsigma-A$, then rewrites as
\begin{align}
|\mb{I}-\bsigma^{-1}{\bf D}|
\le \prod_j (1-d_j)^{c_j},
\quad\forall\, \mb{0}\preceq\mb{D}\preceq\bsigma.
\label{e124}
\end{align}
If \eqref{e87}, then $\mb{C}\preceq \bsigma^{-1}$, and we have
\begin{align}
|\mb{I}-\bsigma^{-1}{\bf D}|
&\le
|\mb{I}-\mb{C}{\bf D}|
\\
&\le\prod_j (1-c_jd_j)
\label{e126}
\\
&\le \prod_j (1-d_j)^{c_j}
\label{e127}
\end{align}
where \eqref{e126} is because $\{1-c_jd_j\}$ is the diagonal entries of $\mb{I}-\mb{C}{\bf D}$, which is majorized by the eigenvalues of that matrix. Inequality \eqref{e127} is because $\mb{C}\preceq \bsigma^{-1}$ and $\mb{D}\preceq \bsigma$ imply the nonnegativity of $1-c_jd_j$.

Conversely, if \eqref{e124} holds, we can apply Taylor expansions on both sides of \eqref{e124}:
\begin{align}
1-\Tr(\bsigma^{-1}{\bf D})+o(\|\mb{D}\|)
\le 1-\sum_j c_j d_j+o(\|\mb{D}\|),
\end{align}
where $\|\mb{D}\|$ can be chosen as, say, the trace norm. Comparing the first order terms, we see that
\begin{align}
\Tr((\bsigma^{-1}-\mb{C})\mb{D})\succeq \mb{0}
\end{align}
must hold for any $\mb{0}\preceq\mb{D}\preceq\bsigma$.
Therefore \eqref{e87} holds.

More generally, if $\bsigma$ is not necessarily invertible, we can consider the inverse of its restriction $\bsigma|_{V}$ where $V$ is its column space. Then the above arguments still carry through with trivial modifications, and we can show that \eqref{e123} holds if and only if
\begin{align}
&{\bsigma|_V}^{-1}\succeq \mb{C}|_V
\\
\Longleftrightarrow \,
&{\bsigma|_V}^\half \mb{C}|_V{\bsigma|_V}^\half
\preceq \mb{I}
\\
\Longleftrightarrow \,
&{\bsigma}^\half \mb{C}{\bsigma}^\half
={\bsigma}^\half \mb{P}^{-1}{\bsigma}^\half
\preceq \mb{I}
\\
\Longleftrightarrow \,
&\mb{P}^{-\half}\bsigma \mb{P}^{-\half}\preceq \mb{I}
\label{e133}
\\
\Longleftrightarrow \,
&\bsigma\preceq \mb{P}
\end{align}
where \eqref{e133} is because ${\bsigma}^\half \mb{P}^{-1}{\bsigma}^\half$ and $\mb{P}^{-\half}\bsigma \mb{P}^{-\half}$ have the same set of eigenvalues.
\end{proof}
\begin{rem}
When $m=2$, Proposition~\ref{prop8} reduces to Nelson's hypercontractivity theorem for a pair of Gaussian scalar random variables $X_1$ and $X_2$, that is, $(p_1,p_2)\in\mathcal{R}(Q_{X^2})$ if and only if
\begin{align}
(p_1-1)(p_2-1)\ge\rho^2(X_1;X_2),
\end{align}
where $\rho^2$ denotes the squared Pearson correlation.
\end{rem}

\subsection{${\rm T}_2$ Inequality for Gaussian Measures}\label{sec_talagrand}
Consider the Euclidean space endowed with $\ell_2$-norm $(\mathbb{R}^n,\|\cdot\|_2)$. Talagrand \cite{talagrand1996transportation} showed that the standard Gaussian measure $Q=\mathcal{N}(\mb{0},\mb{I}_n)$ satisfies the ${\rm T}_2(1)$ inequality (see Definition~\ref{defn_tp}). Below we give a new proof using the Gaussian optimality in the forward-reverse Brascamp-Lieb inequality.
By the tensorization of ${\rm T}_2$ inequality \cite{marton1986simple}, it suffices to prove the $n=1$ case.
By continuity, it suffices to prove ${\rm T}_2(\lambda)$ for any $\lambda>1$.
Moreover, as one can readily check, when $\lambda>1$ and $Q=\mathcal{N}(0,1)$, \eqref{e_tp_defn} is satisfied for all Gaussian $P$ (in which case the optimal coupling in \eqref{e_tp_defn} is also Gaussian), so it suffices to prove Gaussian extremisability in \eqref{e_tp_defn}.

Let $\lambda\in(1,+\infty)$, $b_2\in(0,+\infty)$, and
\begin{align}
F_0(P_{X_1},P_{X_2}):=\inf_{P_{X_1X_2}}\mathbb{E}[|X_1-X_2|^2]
-2\lambda D(P_{X_1}\|Q)-b_2D(P_{X_2}\|Q),
\label{e255}
\end{align}
where the infimum is over coupling $P_{X_1X_2}$ of $P_{X_1}$ and $P_{X_2}$.
Using the rotation invariance argument/doubling trick in the proof of Theorem~\ref{thm_fr_optimality},
we can show that if $(P_{X_1},P_{X_2})$ maximizes $F_0(P_{X_1},P_{X_2})$, then they must be Gaussian, and the optimal coupling $P_{X_1X_2}$ is also Gaussian\footnote{In Theorem~\ref{thm_fr_optimality} the first term of the objective function is the infimum of the relative entropy, rather than the infimum of the expectation of a quadratic cost function. However, the argument in Theorem~\ref{thm_fr_optimality} also works in the latter case, since the expectation functional has a similar tensorization property, and the quadratic cost function also has a rotational invariance property.}. Letting $b_2\to\infty$, we see the Gaussian optimality in the ${\rm T}_2$ inequality.
To make the above argument rigorous, we need to take care of two technical issues:
\begin{description}
\item[(a)] We approximated the functional \eqref{e_119} with $b_2D(P_{X_2}\|Q)$ and let $b_2\to\infty$, but we want that Gaussian optimality continue to hold when the last term in \eqref{e255} is exactly \eqref{e_119}.
\item[(b)] We want to show the existence of a maximizer $(P_{X_1},P_{X_2})$ for \eqref{e255}.
\end{description}
While it might be possible to provide a formal justification of the limit argument in (a),
a slicker way is to circumvent it by directly working with \eqref{e_119} instead of $b_2D(P_{X_2}\|Q)$ with $b_2\to\infty$.
From the tensorization of the functional \eqref{e255}, it is relatively easy to distill the tensorization of the ${\rm T}_2$ inequality (see also \cite{marton1986simple} for a direct proof of tensorization of ${\rm T}_2$ inequality), and then use the rotation invariance argument/doubling trick to conclude Gaussian optimality.

As for (b), note that if $b_2D(P_{X_2}\|Q)$ is replaced with \eqref{e_119}, we want to show the existence of a maximizer $(P_{X_1},Q)$ for \eqref{e255}. If $\sigma^2:=\mathbb{E}[|X_1|^2]$, then
\begin{align}
\inf_{P_{X_1X_2}}\mathbb{E}[|X_1-X_2|^2]
&\le
\sigma^2+
\mathbb{E}[|X_2|^2]
\\
&=\sigma^2+1.
\end{align}
On the other hand,
\begin{align}
D(P_{X_1}\|Q)=
D(P_{X_1}\|\mathcal{N}(0,\sigma^2))
+\mathbb{E}[\imath_{\mathcal{N}(0,\sigma^2)\|Q}(X_1)]
=\frac{\sigma^2}{2}+o(\sigma^2).
\end{align}
Therefore, if $\lambda>1$ and $(P_{X_1}^{(t)})_{t=1}^{\infty}$ is a supremizing sequence, then $(P_{X_1}^{(t)})_{t=1}^{\infty}$ must have bounded second moment, hence must be tight. Thus Prokhorov Theorem implies the existence of a subsequence weakly converging to some $P_{X_1}^{\star}$, and a semicontinuity argument similar to Proposition~\ref{prop_exist2} shows that $P_{X_1}^{\star}$ is in fact a maximizer.

\subsection{Wyner's Common Information for $m$ Dependent Random Variables}
Wyner's common information for $m$ dependent random variables $X_1,\dots,X_m$ is commonly defined as \cite{xu2013wyner}
\begin{align}
C(X^m):=\inf_U I(U;X^m)
\end{align}
where the infimum is over $P_{U|X^m}$ such that $X_1,\dots,X_m$ are independent conditioned on $U$. Previously, to the best of our knowledge, the common information for $m$ Gaussian scalars $X_1,\dots,X_m$ could only be obtained in the special case where the correlation coefficient between $X_i$ and $X_j$ are equal for all $1\le i,j\le m$ \cite[Corollary~1]{xu2013wyner} via a different approach.

Using Theorem~\ref{thm1} and setting $\mb{X}\leftarrow X^m$ and $\mb{Y}_j\leftarrow X_j$, we immediately obtain the following characterization of the multivariate common information for Gaussian sources:
\begin{thm}\label{thm_common}
The common information of $m$ Gaussian scalar random variables $X_1,\dots,X_m$ with covariance matrix $\bsigma\succ\mb{0}$ is given by
\begin{align}
C(X^m)=
\frac{1}{2}\inf_{\bf\Lambda}\log\frac{|\bsigma|}{|\bf\Lambda|}
\label{e137}
\end{align}
where the infimum is over all diagonal matrices $\bf \Lambda$ satisfying $\bf \Lambda\preceq\Sigma$.
More generally when $\bsigma$ is not necessarily invertible, then $\bsigma$ and $\bf \Lambda$ in \eqref{e137} should be replaced by the restrictions $\bsigma|_V$ and ${\bf \Lambda}|_V$ to the column space $V$ of $\bsigma$.
\end{thm}
\begin{rem}
After we completed a draft of this paper,
Jun Chen and Chandra Nair independently found and showed us a simple proof of Theorem~\ref{thm_common} without invoking Theorem~\ref{thm1}:
using the fact that Gaussian distribution maximizes differential entropy given a covariance constraint and the concavity of the log-determinant function,
\begin{align}
I(U;X^m)
&=h(X^m)-h(X^m|U)
\\
&\ge\frac{1}{2}\log|\bsigma|
-\mathbb{E}\left[\frac{1}{2}\log|\Cov(\mb{X}|U)|\right]
\\
&\ge \frac{1}{2}\log|\bsigma|
-\frac{1}{2}\log|\bsigma_{\mb{X}|U}|.
\end{align}
Since $\bsigma_{\mb{X}|U}\preceq \bsigma$ and $\bsigma_{\mb{X}|U}$ is diagonal, we establishes the nontrivial ``$\ge$'' part of \eqref{e137} in the case of invertible $\bsigma$.
This argument is also related to the proof of Theorem~\ref{thm1} in the sense that both convert a mutual information optimization problem (without a covariance constraint) to a conditional differential entropy optimization problem with a covariance constraint.
\end{rem}

\subsection{Key Generation with an Omniscient Helper}\label{sec_key}
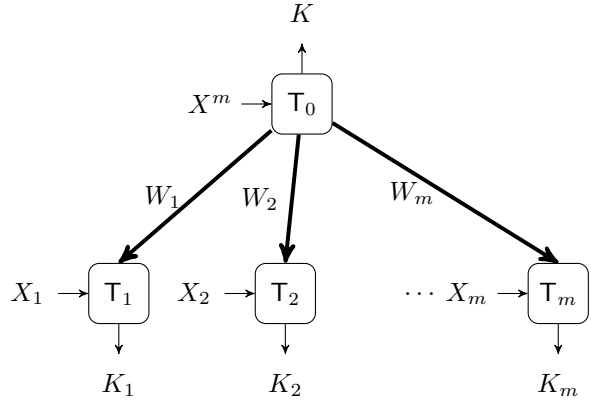
\begin{figure}[h!]
  \centering
\begin{tikzpicture}
[node distance=1cm,minimum height=8mm,minimum width=8mm,arw/.style={->,>=stealth'}]
  \node[rectangle,draw,rounded corners] (A) {${\sf T}_1$};
  \node[rectangle,draw,rounded corners] (B) [right= 1.4cm of A] {${\sf T}_2$};
  \node[rectangle] (C) [right =of B] {$\dots$};
  \node[rectangle,draw,rounded corners] (D) [right =of C] {${\sf T}_m$};
  \node[rectangle,draw,rounded corners] (T) [above right=of B, xshift=-13mm, yshift=10mm] {${\sf T}_0$};
  \node[rectangle] (Z) [left=0.4cm of T] {$X^m$};
  \node[rectangle] (K1) [below =0.4cm of A] {$K_1$};
  \node[rectangle] (K2) [below =0.4cm of B] {$K_2$};
  \node[rectangle] (Km) [below =0.4cm of D] {$K_m$};
  \node[rectangle] (K) [above =0.4cm of T] {$K$};
  \node[rectangle] (X1) [left =0.4cm of A] {$X_1$};
  \node[rectangle] (X2) [left =0.4cm of B] {$X_2$};
  \node[rectangle] (Xm) [left =0.4cm of D] {$X_m$};
 \draw[arw] (Z) to node[]{} (T);
 \draw[arw] (X1) to node[]{} (A);
  \draw[arw] (X2) to node[]{} (B);
   \draw[arw] (Xm) to node[]{} (D);
  \draw [arw] (A) to node[midway,above]{} (K1);
  \draw [arw] (B) to node[midway,above]{} (K2);
  \draw [arw] (D) to node[midway,above]{} (Km);
  \draw [arw] (T) to node[]{} (K);
  \draw [arw,line width=1.5pt] (T) to node[midway,left]{$W_1$} (A.north);
  \draw [arw,line width=1.5pt] (T) to node[midway,left]{$W_2$} (B.north);
  \draw [arw,line width=1.5pt] (T) to node[midway,left]{$W_m$} (D.north);
\end{tikzpicture}
\caption{Key generation with an omniscient helper}
\label{f_1com}
\end{figure}
As an example of applications in the network information theory, we give a simple characterization of the achievable rate region for secret key generation with an omniscient helper \cite{liu2015key} in the case of stationary memoryless Gaussian sources.
Let $Q_{X^m}$ be the per-letter joint distribution of sources $X_1,\dots,X_m$.
As in Figure~\ref{f_1com}, the Terminals ${\sf T}_1,\dots,{\sf T}_m$ observe i.i.d.~realizations of $X_1,\dots,X_m$, respectively, whereas the omniscient helper ${\sf T}_0$ has access to all the sources.
Suppose the terminals perform block coding with length $n$.
The communicator computes the integers $W_1((X^m)^n),\dots,W_m((X^m)^n)$ possibly stochastically and sends them to ${\sf T}_1,\dots,{\sf T}_m$, respectively. Then, all the terminals calculate integers $K((X^m)^n),K_1((X_1)^n,W_1),\dots,K_m((X_m)^n,W_m)$ possibly stochastically.
The goal is to make $K=K_1=\dots=K_m$ with high probability and $K$ almost equiprobable and independent of \emph{each} message $W_l$. In other words, we want to minimize the following quantities:
\begin{align}
\epsilon_n&=\max_{1\le l\le m}\mathbb{P}[K\neq K_l],\label{e_m1}
\\
\nu_n&=\max_{1\le l\le m}\{\log|\mathcal{K}|-H(K|W_l)\}.\label{e_m2}
\end{align}

An $(m+1)$-tuple $(R,R_1,\dots,R_m)$ is said to be \emph{achievable} if a sequence of key generation schemes can be designed to fulfill the following conditions:
\begin{align}
\liminf_{n\to\infty}\frac{1}{n}\log|\mathcal{K}|&\ge R;
\\
\limsup_{n\to\infty}\frac{1}{n}\log|\mathcal{W}_l|&\le R_l,\quad l\in\{1,\dots,m\};
\\
\lim_{n\to\infty}\epsilon_n&=0;\\
\lim_{n\to\infty}\nu_n &=0.
\label{e_145}
\end{align}
Notice that a small $\nu_n$ does not imply that $K$ is nearly independent with \emph{all} the messages $W^m$; the problem appears to be harder to solve if a the stronger requirement that
\begin{align}
\lim_{n\to\infty}(\log|\mathcal{K}|-H(K|W^m))=0
\end{align}
is imposed in place of \eqref{e_145} \cite{liu2015key}.

\begin{thm}\cite{liu2015key}\label{thm_2}
The set of achievable rates is the closure of
\begin{align}
\bigcup_{Q_{{ U|X}^m}}
\left\{
\begin{array}{c}
  (R,R_1,\dots,R_m):  \\
  R\le \min\{I({{ U;X}^m}),H({ X}_1),\dots,H({ X}_m)\}; \\
  R_l\ge I({ U};{ X}^m)-I({ U};{ X}_l),\quad 1\le l\le m
\end{array}\right\}.
\label{e_om}
\end{align}
\end{thm}
A priori, computing the rate region from \eqref{e_om} requires solving an optimization with possibly infinite dimensions. However, using Theorem~\ref{thm1} we easily see that the problem can be reduced to a matrix optimization in the case of Gaussian sources:
\begin{thm}\label{thm23}
When $(X_1,\dots,X_m)$ is jointly Gaussian with non-degenerate covariance matrix $\bsigma$, the achievable region can be represented as the closure of
\begin{align}
\bigcup_{\mb{0}\preceq \bsigma'\preceq \bsigma}
\left\{
\begin{array}{c}
  (R,R_1,\dots,R_m):  \\
  R\le \half\log\frac{|\bsigma|}{|\bsigma'|}; \\
  R_l\ge \half\log\frac{|\bsigma|}{|\bsigma'|}
  -\half \log\frac{\Sigma_{ll}}{\Sigma'_{ll}},\quad 1\le l\le m
\end{array}
\right\}.
\label{e145}
\end{align}
\end{thm}
A related but simpler problem is the common randomness (CR) generation problem, where there is no secrecy (independence) assumption imposed. The CR generation counterpart of Theorem~\ref{thm_2} is derived in \cite[Theorem~4.2]{ahlswede1998common}, which can be expressed by replacing the first bound in \eqref{e_om} with
\begin{align}
R\le I({{ U;X}^m}).
\end{align}
Despite the superficial similarity of the two regions, the achievability part of Theorem~\ref{thm_2} requires more sophisticated coding technique to guarantee secrecy. We observe that Theorem~\ref{thm23} also applies to CR generation from Gaussian sources, because in that case $H(X_j)=\infty$ and hence does not effectively change the bound on $R$.

A generalization of the omniscient helper problem is called the one communicator problem in \cite{liu2015key}, where in Figure~\ref{f_1com} the terminal ${\sf T}_0$ does not see all the random variables $X^m$ but instead another random variable $Z$ which can be arbitrarily correlated with $X^m$. The achievable rate region for one communicator CR generation is known \cite[Theorem~4.2]{ahlswede1998common} to be the closure of
\begin{align}
\bigcup_{Q_{U|Z}}
\left\{
\begin{array}{c}
  (R,R_1,\dots,R_m):  \\
  R\le I(U;Z); \\
  R_l\ge I(U;Z)-I(U;X_l),\quad 1\le l\le m
\end{array}\right\}
\label{e_CRonecom}
\end{align}
and obviously we can use Theorem~\ref{thm1} to reduce \eqref{e_CRonecom} to a matrix optimization problem in the case of Gaussian $(Z,X_1,\dots,X_m)$. The achievable rate region is also known for \emph{key} generation with one communicator \cite{liu2015key}. But the expression of that region is more complicated involving $m+1$ auxiliary random variables, and it is not immediate to conclude that Gaussian auxiliary random variables suffice merely using Theorem~\ref{thm1}.

\section{Discussion}
We have seen that the information-theoretic formulation of the Brascamp-Lieb inequality is often more convenient for proving certain properties, including data processing, tensorization, convexity and the Gaussian optimality for Gaussian distributions and Gaussian random transformations. A point not elaborated in this paper is that, in contrast, the functional formulation of Brascamp-Lieb inequality has the advantage of allowing us to prove strong converses of certain coding theorems, which strengthens the traditional weak converses obtained through manipulations of information-theoretic formulas and Fano's inequality. We examine this complementary viewpoint in \cite{ISIT_lccv_smooth2016} where an idea called smoothing is introduced which is essential for obtaining the strong converse of the full rate region. Other recent applications of related functional inequalities in proving impossibility bounds are summarized in \cite{pw_2015}.

Although Theorem~\ref{thm6} shows the Gaussian extremisability in the non-degenerate case, it is not a necessary condition. Since Gaussian exhaustibility holds in general, in principle verifying extremisability is reduced to a merely finite dimensional optimization problem. However, it is nontrivial to give a ``closed form'' condition for extremisability and finiteness of the optimal value in terms of the structure of linear subspaces involved \cite{bennett2008brascamp}. In \cite{bennett2008brascamp}
such a condition is given for the case of deterministic linear $(Q_{\mb{Y}_j|\bf X})$ (the traditional Brascamp-Lieb inequality).

\section{Acknowledgements}
We thank Sudeep Kamath for many stimulating discussions during the course of this work.
This work was supported in part by NSF Grants CCF-1528132, CCF-0939370 (Center for Science of Information), CCF-1319299, CCF-1319304,
CCF-1350595 and AFOSR FA9550-15-1-0180.

\appendices
\section{A Generalization of Legendre-Fenchel Duality to More than Two functions}\label{app_fr}
The Fenchel-Rockafellar duality (see \cite[Theorem~1.9]{villani2003topics}, or \cite{rockafellar2015convex} in the case of finite dimensional vector spaces) usually refers to the $k=1$ special case of the following result.
\begin{thm}\label{thm_FRduality}
Assume that $A$ is a topological vector space whose dual is $A^*$.
Let $\Theta_j\colon A\to \mathbb{R}\cup\{+\infty\}$, $j=0,1,\dots,k$, for some positive integer $k$. Suppose there exist some $(u_j)_{j=1}^k$ and $u_0:=-(u_1+\dots+u_k)$ such that
\begin{align}
\Theta_j(u_j)<\infty,\quad j=0,\dots,k
\label{e253}
\end{align}
and $\Theta_0$ is upper semicontinuous at $u_0$.
Then
\begin{align}
-\inf_{\ell\in A^*}\left[\sum_{j=0}^k\Theta_j^*(\ell)\right]
=\inf_{u_1,\dots,u_k\in A}
\left[\Theta_0\left(-\sum_{j=1}^ku_j\right)+\sum_{j=1}^k \Theta_j(u_j)\right].
\label{e62}
\end{align}
\end{thm}
For completeness, we provide a proof of this result, which is based on the Hahn-Banach theorem (Theorem~\ref{thm_hb}) and is similar to the proof of \cite[Theorem~1.9]{villani2003topics}.
\begin{proof}
Let $m_0$ be the right side of \eqref{e62}.
The $\le$ part of \eqref{e62} follows trivially from the (weak) min-max inequality since
\begin{align}
m_0&=\inf_{u_0,\dots,u_k\in A}\sup_{\ell\in A^*}
\left\{\sum_{j=0}^k \Theta_j(u_j)
-\ell(\sum_{j=0}^k u_j)
\right\}
\\
&\ge \sup_{\ell\in A^*}\inf_{u_0,\dots,u_k\in A}
\left\{\sum_{j=0}^k \Theta_j(u_j)
-\ell(\sum_{j=0}^k u_j)
\right\}
\\
&=
-\inf_{\ell\in A^*}\left[\sum_{j=0}^k\Theta_j^*(\ell)\right].
\end{align}
It remains to prove the $\ge$ part,
and it suffices to assume without loss of generality that $m_0>-\infty$.
Note that \eqref{e253} also implies that $m_0<+\infty$.
Define convex sets
\begin{align}
C_j&:=\{(u,r)\in A\times \mathbb{R}\colon \,r>\Theta_j(u)\},\quad
j=0,\dots,k;
\\
D&:=\{(0,m)\in A\times \mathbb{R}\colon \,
m\le m_0\}.
\label{e256}
\end{align}
Observe that these are nonempty sets by the assumption \eqref{e253}. Also $C_0$ has nonempty interior by the assumption that $\Theta_0$ is upper semicontinuous at $u_0$. Thus, the Minkowski sum
\begin{align}
C:=C_0+\dots+C_k
\end{align}
is a convex set with a nonempty interior.
Moreover, $C\cup D=\emptyset$.
By the Hahn-Banach theorem (Theorem~\ref{thm_hb}), there exists $(\ell,s)\in A^*\times \mathbb{R}$ such that
\begin{align}
sm\le \ell(\sum_{j=0}^ku_j)+s\sum_{j=0}^k r_j.
\label{e258}
\end{align}
For any $m\le m_0$ and $(u_j,r_j)\in C_j$, $j=0,\dots,k$.
From \eqref{e256} we see \eqref{e258} can only hold when $s\ge0$.
Moreover, from \eqref{e253} and the upper semicontinuity of $\Theta_0$ at $u_0$ we see the $\sum_{j=0}^ku_j$ in \eqref{e258} can take value in a neighbourhood of $0\in A$, hence $s\neq0$.
Thus, by dividing $s$ on both sides of \eqref{e258} and setting $\ell\leftarrow-\ell/s$, we see that
\begin{align}
m_0&\le \inf_{u_0,\dots,u_k\in A}\left[-\ell(\sum_{j=0}^ku_j)+\sum_{j=0}^k \Theta_j(u_j)\right]
\\
&=-\left[\sum_{j=0}^k\Theta_j^*(\ell)\right]
\end{align}
which establishes the $\ge$ part in \eqref{e62}.
\end{proof}

\begin{thm}[Hahn-Banach]\label{thm_hb}
Let $C$ and $D$ be convex, nonempty disjoint subsets of a topological vector space $A$. If the interior of $C$ is non-empty, then there exists $\ell\in A^*$, $\ell\neq0$ such that
\begin{align}
\sup_{u\in D}\ell(u)\le \inf_{u\in C}\ell(u).
\end{align}
\end{thm}
\begin{rem}
The assumption in Theorem~\ref{thm_hb} that $C$ has nonempty interior is only necessary in the infinite dimensional case.
However, even if $A$ in Theorem~\ref{thm_FRduality} is finite dimensional, the assumption in Theorem~\ref{thm_FRduality} that $\Theta_0$ is upper semicontinuous at $u_0$ is still necessary, because this assumption was not only used in applying Hahn-Banach, but also in concluding that $s\neq0$ in \eqref{e258}.
\end{rem}

\section{Existence of Maximizer in Theorem~\ref{thm:GaussEntropy}}\label{app_exist}

\begin{prop}
In the non-degenerate case,
for any $\bsigma\succeq 0$,
\begin{align}
\phi(\bsigma):=\sup_{P_{U\mb{X}}\colon\Sigma_{\mb{X}|U}\preceq \mb{\Sigma}}F(P_{U\mb{X}})
\label{e_obj}
\end{align}
is finite and is attained by some $P_{U\mb{X}}$ with $|\mathcal{U}|<\infty$.
\end{prop}
\begin{proof}
First, observe that if we let $\tilde{\phi}(\cdot)$ be the supremum in \eqref{e_obj} with the additional restriction that $|\mathcal{U}|<\infty$,
then $\tilde{\phi}(\cdot)$ is a concave function on a convex set of \emph{finite} dimension. Hence Jensen's inequality\footnote{Luckily, $\tilde{\phi}$ is defined on a finite dimensional set of matrices (rather than a possibly infinite dimensional set of distributions $P_{\mb{X}}$). In the infinite dimensional case without further continuity assumptions, Jensen's inequality can fail; see the example in \cite[equation (1.3)]{perlman1974}.\label{f18}} implies that $\phi(\cdot)\le \tilde{\phi}(\cdot)$, while $\tilde{\phi}(\cdot)\le\phi(\cdot)$ is obvious from the definition. Thus $\tilde{\phi}(\cdot)=\phi(\cdot)$.

The set
\begin{align}
\mathcal{C}:=\bigcup_{P_{\mb{X}}}
\{(F_0(P_{\mb{X}}),\Cov(\mb{X}))\}
\end{align}
lies in a linear space of dimension $1+\frac{\dim(\mathcal{X})(\dim(\mathcal{X})+1)}{2}$. By Carath\'{e}odory's theorem \cite[Theorem~17.1]{rockafellar2015convex}, each point in the convex hull of $\mathcal{C}$ is a convex combination of at most $D:=2+\frac{\dim(\mathcal{X})(\dim(\mathcal{X})+1)}{2}$ points in $\mathcal{C}$:
\begin{align}
\bigcup_{P_{U\mb{X}}\colon \mathcal{U}\textrm{ finite} }\{(F(P_{U\mb{X}}),\Sigma_{\mb{X}|U})\}
=\bigcup_{P_{U\mb{X}}\colon|\mathcal{U}|\le D }\{(F(P_{U\mb{X}}),\Sigma_{\mb{X}|U})\}
\end{align}
hence
\begin{align}\label{e_inf}
\phi(\bsigma)
&=\sup_{P_{U\mb{X}}\colon\mathcal{U}= \{1,\dots,D\},\Sigma_{\mb{X}|U}\preceq \mb{\Sigma}}F(P_{U\mb{X}}).
\end{align}
Now suppose $\{P_{U_n\mb{X}_n}\}_{n\ge 1}$ is a sequence
satisfying $\bsigma_{\mb{X}_n|U_n}\preceq \mb{\Sigma}$, $|\mathcal{U}_n|=\{1,\dots,D\}$ for each $n$, and
\begin{align}
\lim_{n\to\infty}F(P_{U_n\mb{X}_n})=\eqref{e_inf}.
\end{align}
We can assume without loss of generality that $P_{U_n}$ converges to some $P_{U^{\star}}$, since otherwise we can pass to one convergent subsequence instead. Moreover, by the translation invariance we can assume without loss of generality that
\begin{align}
\EE[\mb{X}_n|U_n=u]=0
\end{align}
for each $u$ and $n$.

If $u\in\{1,\dots,D\}$ is such that $P_{U^{\star}}(u)>0$, then for $n$ sufficiently large, we have $P_{U_n}>\frac{P_{U^{\star}}(u)}{2}$ and
\begin{align}
\Cov(\mb{X}_n|U_n=u)\preceq \frac{2\mb{\Sigma}}{P_{U^{\star}}(u)}.
\label{e_109}
\end{align}
Thus $\{P_{\mb{X}_n|U_n=u}\}_{n\ge 1}$ is a tight sequence of measures by Chebyshev's inequality, and Prokhorov's theorem \cite{prokhorov1956convergence} guarantees the existence of a subsequence of $\{P_{\mb{X}_n|U_n=u}\}_{n\ge 1}$ converging weakly to some Borel measure $P_{\mb{X}^{\star}_u}$. We might as well assume that $\{P_{\mb{X}_n|U_n=u}\}_{n\ge 1}$ converges to $P_{\mb{X}^{\star}_u}$ since otherwise we pass to a convergent subsequence instead. This argument can applied to each $u\in\{1,\dots,D\}$ satisfying $P_{U^{\star}}(u)>0$ iteratively, hence we can assume the existence of the weak limits
\begin{align}\label{e_wconv}
\lim_{n\to\infty}P_{\mb{X}_n|U_n=u}=P_{\mb{X}^{\star}_u}
\end{align}
for all such $u$. Next, we show that
\begin{align}
\limsup_{n\to\infty}F_0(P_{\mb{X}_n|U_n=u})\le F_0(P_{\mb{X}^{\star}_u})
\label{e_s1}
\end{align}
for all such $u$. Using Lemma~\ref{lem28} below, we obtain
\begin{align}
\limsup_{n\to\infty} h(\mb{X}_n|U_n=u)
\le h(\mb{X}^{\star}_u).
\label{e_c1}
\end{align}
Because of the moment constraint \eqref{e_109}, the differential entropy of the output distribution, which is smoothed by the Gaussian kernel, enjoys weak continuity in the input distribution (see e.g.~\cite[Proposition~18]{geng2014yanlin}, \cite[Theorem~7]{wu2012functional}, or \cite[Theorem~1, Theorem~2]{godavarti2004convergence}):
\begin{align}
\lim_{n\to\infty} h(\mb{Y}_{jn}|U_n=u)
= h(\mb{Y}^{\star}_j|U^{\star}=u)
\quad
\textrm{for each $u\in\{1,\dots,D\}$}
\label{e_c2}
\end{align}
where $(U,\mb{X}_n,\mb{Y}_{jn})\sim P_{\mb{X}_nU_n}Q_{\mb{Y}_j|\mb{X}}$ and $(U^{\star},\mb{X}^{\star},\mb{Y}^{\star}_j)\sim P_{\mb{X}^{\star}U^{\star}}Q_{\mb{Y}_j|\mb{X}}$.
As for the trace term, consider
\begin{align}
\liminf_{n\to\infty}\Tr[\mb{M}\Cov(\mb{X}_n|U_n=u)]
&=\liminf_{n\to\infty}\EE[\Tr[\mb{M}\mb{X}_n
\mb{X}_n^{\top}]|U_n=u]
\\
&=\liminf_{n\to\infty}\sup_{K>0} \EE[\Tr[\mb{M}\mb{X}_n\mb{X}_n^{\top}]\wedge K|U_n=u]
\label{e_mono}
\\
&\ge\sup_{K>0}\liminf_{n\to\infty}
\EE[\Tr[\mb{M}\mb{X}_n\mb{X}_n^{\top}]\wedge K|U_n=u]
\\
&\ge\sup_{K>0}\EE[\Tr[\mb{M}\mb{X}_u^{\star}\mb{X}_u^{{\star}\top}]
\wedge K]
\label{e_wconv1}
\\
&=\EE[\Tr[\mb{M}\mb{X}_u^{\star}\mb{X}_u^{{\star}\top}]]
\label{e_mono2}
\end{align}
where ``$\wedge$'' takes the minimum of two numbers, \eqref{e_mono} and \eqref{e_mono2} are from monotone convergence theorem, and \eqref{e_wconv1} uses the weak convergence \eqref{e_wconv}. The proof of \eqref{e_s1} is finished by combining  \eqref{e_c1} \eqref{e_c2} and \eqref{e_mono2}.

The final step deals with any $u\in\{1,\dots,D\}$ satisfying $P_{U^{\star}}(u)=0$. The variance constraint implies that
\begin{align}
\Cov(\mb{X}_n|U_n=u)\preceq \frac{1}{P_{U_n}(u)}\bsigma,
\end{align}
hence by the fact that Gaussian distribution maximizes the differential entropy under a covariance constraint, we have the bound
\begin{align}
h(\mb{X}_n|U_n=u)
&\le \frac{\dim(\mathcal{X})}{2}\log(2\pi)
+\frac{1}{2}\log e
+\frac{1}{2}
\log\left|
\frac{1}{P_{U_n}(u)}\bsigma\right|
\\
&=\frac{\dim(\mathcal{X})}{2}\log(2\pi)
+\frac{1}{2}\log e
+\frac{1}{2}\log\left|\bsigma\right|
+\frac{\dim(\mathcal{X})}{2}\log \frac{1}{P_{U_n}(u)}.
\end{align}
This combined with the fact that $h({\bf Y}_{jn}|U=u)
\ge h({\bf Y}_{jn}|{\bf X}_n)$ is bounded below, $j=1,\dots,m$ in the non-degenerate case, implies that if $P_{U_n}(u)$ converges to zero, then
\begin{align}
\limsup_{n\to\infty}P_{U_n}(u)F_0(P_{\mb{X}_n|U_n=u})\le0.
\label{e_s2}
\end{align}
Combining \eqref{e_s1} and \eqref{e_s2}, we see
\begin{align}
F(P_{U^{\star}\mb{X}^{\star}})\ge \limsup_{n\to\infty}F(P_{U_n\mb{X}_n}),
\end{align}
where $P_{\mb{X}^{\star}|U^{\star}=u}:=P_{\mb{X}^{\star}_u}$ for each $u=1,\dots,D$.
\end{proof}
\begin{lem}\label{lem28}
Suppose $(P_{{\bf X}_n})$ is a sequence of distributions on $\mathbb{R}^d$ converging weakly to $P_{{\bf X}^{\star}}$, and
\begin{align}
\mathbb{E}[{\bf X}_n{\bf X}_n^{\top}]\preceq \bsigma
\label{e175}
\end{align}
for all $n$. Then
\begin{align}
\limsup_{n\to\infty} h({\bf X}_n) \le h({\bf X}^{\star}).
\end{align}
\end{lem}
\begin{rem}
The result fails without the condition \eqref{e175}. Also, related results when the weak convergence is replaced with pointwise convergence of density functions and certain additional constraints was shown in \cite[Theorem~1, Theorem~2]{godavarti2004convergence} (see also the proof of \cite[Theorem~5]{geng2014yanlin}). Those results are not applicable here since the density functions of ${\bf X}_n$ do not converge pointwise. They are applicable for the problems discussed in \cite{geng2014yanlin} because the density functions of the output of the Gaussian random transformation enjoy many nice properties due to the smoothing effect of the ``good kernel''.
\end{rem}
\begin{proof}
It is well known that in metric spaces and for probability measures, the relative entropy is weakly lower semicontinuous (cf.~\cite{verdubook}). This fact and a scaling argument immediately show that, for any $r>0$,
\begin{align}
\limsup_{n\to\infty} h({\bf X}_n|\|\mb{X}_n\|\le r) \le h({\bf X}^{\star}|\|\mb{X}^{\star}\|\le r).
\label{e177}
\end{align}
Let $p_n(r):=\mathbb{P}[\|\mb{X}_n\|> r]$, then \eqref{e175} implies
\begin{align}
\mathbb{E}[\mb{X}\mb{X}^{\top}|\|\mb{X}_n\|>r]\le \frac{1}{p_n(r)}\bsigma.
\end{align}
Therefore, since the Gaussian distribution maximizes differential entropy given a second moment upper bound, we have
\begin{align}
h(\mb{X}_n|\|\mb{X}_n\|>r)\le
\half \log\frac{(2\pi)^d e|\bsigma|}{p_n(r)}.
\end{align}
Since $\lim_{r\to\infty}\sup_n p_n(r)=0$ by \eqref{e175} and Chebyshev's inequality, the above implies that
\begin{align}
\lim_{r\to\infty} \sup_n p_n(r)h(\mb{X}_n|\|\mb{X}_n\|>r)=0.
\label{e180}
\end{align}
The desired result follows from \eqref{e177}, \eqref{e180} and the fact that
\begin{align}
h(\mb{X}_n)=p_n(r)h(\mb{X}_n|\|\mb{X}_n\|>r)
+(1-p_n(r))h(\mb{X}_n|\|\mb{X}_n\|\le r)+h(p_n(r)).
\end{align}
\end{proof}

\bibliographystyle{ieeetrans}
\bibliography{ref_om}
\end{document}